\documentclass[a4paper,onecolumn,11pt,accepted=2022-12-04]{quantumarticle}
\pdfoutput=1
\usepackage[utf8]{inputenc}
\usepackage[english]{babel}
\usepackage[T1]{fontenc}

\usepackage{amsmath}
\usepackage{hyperref}
\usepackage{amssymb}

\usepackage{multirow}
\usepackage{tablefootnote}
\usepackage{xcolor,array,bm}

\usepackage{float}
\usepackage{tikz}
\usepackage{lipsum}
\usepackage{amsthm}
\usepackage{hyperref,cleveref}
\usepackage{braket}
\usepackage{tikz-cd}
\usepackage[sorting=none, backend = bibtex, style=numeric-comp]{biblatex}
\newtheorem{theorem}{Theorem}[section]
\newtheorem{lemma}[theorem]{Lemma}

\newtheorem{definition}[theorem]{Definition}
\newtheorem{prop}[theorem]{Proposition}
\newtheorem{remark}[theorem]{Remark}
\newcommand{\swap}[1]{\operatorname{SWAP} #1}
\newcommand{\qn}[1]{\textcolor{black}{#1}}

\begin{document}

\title{Block-encoding dense and full-rank kernels using hierarchical matrices: applications in quantum numerical linear algebra}

\author{Quynh T. Nguyen}
\affiliation{Department of Electrical Engineering and Computer Science, Massachusetts Institute of Technology, USA}
\affiliation{Department of Physics, Massachusetts Institute of Technology, USA}
% \email{latex@quantum-journal.org}
% \homepage{http://quantum-journal.org}
% \orcid{0000-0003-0290-4698}
% \thanks{You can use the \texttt{\textbackslash{}email}, \texttt{\textbackslash{}homepage}, and \texttt{\textbackslash{}thanks} commands to add additional information for the preceding \texttt{\textbackslash{}author}. If applicable, this can also be used to indicate that a work has previously been published in conference proceedings.}
\author{Bobak T. Kiani}
\affiliation{Department of Electrical Engineering and Computer Science, Massachusetts Institute of Technology, USA}
\affiliation{Research Laboratory of Electronics, Massachusetts Institute of Technology, USA}
\author{Seth Lloyd}
\affiliation{Research Laboratory of Electronics, Massachusetts Institute of Technology, USA}
\affiliation{Department of Mechanical Engineering, Massachusetts Institute of Technology, USA}
\affiliation{Turing Inc., Cambridge, MA, USA}

\maketitle

\begin{abstract}
  Many quantum algorithms for numerical linear algebra assume black-box access to a block-encoding of the matrix of interest, which is a strong assumption when the matrix is not sparse. Kernel matrices, which arise from discretizing a kernel function $k(x,x')$, have a variety of applications in mathematics and engineering. They are generally dense and full-rank. Classically, the celebrated fast multipole method performs matrix multiplication on kernel matrices of dimension $N$ in time almost linear in $N$ by using the linear algebraic framework of hierarchical matrices. In light of this success, we propose a block-encoding scheme of the hierarchical matrix structure on a quantum computer. When applied to many physical kernel matrices, our method can improve the runtime of solving quantum linear systems of dimension $N$ to $O(\kappa \operatorname{polylog}(\frac{N}{\varepsilon}))$, where $\kappa$ and $\varepsilon$ are the condition number and error bound of the matrix operation. This runtime is near-optimal and, in terms of $N$, exponentially improves over prior quantum linear systems algorithms in the case of dense and full-rank kernel matrices. We discuss possible applications of our methodology in solving integral equations and accelerating computations in N-body problems.
\end{abstract}

\section{Introduction}
One of the most promising applications of quantum information processing is in solving linear algebraic problem in high dimensions efficiently. Since the Harrow-Hassidim-Lloyd (HHL) algorithm was first developed for solving linear systems of equations \cite{hhl}, a number of improvements have been proposed. Significantly, the breakthrough algorithms in \cite{childshhl} and \cite{gilyen2019quantum} solved sparse linear systems of equations in time polylogarithmic in the system size $N$ and the target error bound $\varepsilon$.
% and nearly optimal dependence on the target tolerance $\varepsilon$ and condition number $\kappa$. 
When these algorithms are applied to dense matrices, however, their complexity scales linearly in $N$ (see \Cref{tab:compare}). Using a quantum random access memory (QRAM) data structure \cite{kerenidis2016quantum}, the work of \cite{wossnig2017} improved the runtime to $\widetilde{O}(\sqrt{N})$. Classically, quantum-inspired  algorithms \cite{gilyen2022improved, shao2021faster} use an analog of QRAM called sample-query access and apply sketching techniques \cite{sketching2014} to solve linear systems with runtimes independent of $N$ but scaling poorly with the rank $k$ of the matrix as $\widetilde{O}(k^6)$ (see \Cref{app:compare}).

In fact, achieving a runtime logarithmic in $N$ for dense and full-rank matrices is generally challenging unless there are specific symmetries or structures inherent in the matrix. For example, quantum algorithms have been developed to achieve polylogarithmic complexity in $N$ for Toepliz systems \cite{toepliz}, Hankel matrices \cite{hankel}, and linear group convolutions \cite{groupconv} which all feature the group quantum Fourier transforms (QFT) \cite{algebraic} to implement operations as sparse matrices in the Fourier basis. This approach of using group QFTs, however, only applies for certain groups and requires the matrix entries to be \emph{exactly} uniformly sampled from a generating function, thus inapplicable to more general settings (e.g., matrix entries corresponding to atoms on a lattice vibrating around their equilibrium positions).

% can solve a general sparse system with nearly optimal dependence on  via using series expansions quantum techniques such as linear combination of unitaries (LCU) \cite{kothari2014efficient} and variable-time amplitude amplification (VTAA) \cite{ambainis2010variable}. More recently, ref. \cite{gilyen2019quantum} introduces the breakthrough quantum singular value transform (QSVT), building on their earlier works in block-encoding and quantum signal processing \cite{Low_2019, power-block, qsp}, which

Kernel matrices are a class of matrices whose entries are calculated as a function of a kernel $k(x,x')$. They have applications in mathematical physics \cite{fmm_shortcourse}, engineering \cite{atkinson2009numerical} and machine learning \cite{Rasmussen2003gp}, where they are often obtained by sampling or discretizing a smooth kernel $k(x,x')$ that reflects the physical model at hand. In general, kernel matrices are dense and full-rank. Hierarchical matrices ($\mathcal{H}$-matrices) \cite{part1hmatrices, part2hmatrices} are a classical framework for efficiently performing matrix operations on commonly found kernel matrices by hierarchically splitting the ``space" dimension of the matrix entries. In terms of the matrix dimension $N$, classically performing matrix-vector multiplication with $\mathcal{H}$-matrices takes time $O(N \operatorname{polylog} N)$ which is a significant improvement over the generic $O(N^2)$ time. This $\mathcal{H}$-matrix structure lies at the heart of the celebrated fast multipole method \cite{fmm_shortcourse}.
% \bk{needs citation, I would also add a sentence about FMM and how it can improve to O(N)} \qn{The citations are immediately above. Also, FMM's runtime is $\sim N \log \frac{1}{\varepsilon} $, but the error $\varepsilon$ there is for each individual sum, whereas in our case this error would need to be $\varepsilon/N$.} 
This improvement intuitively arises because kernel matrices often decay or change slowly along entries that are farther from the diagonals of the matrix. $\mathcal{H}$-matrices progressively approximate the matrix in larger low-rank blocks for farther off-diagonal entries to efficiently implement kernel matrices.

% \qn{This $\mathcal{H}$-matrix structure allows us to directly exploit the magnitudes of the entries in a kernel matrix in the quantum setting. Thus, our quantum algorithm, unlike prior works \cite{toepliz, hankel, groupconv}, does not employ the quantum Fourier transforms.} \bk{I would remove these last two sentences as you've already made this point.}

We present an efficient quantum algorithm to implement $\mathcal{H}$-matrices in the block-encoding framework, which embeds matrices inside a unitary operator that can be queried multiple times. This construction, which is at the core of many quantum algorithms, allows one to conveniently perform matrix computations on a  quantum computer \cite{Low_2019, power-block,  qsp, quantumSDP}. As many recently proposed algorithms use block-encodings as a black box, one of the major challenges for quantum numerical linear algebra in this framework is to efficiently block-encode matrices so that the asymptotic runtimes are efficient even when the block-encoding complexity is included. Our $\mathcal{H}$-matrix block-encoding procedure, combined with existing quantum linear systems algorithms, \emph{e.g.,} adiabatic solver in \cite{costa2021optimal}, provides near-optimal runtime in solving quantum linear systems for many typical kernels that are neither sparse nor low-rank, \textit{e.g.,} $k(x,x')=\|x-x'\|^{-1}$ (see \Cref{tab:compare} for comparisons to prior work). Apart from the $\mathcal{H}$-matrix block-encoding, we also provide in the Appendix a list of common kernel matrices that can be efficiently block encoded using existing techniques.

We proceed as follows. First, we overview classical hierarchical matrices in \Cref{sec:hmatrices} and the quantum block-encoding framework in \Cref{sec:blockenc}. We then present an algorithm for efficiently block-encoding hierarchical matrices and apply this to kernel matrices in \Cref{sec:main}. In \Cref{sec:variants}, we generalize our methodology to construct block encodings for a more general class of matrices whose structure is derived from hierarchical matrices. Finally, we discuss applications of our algorithm in calculating $N$-body gravitational potentials and solving integral equations in \Cref{sec:applications}.

\begin{table*}[]
\small
    {\centering
    \renewcommand{\arraystretch}{1.5}
    \begin{tabular}{c|c|c|c|c}
          & \multicolumn{2}{c|}{Complexity}&   \\
          \cline{2-3}
          \textbf{Result} & Forward $\widetilde{O}(\cdot)$ & Inversion $\widetilde{O}(\cdot)$ & Input model & Algorithm \\
         \hline
         \hline
        Classical \cite{ part1hmatrices} & $N $ & $N$ & Classical & Classical $\mathcal{H}$-matrix \\
         \hline
         GST \cite{gilyen2022improved}/SM \cite{shao2021faster}
         & -
         & $\frac{N^6 \kappa^8}{\varepsilon^2}$
         & Sample query & Classical sketching \\
         \hline
         CKS \cite{childshhl}/GSLW  \cite{gilyen2019quantum} & $N \kappa$  &
         $N\kappa $& Oracle access$^a$ & \Cref{lem:inv}$^b$\\ 
        \hline
         KP \cite{kerenidis2016quantum}/WZP \cite{wossnig2017} 
         & $\sqrt{N} \kappa  $ 
         & $\sqrt{N}\kappa   $ 
         & QRAM$^c$ & \Cref{lem:inv} \\
         \hline
         This paper & $ \operatorname{polylog}(N) \kappa$ & $\operatorname{polylog}(N) \kappa$ & Oracle access & \Cref{lem:main}, \Cref{lem:inv} \\ 
    \end{tabular}\\}
  {\raggedleft \footnotesize{$^a$ Naive block-encoding using oracle access via Lemma 48 or 49 in \cite{gilyen2019quantum} (also see \Cref{lem:densenaive})}}\\
{\raggedleft \footnotesize{$^b$ \Cref{lem:inv}} uses the adiabatic solver in \cite{costa2021optimal} which has optimal runtime dependence on $\kappa$}\\
{\raggedleft \footnotesize{$^c$ This input model is strictly more powerful than the oracle access model}}\\

% {\raggedleft 
%  \footnotesize{$^b$ LCU: Linear combination of unitaries \cite{childshhl}, QSVT: Quantum singular value transform \cite{power-block, gilyen2019quantum}}}\\
%      {\raggedleft 
%  \footnotesize{$^c$ QSVE: Quantum singular value estimation \cite{kerenidis2016quantum}}}\\
     \vspace{-0.5cm}
    \caption{Runtimes of previous methods versus our method for matrix multiplication (forward) and solving linear systems (inversion) of dense and full-rank kernel matrices, \textit{e.g.,} $k(x,x')=\|x-x'\|^{-1}$. For purposes of comparison, we assume the operator norm of the matrix is bounded. The $\widetilde{O}$ notation hides terms that grow as $\operatorname{polylog}(\frac{N}{\varepsilon})$. Here, $\kappa$ is the condition number of the matrix operation and $\varepsilon$ is the error bound on the output. For quantum algorithms, we assume the input vector is given as a quantum state and the runtime is equal to the circuit depth times the number of queries needed to obtain the output as a quantum state. We note that classical $\mathcal{H}$-matrix literature and this work focus on kernel matrices, whereas the quantum algorithms we compare to are ``general-purpose''. For these general-purpose algorithms, runtimes are calculated assuming they are applied directly on kernel matrices. More details about prior works can be found in \Cref{app:compare}.}
    \label{tab:compare}
\end{table*}

\section{Background in $\mathcal{H}$-matrices and hierarchical splitting} \label{sec:hmatrices}
To observe why hierarchical matrices are useful, consider a physical model which requires summation or integration over spacetime. The $N$-body gravitational force problem \cite{Barnes1986}, for example, requires evaluating $N$ potentials of the form
\begin{equation}
    \Phi\left(\mathbf{x}_{j}\right)= \sum_{i=1 \atop i \neq j}^{N} \frac{m_{i}}{\|\mathbf{x}_j - \mathbf{x}_i\|^2}, \hspace{.5cm} 1 \leq j \leq N
\end{equation}
where $m_i$ and $\mathbf{x}_i$ denote the mass and location of particle $i$. 
% Of course, we can also consider problems where the sources are continuously distributed: a solution of the heat equation in diffusion problems \cite{fmm_shortcourse} takes the form of
% \begin{equation}
%     u(\mathbf{x}) = \frac{1}{\sqrt{(4\pi T)^{3/2}}} \int e^{- \|\mathbf{x}-\mathbf{y}\|^2/4T} v(\mathbf{y}) d\mathbf{y},
% \end{equation}
% where $v(\mathbf{y})$ is the initial temperature distribution and $u(\mathbf{x})$ is the distribution at time $T$. To numerically evaluate integrals, one  discretizes the integration domain appropriately. Ultimately, 
Summing or integrating over all particles is equivalent to performing matrix-vector multiplication, generally requiring $O(N^2)$ operations unless the underlying matrix is sparse or has special properties.
% (\textit{e.g.,} the discrete Fourier transform matrices only require $O(N \log N)$ operations via the fast Fourier transforms \cite{fft}). 
% In the above setting, such structure is not obviously apparent unless the particles are placed in a specific arrangement (\textit{e.g.,} uniformly spaced).
As we will show, when particles are placed in a favorable arrangement (\textit{e.g.,} approximately uniformly), this matrix-vector multiplication can be approximately performed in time $O(N \operatorname{polylog} N)$.

% Similarly, the problem of kernel matrix inversions arises from integral equations (IEs). For example, consider the following integral equation:
% \begin{equation}
%     g(x) = h(x) f(x) +  \lambda \int_{\Omega} k(x,x') f(x') dx',
% \end{equation}
% where $k(x,x')=k(|x-x'|)$ is a \emph{displacement kernel} and $f(x)$ is the unknown function. In practice, one discretizes the integration domain and solves the resulting linear system via the Cholesky decomposition, which costs time $O(N^3)$. \bk{I would just remove this integral part since this is all illustrative anyway.}

Hierarchical matrices ($\mathcal{H}$-matrices) \cite{part1hmatrices, part2hmatrices,borm2003} are a powerful algebraic framework for accelerating matrix computations involving displacement kernels $k(\mathbf{x},\mathbf{x}')=k(\|\mathbf{x}- \mathbf{x}'\|)$ which satisfy certain properties outlined below. The key idea is to (i) split the kernel matrix into hierarchically off-diagonal blocks and (ii) approximate each block by a low-rank matrix. This low rank approximation is justified because off-diagonal blocks represent distant interactions that are small or slowly changing with distance. %Off-diagonal blocks can be approximated by low-rank matrices since the entries in these blocks, which represent distant interactions, are often small or slowly changing with distance.Intuitively, ``more" approximation can be performed farther from the diagonal, enabling matrix-vector multiplication and matrix inversion in time $O(N \operatorname{polylog} N)$ \cite{part1hmatrices}. 
We now detail the two ingredients of $\mathcal{H}$-matrices, namely the hierarchical splitting and the off-diagonal low-rank approximation.

\begin{figure*}[]
    \centering
    \includegraphics[width=0.99\textwidth]{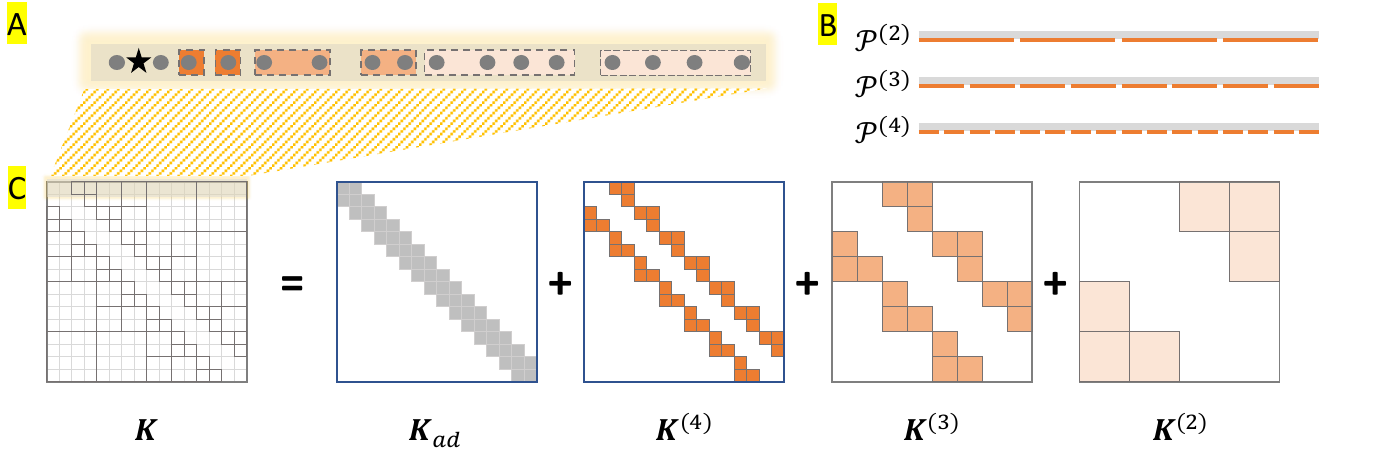}
    \caption{(A) Structure of a hierarchical splitting for a 1D array of 16 particles (gray circles), which are close to uniformly distributed. The accumulative potential at the destination location (star, which can also be one of the sources) is decomposed into contributions from clusters of source particles whose sizes grow with the distance to the destination. (B) Cluster size decreases exponentially with the level of the hierarchical splitting. (C) The associated kernel matrix of this 16-particle system takes the form of an $\mathcal{H}$-matrix, assuming that the destination locations are the same as the source locations. This matrix is decomposed into hierarchically structured blocks using the hierarchical splitting in (A), where each block is numerically low-rank.}
    \label{fig:splitting}
\end{figure*}

\subsection{Hierarchical splitting into admissible blocks}
Hierarchical splitting partitions a kernel matrix into so-called ``admissible'' blocks where entries representing more distant interactions are grouped together into larger blocks. Consider a system of $N$ particles, whose pairwise interaction is described by the kernel $k(x,x')$. The kernel matrix of this system is defined as
\begin{equation}
    \mathbf{K}  = (k(x_i, x_j))_{i,j=0}^{N-1},
\end{equation}
where $x_i$ is the location of particle $i$.
For ease of analysis, we assume the particles are uniformly distributed and choose $x_i = i/N,  \text{ for any } i \in \{0,1,\hdots,N-1\}$ and $N=2^L$. For example, this is the case when the kernel matrix is obtained by discretizing an integral operator over the domain $[0,1]$ (see \cite{part1hmatrices}). In general, a kernel matrix need not be square and the source locations $x_i$ need not be 1-dimensional nor equally spaced and we will show how to generalize the current setup to these cases in later sections. We refer the interested readers to \cite{Fenn2002FMMAH, h2matrices, part2hmatrices, adaptiveHmatrix, Krishnan1995APA} for further details.

% \bk{I would first say what a cluster is and then go into admissible blocks.}

A set of $m$ particles located at $\{x_{j_1},\hdots, x_{j_m}\}$ is denoted as a \emph{cluster} $\sigma$, storing the list of individual indices $\sigma = \{j_1,\hdots, j_m\}$.

\begin{definition}[Admissible blocks] Let $\mathbf{K}  = (k(x_i, x_j))_{i,j=0}^{N-1}$ and $\mathcal{N}=\{0,1,\hdots,N-1\}$. Let $\mathcal{P}(\mathcal{N})$ be a partition of $\mathcal{N}$, i.e., $\mathcal{N}=\bigcup_{\sigma \in \mathcal{P}(\mathcal{N})}\sigma$. For a cluster $\sigma \in  \mathcal{P}(\mathcal{N})$, define the center of $\sigma$ as $c_{\sigma} = \operatorname{mean}\{x_i| i\in \sigma\}$ and the radius of $\sigma$ as $r_{\sigma}= \max_{i \in \sigma} |x_i - c_{\sigma}  |.$
Furthermore, for $\sigma, \rho \in \mathcal{P}(\mathcal{N})$, define the distance between the two clusters $\sigma$ and $\rho$ as $
    \operatorname{dist}(\sigma, \rho) = \min_{i\in \sigma,j \in \rho} |x_i -x_j |.
$
Then, the matrix block $\mathbf{K}^{\sigma, \rho}= (k(x_i,x_j))_{i\in \sigma, j\in \rho}$ is called admissible if $
\eta \cdot \max \{ r_{\sigma}, r_{\rho}\} \leq \operatorname{dist}(\sigma, \rho)$, where $\eta=2$.
\label{def:admissible}
\end{definition}

% Note that an admissible block is not necessarily `rigid' block by this definition. As we have sorted $x_i$ ($x_i=i/N$) in our setup, however, our admissible blocks are indeed blocks in the normal sense. 

Informally, a block is admissible if the distance between its two corresponding clusters is larger than their radii (other values of $\eta$ can give rise to different forms of the hierarchical splitting, see \cite{Fenn2002FMMAH,part1hmatrices}). We now split the matrix $\mathbf{K}$ into a hierarchy of admissible blocks consisting of $L  = \log_2 N$ levels of decreasing radii. At level $\ell$, define the following partition of $\mathcal{N}=\{0,1,\hdots,N-1\}$: 
\begin{equation}    \mathcal{P}^{(\ell)}=\{\sigma^{(\ell)}_I|I \in \{0,\hdots,2^{\ell}-1\}\},
    \label{eq:splittings}
\end{equation}
where
\begin{equation}
    \sigma^{(\ell)}_I=\{i \in \mathcal{N}|x_i \in [I/2^{\ell}, (I+1)/2^{\ell})\}.
\end{equation}
For example, $\mathcal{P}^{(0)}=\{\mathcal{N}\}$, $\mathcal{P}^{(L)}=\{\{i\}| i \in \mathcal{N}\}$, and $\mathcal{P}^{(\ell)}$ for $\ell \in \{2,3,4\}$ are shown in \autoref{fig:splitting}B. Note that admissible blocks only exist for $\ell \geq 2$. We now construct a hierarchical splitting of our kernel matrix $\mathbf{K}$ into admissible blocks starting from the coarsest level $\ell = 2$. First, we split $\mathbf{K}$ into admissible and inadmissible blocks in $\mathcal{P}^{(2)}$
\begin{equation*}
    \mathbf{K}= \mathbf{K}^{(2)}+\mathbf{K}^{(2)}_{inadmissible}.
\end{equation*}
Next, we split $\mathbf{K}^{(2)}_{inadmissible}$ into admissible and inadmissible blocks in $\mathcal{P}^{(3)}$ (excluding the admissible blocks of $\mathcal{P}^{(3)}$ that are already covered in $\mathbf{K}^{(2)}$)
\begin{equation*}
    \mathbf{K}^{(2)}_{inadmissible} =  \mathbf{K}^{(3)}+\mathbf{K}^{(3)}_{inadmissible}.
\end{equation*}
Recursively applying the decomposition $\mathbf{K}^{(\ell-1)}_{inadmissible} =  \mathbf{K}^{(\ell)}+\mathbf{K}^{(\ell)}_{inadmissible}$ up to level $L$, we obtain a hierarchically structured representation of $\mathbf{K}$
\begin{equation}
    \mathbf{K} = \sum_{\ell=2}^{L} \mathbf{K}^{(\ell)}+\mathbf{K}_{ad},
\end{equation}
where $\mathbf{K}_{ad}=\mathbf{K}^{(L)}_{inadmissible}$ is a tridiagonal matrix which represents the ``adjacent" interaction of neighboring particles (see \autoref{fig:splitting}C). Observe that the matrices $\mathbf{K}^{(\ell)}$ are block-sparse; that is, each block-row or block-column of $\mathbf{K}^{(\ell)}$ contains at most three non-zero admissible blocks. Furthermore, these blocks represent ``far-field'' interactions which can be approximated by low-rank matrices if the kernel $k(x,x')$ satisfies some additional properties discussed in the next section.

In \Cref{app:high-dim}, we describe a generalization of the hierarchical splitting to higher dimensions, where the two key properties still hold, i.e., there are only $\log N$ levels and each level $\mathbf{K}^{(\ell)}$ is block-sparse.

\subsection{Low-rank approximation of admissible blocks}\label{sec:lowrankapprox}

For far-field interactions in admissible blocks, the kernel $k(x,x')$ can be well approximated by a truncated Taylor series if it satisfies the following \emph{asymptotic smoothness} condition:
\begin{align}
    | \partial_{x'}^q k(x,x') | & \leq C q! |x-x'|^{-q} \hspace{0.2cm} (\forall q \in \mathbb{N}),\label{eq:condition1}
\end{align}
where $C$ is a constant. For instance, the log kernel $k(x,x')=\log(|x-x'|)$ directly satisfies the above.
This condition enables approximating any admissible block $\mathbf{K}^{\sigma, \rho}$ by a rank-$p$ matrix
\begin{equation}
    \mathbf{K}^{\sigma, \rho} \approx \mathbf{\widetilde{K}}^{\sigma, \rho} = \boldsymbol{\Psi}^{\sigma, \rho} \mathbf{D} \left(\boldsymbol{\Phi}^{ \rho}\right)^{\dagger},
    \label{eq:admissible-block}
\end{equation}
where $\mathbf{D} =\operatorname{Diag}(1/q!)_{q\in P} \in \mathbb{R}^{p\times p}$, $P=\{0,1,\hdots ,p-1\}$ and
\begin{align}
    & \boldsymbol{\Psi}^{\sigma, \rho} = (\partial_{x'}^q k(x_i,c_{\rho}))_{i\in\sigma, q \in P} \in \mathbb{R}^{|\sigma| \times p}, \\
    & \boldsymbol{\Phi}^{\rho} = ((x_{j}'-c_{\rho})^q)_{j\in\rho, q \in P} \in \mathbb{R}^{|\rho| \times p}.
\end{align}
Here, $p$ regulates the precision of the approximation (typically $p=\operatorname{polylog} (N/\varepsilon)$ for an error bound $\varepsilon$). A detailed example and error analysis of this far-field approximation can be found in \Cref{app:hmatrix}. In addition, we note that conditions other than that of \autoref{eq:condition1} can also be used to justify the approximation (see \cite{Tyrtyshnikov, BRANDT199124, Beylkin1991FastWT}) of admissible blocks by low-rank kernels $k(x,x')=\sum_{q=0}^{p-1}\psi_q(x)\varphi_q(x')$.

Combining this low-rank approximation and the hierarchical splitting forms a complete $\mathcal{H}$-matrix. In \Cref{app:hmatrix}, we show that matrix-vector multiplication with an $\mathcal{H}$-matrix can be performed in time $O(Np \log N)$. In addition, other $\mathcal{H}$-matrix operations such as matrix-matrix multiplication and matrix inversion run in time $O(Np^2 \log^2 N)$, we refer the reader to \cite{Fenn2002FMMAH, part1hmatrices, borm2003} for further details. The rank $p$ can be chosen at each level to further optimize computational complexity versus approximation error \cite{h2matrices}. 

%For rectangular $\mathcal{H}$-matrices, the splitting hierarchy is consistent with that of a square $\mathcal{H}$-matrix, but the admissible blocks are rectangular with the same aspect ratio of the overall matrix.

\section{Block-encoding arithmetic} \label{sec:blockenc}

Throughout this paper, we employ the block-encoding framework to perform matrix computations on a quantum computer \cite{Low_2019}. Here, an $s$-qubit (non-unitary) operator $A \in \mathbb{C}^{2^s \times 2^s}$ (let $N=2^s$) is embedded in a unitary operator $U \in \mathbb{C}^{2^{(s+a)} \times 2^{(s+a)}}$ using $a$ ancilla qubits such that the top left block of $U$ is proportional to $A$:
\begin{equation}
U = \begin{pmatrix}
A/\alpha & \cdot \\ \cdot & \cdot
\end{pmatrix},
\end{equation}
where $\alpha$ is the \emph{normalization factor} of the block-encoding.
Formally, the $(s+a)$-qubit unitary $U$ is an $(\alpha, a, \varepsilon)$-block-encoding of $A$ if 
\begin{equation}
\left\|A-\alpha(\bra{0^a} \otimes I_s) U (\ket{0^a} \otimes I_s )\right\| \leq \varepsilon,
\end{equation}
where $\|\cdot\|$ denotes the operator norm of a matrix and $I_s$ is the identity operator on $s$ qubits. In other words, applying the unitary $U$ to a quantum state $\ket{\psi}$ and post-selecting on the measurement outcome $\ket{0^a}$ on the ancilla qubits is equivalent to applying the operation $A$ on $\ket{\psi}$:
\begin{equation}
    U \ket{0^a} \ket{\psi} = \frac{1}{\alpha}\ket{0^a} A \ket{\psi} + \ket{G},
\end{equation}
where $\ket{G}$ is a garbage state that is orthogonal to the subspace $\ket{0^a}$, \textit{i.e.,} $( \bra{0^a} \otimes I_s ) \ket{G} = 0$. The probability of successfully post-selecting $\ket{0^a}$ is thus equal to $\| A \ket{\psi} \|^2 \alpha^{-2} \leq (\|A\|/\alpha)^2$. Observe that $\alpha \geq \|A\|$ is required to embed $A$ inside a unitary matrix. Further, since the probability of success of the post-selection step depends on the \emph{ratio} $\|A\|/\alpha$, a block-encoding $U$ is \emph{optimal} if $\alpha = \Theta(\|A\|)$. More generally, if the ratio $\alpha/ \|A\|$ grows logarithmically with the size of the matrix $N$, the probability of success also only changes logarithmically with $N$. \qn{More generally, as we will see some examples in \Cref{sec:applications}, the complexity of quantum algorithms \cite{gilyen2019quantum} in the block-encoding framework depends on the ratio $\alpha/\|A\|$}. This important observation motivates the following definition of a ``good'' block-encoding.

\begin{definition}[Good block-encoding]
Given a matrix $A \in \mathbb{C}^{N \times N}$, an $(\alpha, a, \varepsilon)$-block-encoding $U$ of $A$ is good if and only if the ratio $\alpha/\|A\|=O(\operatorname{polylog}(N))$, and additionally, $a=O(\operatorname{polylog}(\frac{N}{\varepsilon}))$ and the circuit depth of $U$ is $O(\operatorname{polylog}(\frac{N}{\varepsilon}))$. The block-encoding is called \emph{optimal} if $\alpha = \Theta(\|A\|)$.
\label{def:optimal}
\end{definition}

Throughout this paper, we assume that matrices are normalized to have the largest entry at most $1$, and we analyze the optimality of block-encodings by comparing the normalization factor $\alpha$ to the operator norm $\|A\|$.

Before proceeding to block-encode a hierarchical matrix, we provide various methods for block-encoding the individual matrix blocks in the hierarchical splitting. Depending on the locations of the matrix blocks and how the values of the kernel entries can be accessed by a quantum computer, there are various different ways one can block-encode a given matrix. Previous work has proposed efficient means to block-encode sparse matrices, purified density matrices, etc., as well as transformations of these matrices in the block-encoding framework \cite{gilyen2019quantum, kothari2014efficient, Low_2019}. We list some previously established results in this framework in \Cref{sec:deferred-proofs}, among which we particularly use \Cref{lem:sparse} to block-encode sparse matrices with oracle access to entries, \Cref{lem:linear_comb} to linearly combine block-encoded matrices, and \Cref{lem:multiply-encode} to multiply block-encoded matrices. Later, we will apply and combine these individual block-encoding Lemmas to form a block-encoding of the complete hierarchical matrix.

First, we provide a version of \Cref{lem:sparse} applied to dense matrices, which we refer to as the ``naive'' block-encoding approach since it is unaware of the inherent structure of the block-encoded matrix.

\begin{lemma}[Naive block-encoding of dense matrices with oracle access]
Let $A \in \mathbb{C}^{N \times N}$ (where $N=2^s$) and let $\hat{a} \geq \max_{i,j} |a_{ij}|$. Suppose the following oracle is provided
\begin{equation*}
    \mathcal{O}_{A}: \ket{i}\ket{j}\ket{0^b} \rightarrow \ket{i}\ket{j}\ket{\Tilde{a}_{ij}},
\end{equation*}
where $0\leq i,j < N$ and $\Tilde{a}_{ij}$ is the (exact) $b$-qubit description of $a_{ij}/\hat{a}$. Then one can implement a $(N \hat{a},s+1, \varepsilon)$-block-encoding of $A$ with two uses of $\mathcal{O}_A$, $O(\operatorname{polylog}(\frac{\hat{a} N}{\varepsilon}))$ one- and two-qubit gates and $O(b,\operatorname{polylog}(\frac{\hat{a} N}{\varepsilon}))$ extra qubits (which are discarded before the post-selection step).
\label{lem:densenaive}
\end{lemma}

Assuming $\hat{a}=1$, the normalization factor of this block-encoding is $N$, which is suboptimal if $\|A\| \ll N$, \textit{e.g.,} this situation arises when the matrix $A$ has many small entries. This suboptimal block-encoding could lead to an exponentially small probability of success in the post-selection step. Therefore, a structure-aware block-encoding procedure is needed to achieve an optimal runtime in terms of $N$.

%Before we present block-encoding procedures for $\mathcal{H}$-matrices in the next section, we outline here two main lemmas. 
To directly implement the $\mathcal{H}$-matrices outlined in \Cref{sec:hmatrices}, we need procedures to block-encode block-sparse and low-rank matrices which we provide in the Lemmas below (all proofs are deferred to \Cref{sec:deferred-proofs}). First we define state preparation pairs, which are frequently used in this study to prepare the coefficients in a linear combination of matrices.

\begin{definition}[State preparation pair \cite{gilyen2019quantum}] Let $y \in \mathbb{C}^{m}$ and $\|y\|_1\leq \beta$, the pair of unitaries $(P_L,P_R)$ is called a $(\beta,n,\varepsilon)$-state-preparation-pair of $y$ if $P_{L}\ket{0^n}=\sum_{j=0}^{2^{n}-1} c_{j}\ket{j}$ and $P_{R}\ket{0^n}=\sum_{j=1}^{2^{n}-1} d_{j}\ket{j}$ such that $\sum_{j=0}^{m-1}\left|\beta c_{j}^{*} d_{j} -y_{j}\right| \leq \varepsilon_1$ and $c_j^*d_j=0$ for any $j \in \{m,\hdots, 2^n-1\}$.
\label{def:preppair}
\end{definition}
Next, we show how to implement a block-encoding of low-rank matrices.
\begin{lemma}[Block-encoding of low-rank operators with state preparation unitaries, inspired by Lemma 1 of \cite{quek2021fast}]
 Let $A = \sum_{i=0}^{p-1} \sigma_i \mathbf{u}_i \mathbf{v}_i^{\dagger} \in \mathbb{C}^{2^s \times 2^s}$ where $\|\mathbf{u}_i\|= \|\mathbf{v}_i\|=1$. 
Let $r=\lceil \log p \rceil$ and  $\sum_{i=0}^{p-1} |\sigma_i| \leq \beta$. Suppose the following $(r+s)$-qubit unitaries are provided:
\begin{equation*}
\begin{aligned}
    G_L: &\ket{i}\ket{0^s} \rightarrow \ket{i}\ket{\mathbf{u}_i}, \\
    G_R: &\ket{i}\ket{0^s} \rightarrow \ket{i}\ket{\mathbf{v}_i},
\end{aligned}
\end{equation*}
where $0 \leq i<p$ and $\ket{\mathbf{u}_i}$ $(\ket{\mathbf{v}_i})$ is the quantum state whose amplitudes are the entries of $\mathbf{u}_i$ $(\mathbf{v}_i)$. Let $(P_L,P_R)$ be a $(\beta,r, \varepsilon)$-state-preparation-pair for the vector $[\sigma_0, \hdots, \sigma_{p-1}]$. Then, one can construct a $(\beta, r+s, \varepsilon)$-block-encoding of $A$ with one use of each of $G_L, G_R^{\dagger}, P_L^{\dagger}, P_R$, and a $\swap{}$ gate on two $s$-qubit registers.
\label{lem:lowrank-encode}
\end{lemma}
The above lemma provides one method to block-encode the approximately low-rank matrices described in \Cref{sec:hmatrices}, but as we will see, is not the only option one has at their disposal. When the rank of a matrix is small and bounded, the time needed to prepare unitaries $P_L, P_R$ is typically negligible; \textit{e.g.,} they can be efficiently implemented with an oracle providing access to $\sigma_i$  (and $\beta$ is close to the sum of the singular values). In addition, the unitaries $G_R, G_L$ can be constructed using existing state preparation methods based on the structure of the entries in the singular vectors (\textit{e.g.,} \cite{grover2002creating} utilized efficiently integrable distributions, \cite{kerenidis2016quantum} used a QRAM data structure).
In \Cref{app:proofa3}, we also present an approach to efficiently construct $G_R, G_L$ from oracles providing access to the entries of $\mathbf{u}_i, \mathbf{v}_i$ for cases where the entries change slowly
% (section III, \cite{Mitarai_2019})
(\Cref{lem:rank1-oracle}).
% \bk{This paragraph and definition 3.2 and lemma 3.3 can all be moved to the appendix if you want to save space. Since you are submitting to PRA, however, I wouldn't worry about it.}

Next, the following lemma block-encodes a block-sparse matrix, whose blocks are also block-encoded.

\begin{lemma}[Block-encoding of block-sparse matrices] Let $A = \sum_{i,j=0}^{2^t-1} \ket{i}\bra{j} \otimes A^{ij}$ be a $d_r$-row-block-sparse and $d_c$-column-block-sparse matrix, where each $A^{ij}$ is an $s$-qubit operator. Let $U^{ij}$ be an $(\alpha_{ij},a,\varepsilon)$-block-encoding of $A^{ij}$. Suppose that we have the access to the following $2(t+1)$-qubit oracles
\begin{equation*}
\begin{aligned}
    \mathcal{O}_r: &\ket{i}\ket{k} \rightarrow \ket{i} \ket{r_{ik}},\\
    \mathcal{O}_c: &\ket{l}\ket{j} \rightarrow \ket{c_{lj}} \ket{j},
\end{aligned}
\end{equation*}
where $0\leq i  < 2^t, 0\leq  k <d_r$, $0\leq j < 2^t, 0 \leq l < d_c$, $r_{ik}$ is the index for the $k$-th non-zero block of the $i$-th block-row of A, or if there are less than $k$ non-zero
blocks, then $r_{ik}=k+2^t$, and similarly $c_{lj}$ is the index for the $l$-th non-zero block of the $j$-th block-column of $A$, or if there are less than $l$ non-zero blocks, then $c_{lj}=l+2^t$. Additionally, suppose the following oracle is provided: 
\begin{equation*}
\begin{aligned}
    &\mathcal{O}_{\alpha}: \ket{i}\ket{j}\ket{z} \rightarrow \ket{i}\ket{j}\ket{z\oplus \Tilde{\alpha}_{ij}}\\
\end{aligned}
\end{equation*}
where $0\leq i,j <2^t$ and $\Tilde{\alpha}_{ij}$ is the $b$-qubit description of $\alpha_{ij}$ (if $i$ or $j$ is out of range then $\Tilde{\alpha}_{ij}=0$), and let the $(2t+2+s+a)$-qubit unitary $W = \sum_{i,j: A^{ij}\neq 0} (\ket{i}\ket{j} \bra{i}\bra{j}) \otimes U^{ij} + \left(I-\sum_{i,j: A^{ij}\neq 0} (\ket{i}\ket{j} \bra{i}\bra{j}) \right) \otimes I_{s+a}$. Let $\hat{\alpha} = \max_{i,j} \alpha_{ij}$. Then one can implement an $(\hat{\alpha}\sqrt{d_r d_c},t+a+3, 2\sqrt{d_r d_c} \varepsilon)$-block-encoding of $A$, with one use of each of $\mathcal{O}_r, \mathcal{O}_c,$ and $W$, two uses of $\mathcal{O}_{\alpha}$, $O(t+\operatorname{polylog}(\frac{\hat{\alpha} }{\varepsilon}))$ additional one- and two-qubit gates, and $O(b,\operatorname{polylog}(\frac{\hat{\alpha}}{\varepsilon}))$ added ancilla qubits.
\label{lem:block-sparse}
\end{lemma}

\begin{remark}
\Cref{lem:block-sparse} is efficient when the matrix is sparse over blocks (where each block may be dense), which is not in general the same as the matrix being sparse. It generalizes Lemma 48 of \cite{gilyen2019quantum} (also see \Cref{lem:sparse}) for block-encoding matrices with oracle access to entries. In addition, if $\alpha_{ij}$ are the same for all blocks, we do not need the oracle $\mathcal{O}_{\alpha}$, hence some ancilla qubits and gates can be saved (see proof in \Cref{app:proofa6}).
\label{remark:blocksparse}
\end{remark}

\section{Block-encoding of kernel matrices via hierarchical splitting}\label{sec:main}

% \bk{Again, let's have a few more informal statements here about what we achieve in this section. \textit{e.g.,}, first step is block encoding then we can perform matrix operations, we can block encode a wide variety of kernels that show up in these various problems, ...}

In this section, we apply the block-encoding techniques delineated in \Cref{sec:blockenc} to block-encode $\mathcal{H}$-matrices. \qn{Our goal is to obtain good block-encodings (\Cref{def:optimal}) of these matrices where the ratio between the normalization factor and the operator norm is at most $O( \operatorname{polylog}(N))$.} We provide an optimal block-encoding procedure for kernel matrices arising from a variety of polynomially decaying and exponentially decaying kernels. We show that using only the hierarchical splitting \emph{without subsequent low-rank approximation} gives an optimal block-encoding for these kernel matrices, enabling matrix transformations (\textit{e.g.,} multiplication, inversion, polynomial transformations) with optimal query and gate complexities. Then, we provide a block-encoding procedure for general $\mathcal{H}$-matrices whose admissible blocks are approximated by low-rank matrices and discuss the applicability of our procedure on other variants of hierarchical splitting.

\subsection{Polynomially decaying kernels}\label{sec:polykernel}
We use the hierarchical splitting described in \Cref{sec:hmatrices} to implement an optimal block-encoding procedure for polynomially decaying kernels, which take the following form
\begin{equation}
    k(x,x')=\begin{cases} C & \text{ if } x =x'\\
\|x- x'\|^{-p} & \text{ otherwise}
\end{cases},
\label{eq:polykernel}
\end{equation}
where $p>0$ and \qn{$|C| \leq  1$ represents self-interaction terms} (\textit{e.g.,} $C=0$ in Coulomb point source interactions). For concreteness, we consider the problem of computing pairwise interactions in a system of $N$ particles. Particularly, consider a 1D system of $N$ particles $\{(x_i,m_i)\}_{i=0}^{N-1}$, where $x_i$ and $m_i$ are the location and mass of the particle $i$, respectively. The potential at $x_i$ is the sum over the pairwise interactions
\begin{equation}
    \Phi (x_i) = \sum_{j=0}^{N-1} m_j k(x_i, x_j).
    \label{eq:polydecaysum}
\end{equation}
For simplicity and ease of analysis, we consider a system of equally distanced particles in which $x_j = j$ and $m_j=1$ for any $ j \in [N]$, where $N=2^L$. Note that we choose the domain to be $[1,N]$ instead of $[0,1]$ as considered in \Cref{sec:hmatrices}, so that the maximum entry in the kernel matrix is $1$. \qn{Had the particles been placed in the region $[0,1]$, we could simply downscale the matrix by dividing it by $N^p$ so that the entries are bounded by $1$. Recall from \Cref{sec:blockenc} that doing so does not change the block-encoding and algorithmic runtime as these depend on the ratio between the operator norm and the normalization factor. This ratio is unchanged under global rescaling of the matrix because the normalization factor will rescale accordingly in the block-encoding procedure.}

To access the kernel function on a quantum computer, we assume the following oracle is given
\begin{equation}
    \mathcal{O}_{k}: \ket{i}\ket{j}\ket{0}^{\otimes b} \rightarrow \ket{i}\ket{j}\ket{\Tilde{k}(x_i,x_j)},
\end{equation}
where $\Tilde{k}(x_i,x_j)$ is the $b$-qubit description of $k(x_i,x_j)$. In fact, if the function $k(x_i,x_j)$ can be computed via an efficient classical circuit, one can construct a comparably efficient quantum circuit for $\mathcal{O}_{k}$ \cite{mikeike}. In particular, this is the case when the distribution of particles is known and efficiently computable (e.g. $x_j$ are equally spaced points in the discretization of integral operators). 

We optimally block-encode the kernel matrix $\mathbf{K}= (k(x_i,x_j))_{i,j=0}^{N-1}$ up to error $\varepsilon$ with a normalization factor of $\Theta(\|\mathbf{K}\|)$ by using only two queries to $\mathcal{O}_k$ and $O(\operatorname{polylog}(N,\frac{1}{\varepsilon}))$ additional resources (i.e., one- and two-qubit gates and ancilla qubits). At a high level, our method proceeds as follows. First, we decompose $\mathbf{K}$ into a hierarchical structure of admissible blocks as described in \Cref{sec:hmatrices}
\begin{equation}
    \mathbf{K} = \sum_{\ell=2}^{L} \mathbf{K}^{(\ell)}+\mathbf{K}_{ad}.
\end{equation}
Note that this hierarchical decomposition can also be performed in higher dimensional systems (see \Cref{app:high-dim}). Next, for each admissible block (which is a dense matrix) in the level $\mathbf{K}^{(\ell)}$, we block-encode it by using the naive procedure for dense matrices (\Cref{lem:densenaive}). Then, we use the procedure for block-sparse matrices (\Cref{lem:block-sparse}) to form a block-encoding of $\mathbf{K}^{(\ell)}$. Finally, we sum over hierarchical levels to obtain a block-encoding of $\mathbf{K}$.

Specifically, at level $\ell$ of the hierarchy, each admissible block has dimension $2^{L-\ell}$. In addition, the minimum pairwise distance between particles in an admissible block in $\mathbf{K}^{(\ell)}$ is $d_{\min}^{(\ell)}=2^{L-\ell}$, so the maximum entry of an admissible block $\mathbf{K}^{\sigma,\rho}$ in $\mathbf{K}^{(\ell)}$ is bounded by $k^{(\ell)}_{\max}=d_{\min}^{-p}= 2^{-(L-\ell)p}$. It follows that an admissible block at level $\ell$ can be $(\alpha_{\ell},L-\ell+1,\frac{\varepsilon}{3})$-block-encoded via the naive procedure for dense matrices (\Cref{lem:densenaive}), where $\alpha_{\ell} = 2^{(L-\ell)(1-p)}$, using two queries to $\mathcal{O}_k$ and $O(\operatorname{polylog}( \frac{N}{\varepsilon}))$ additional one- and two-qubit gates. Note that in using \Cref{lem:densenaive} to block-encode a block $\mathbf{K}^{\sigma,\rho}$ at level $\ell$, one actually needs to query $k(x_i,x_j)/k^{(\ell)}_{\max}$ rather than $k(x_i,x_j)$ itself. This is readily achieved from the oracle $\mathcal{O}_k$ by conditioning on a register that indexes the level $\ell$ and the block indices $\sigma,\rho$ (see \Cref{app:detail-encode} for details). From here, each $\mathbf{K}^{(\ell)}$, which is a block-sparse matrix of block-sparsity 3, can be $(3\alpha_{\ell},L+3,\varepsilon)$-block-encoded via \Cref{lem:block-sparse} (particularly \Cref{remark:blocksparse}). Observe that while $\mathbf{K}^{(\ell)}$ is a $2^L \times 2^L$ matrix, the normalization factor of this block-encoding is only $3\alpha_{\ell}$, which is better than the normalization factor of $N$ that one would get when naively applying \Cref{lem:densenaive} to block-encode $\mathbf{K}^{(\ell)}$. In addition, we still only need two queries to $\mathcal{O}_k$ and $O(\operatorname{polylog}( \frac{N}{\varepsilon}))$ additional gates because all blocks $\mathbf{K}^{\sigma,\rho}$ are constructed in parallel in \Cref{lem:block-sparse} (see detailed analysis in \Cref{app:detail-encode}).

The adjacent interaction part $\mathbf{K}_{ad}$, which is a tridiagonal matrix, can be $(3,L+3,\varepsilon)$-block-encoded with the procedure for sparse matrices (\Cref{lem:sparse}), using two queries to $\mathcal{O}_k$ and $O(\operatorname{polylog}( \frac{N}{\varepsilon}))$ one- and two-qubit gates. Finally, we use the procedure for linearly combining block-encoded matrices (\Cref{lem:linear_comb}) to obtain a block-encoding of $\mathbf{K} = \sum_{\ell=2}^{L} \mathbf{K}^{(\ell)} + \mathbf{K}_{ad}$ as $U = \sum_{\ell=2}^{L} 3 \alpha_{\ell} U^{(\ell)} +  3 U_{ad}$. This step requires a $(\log L)$-qubit state preparation pair (\Cref{def:preppair}) which prepares the coefficients in the linear combination. We neglect the error in implementing this state preparation pair since it can be performed with $\log L = \log \log N$ qubit operations (see \Cref{app:detail-encode}). This entire procedure results in a $(\alpha, L +\log L + 3,\alpha \varepsilon)$-block-encoding of the \emph{exact} kernel matrix $\mathbf{K}$, where for $L = \log N$,
% \begin{equation*}
%     \alpha = 3+3\sum_{\ell=2}^{L} 2^{(L-\ell)(1-p)} = 3+3\frac{1 - 2^{(L-1)(1-p)}}{1-2^{1-p}} =3+ 3 \frac{1 - (N/2)^{1-p}}{1-2^{1-p}} \leq 3+\frac{3}{1-2^{1-p}}.
% \end{equation*}
\begin{equation}
\begin{aligned}
    \alpha &= 3+3\sum_{\ell=2}^{L} 2^{(L-\ell)(1-p)}\\
    &=\begin{cases} 3(1+ \frac{2^{1-p} -N^{1-p}}{2^{1-p}- 4^{1-p}}) \leq \Theta(1)  & \text{ if } p >1 \\
     3\log N  & \text{ if } p = 1\\
     3(1+ \frac{N^{1-p}-2^{1-p}}{4^{1-p}-2^{1-p}}) & \text{ if } p < 1\\
    \end{cases}.
\end{aligned}
\label{eq:alphasum}
\end{equation}
As before, the linear combination step does not significantly increase the circuit complexity since all levels can be constructed in parallel with only two queries to the oracle $\mathcal{O}_k$. See \Cref{app:detail-encode} for more detailed analysis of the entire block-encoding procedure.

\begin{remark}[Optimality] The block-encoding of the polynomially decaying kernels in \autoref{eq:polykernel} is optimal in the normalization factor $\alpha$. Indeed, a lowerbound on the operator norm of $\mathbf{K}$ can be obtained by evaluating the norm of the vector $\mathbf{K}\ket{\mathbf{1}}$, where $\ket{\mathbf{1}}$ is the normalized all-ones state. Observe that, $\|\mathbf{K}\| \geq \|\mathbf{K}\ket{\mathbf{1}}\| \geq \int_1^{N} x^{-p} dx$. Thus, for $p\neq 1$, $\|\mathbf{K}\| \geq \frac{N^{1-p}-1}{1-p}$ and for $p=1$, $\|\mathbf{K}\| \geq \ln N$. Clearly, the existence of our block-encoding procedure also implies that these (up to a constant factor) are an upper bound for $\|\mathbf{K}\|$.
\label{remark:opt}
\end{remark}

\begin{remark}[Generalized polynomially decaying kernel]
Since the above analysis only relies on the maximum entry in admissible blocks, it still holds for the case of non-uniform (bounded) masses $m_i$ or more general kernels of the form $k(x,x')=|x-x'|^{-p}G(x,x')$, where $G(x,x')$ is a bounded function such as $G(x,x')=\frac{1}{\sqrt{1+\epsilon|x-x'|^2}}$, $G(x,x')=e^{-i|x-x'|}$, etc. In these cases, we simply scale the kernel matrix such that the maximum entry is $1$ (if needed) and apply the above analysis. This generalized class includes a variety of kernels, ranging from power-law interactions to common Green's functions \cite{green}.
\label{remark:generalized}
\end{remark}

The above block-encoding procedure similarly applies for higher-dimensional settings where coordinates $x_j$ are in a $d$-dimensional space with $d$ constant. In \Cref{app:high-dim}, we show that there are still only $\log N$ hierarchical levels in this higher dimensional setting and that each level $\mathbf{K}^{(\ell)}$ is still block-sparse. The two key features of the hierarchical splitting that allow for optimal block-encodings are maintained, and the normalization factor is still $\alpha=\Theta(\|\mathbf{K}\|)$. The above optimal block-encoding construction and remarks are summarized in the following lemma.

\begin{lemma}[Optimal block-encoding of polynomially decaying kernel matrix using $\mathcal{H}$-matrix] Let $k(x,x')=\|x-x'\|^{-p}G(x,x')$, where $p >0$  and $G(x,x')$ is a bounded function, i.e. $|G(x,x')|< C$. Assume there exists an  efficient quantum circuit $\mathcal{O}_k$ that performs $\mathcal{O}_k: \ket{i}\ket{j}\ket{z} \rightarrow \ket{i}\ket{j}\ket{z\oplus \Tilde{k}(x_i,x_j)},$
where $\Tilde{k}(x_i,x_j)$ is a qubit string description of $k(x_i,x_j)$.
Then the matrix $\mathbf{K}= (k(x_i,x_j))_{i,j=0}^{N-1}$ can be $(\alpha, \operatorname{polylog} N, \varepsilon)$-block-encoded with $\alpha=\Theta(\|\mathbf{K}\|)$ using the hierarchical matrix splitting. That is, the block-encoding is optimal according to \Cref{def:optimal}. This block-encoding can be constructed with two queries to $\mathcal{O}_k$ and $O(\operatorname{polylog}  \frac{N}{\varepsilon})$ additional one- and two-qubit gates.
% Furthermore, the normalization factor $\alpha$ can be computed via \autoref{eq:alphasum}.
\label{lem:main}
\end{lemma}

%  We refer to \Cref{app:detail-encode} for detailed analysis.

Our block-encoding procedure is optimal for common kernel matrices, which are neither sparse nor low-rank, thus expanding the scope of application of quantum computers in the block-encoding framework. In \Cref{sec:applications}, for example, we discuss applications of this block-encoding in $N$-body force computations and solving quantum linear systems.

\begin{remark} When sources are not close to uniformly distributed, one must first apply classical adaptive methods \cite{adaptiveHmatrix, Krishnan1995APA} which discretize the domain to form a finer mesh, such that the number of nonzero-mass sources in each mesh site is $O(1)$. The dimension of the kernel matrix in this method depends on the inverse of the smallest pairwise distances between the actual sources. Therefore, the runtime of the block-encoding remains polylogarithmic in the number of particles $N$ so long as the minimum pairwise distance is not exponentially small (see \Cref{app:adaptive}).
\label{remark:adaptive}
\end{remark}

\subsection{Exponentially decaying kernels}
The procedure used in the previous section readily generalizes to exponentially decaying kernels. Consider, for example, the following kernel
\begin{equation}
    k(x,x') = e^{-\|x-x'\|^{q}} P(\|x-x'\|),
\end{equation}
where  $q>0$ and $P(\|x-x'\|)$ is a function growing polynomially or slower. Specifically, $|P(y)| \leq O(y^k)$.
Typical kernels of this form include the Laplace and Gaussian kernels, which are commonly used in machine learning \cite{Rasmussen2003gp}.

As before, we first hierarchically decompose the matrix into admissible blocks. For each admissible block $\mathbf{K}^{\sigma,\rho}$ at level $\ell$, the minimum pairwise distance between particles is $d_{\min}=2^{L-\ell}$ and therefore the maximum entry in the block can be bounded as
\begin{equation}
\begin{aligned}
    \max_{(i,j) \in (\sigma,\rho)} |k(x_i,x_j)| & \leq \max_{(i,j) \in (\sigma,\rho)} \frac{|x_i-x_j|^k}{\sum_{t=0}^{\infty} |x_i-x_j|^{qt}/t!} \\
    & \leq \max_{(i,j) \in (\sigma,\rho)} \frac{t!}{|x_i-x_j|^{qt-k}} \\
    &\leq t!2^{-(L-\ell)(qt-k)},
\end{aligned}
\end{equation}
where we choose an integer $t$ that satisfies $qt -k> 1$.

Thus, an admissible block $\mathbf{K}^{\sigma,\rho}$ at level $\ell$ can be $(\alpha_{\ell},L-\ell+3,\varepsilon/3)$-block-encoded via the naive procedure for dense matrices (\Cref{lem:densenaive}), where $\alpha_{\ell}=t!2^{(L-\ell)(1-(qt-k))}$. From here, an efficient block-encoding of $\mathbf{K}$ can be obtained via the same procedure described in the previous section, where $p$ is now replaced by $qt-k >1$. It can be verified that the normalization factor of this block-encoding is $\Theta(1)$, which is optimal since the maximum entry of $\mathbf{K}$ is of magnitude $\Theta(1)$. 

We note, however, that one can make use of a sparsification approach, in which the $\mathbf{K}$ is approximated as a sparse band matrix, and apply the block-encoding procedure for a sparse matrix (\Cref{lem:sparse}). This approach yields a block-encoding with normalization factor logarithmic in the sparsification error (see \Cref{sec:sparsify}).

\subsection{Block-encoding of general $\mathcal{H}$-matrices}
In the setting most consistent with that of classical $\mathcal{H}$-matrix literature, one may be provided oracle access to the low-rank vectors in $\boldsymbol{\Psi}^{\sigma, \rho}$ and $\boldsymbol{\Phi}^{ \rho}$ which approximate an admissible block  $\mathbf{\widetilde{K}}^{\sigma, \rho} = \boldsymbol{\Psi}^{\sigma, \rho} \mathbf{D} \left(\boldsymbol{\Phi}^{ \rho}\right)^{\dagger}$ (as seen in \autoref{eq:admissible-block}). Such a situation may also arise when one wants to apply an arbitrary $\mathcal{H}$-matrix which has no connection to any kernel function, \textit{e.g.,} see \cite{fan-multiscale-net}. In this case, we can use the block-encoding procedure for low-rank matrices (\Cref{lem:lowrank-encode}) to block-encode  $\mathbf{\widetilde{K}}^{\sigma, \rho}$. 

% This lemma assumes unitary state preparation procedures for the singular vectors. One approach to construct these unitaries is to use a QRAM data structure (\textit{e.g.,} see \cite{kerenidis2016quantum}). 

% \Cref{lem:rank1-oracle} ()

% , while \Cref{lem:rank1-oracle} assumes oracle access to the entries of them are provided. Depending on a given scenario, one can choose a suitable approach. 

After block-encoding the admissible blocks in each level of the hierarchical splitting, we follow the same procedure in previous sections to obtain the entire kernel matrix. Namely we apply the procedure for block-sparse matrices (\Cref{lem:block-sparse}) to construct a block-encoding of $\mathbf{K}^{(\ell)}$ for each level $\ell$ of the hierarchy; then we sum over all levels using the procedure for linear combination of block-encoded matrices (\Cref{lem:linear_comb}) to obtain a block-encoding of $\mathbf{K}$. The normalization factor of the block-encoding of $\mathbf{K}$ is bounded as
\begin{equation}
    \alpha = O\left(\sum_{\ell=2}^{\log N} \alpha_{\ell}\right),
\end{equation}
where $\alpha_{\ell}$ is normalization factor of the admissible blocks $\mathbf{K}^{\sigma,\rho}$ at level $\ell$ (for ease of analysis we assume admissible blocks at the same level have the same block-encoding normalization factor). As before, assuming the maximum entry of $\mathbf{K}$ is 1, the naive block-encoding (\Cref{lem:densenaive}) yields a block-encoding with a normalization factor of $N$, the dimension of the matrix. Whereas, $\mathcal{H}$-matrix block-encoding can yield an exponentially better normalization factor over the naive approach when $\alpha_{\ell}$ are $O(\operatorname{polylog}(N))$. This is \emph{not} the case for some kernels studied in classical $\mathcal{H}$-matrix literature \cite{Fenn2002FMMAH, fmm_shortcourse} such as the log kernel $k(x,x')=\log (|x-x'|)$ and the multiquadric kernel $k(x,x')=\sqrt{c+|x-x'|^2}$. In fact, we show in \Cref{app:naiveoptimal} that these particular kernels and most polynomially \emph{growing} kernels can already be optimally block-encoded using the naive approach of \Cref{lem:densenaive}. The naive approach works intuitively because these kernels are numerically low-rank and their top singular vectors are nearly uniform.
 
 % In this setting, the normalization factor of the block-encoding for each admissible block depends on the sum of its first $p$ singular values, where $p=O(\operatorname{polylog}(N/\varepsilon))$ is the rank of admissible blocks.

%  \bk{What does it mean that singular values are $O(\operatorname{polylog}(N))$? Does this mean that singular values of a given block increase at most at that rate? }. \bk{This sentence is not needed. Instead, I would outline the steps of obtaining this more general low-rank block encoding given all the inputs that you need. \textit{e.g.,} step 1 is to build a block encoding for each split in the hierarchy. Step 2 is to apply the sum over those splittings, etc. This can also go in the appendix.}

\subsection{Variants of hierarchical splitting}\label{sec:variants}

% So far, we have only considered (near) uniform distribution of particles. This is the most common case especially when kernel matrices  arise from discretizing integral operators in boundary element methods \cite{fmm_shortcourse} (also see \Cref{sec:applications}). As discussed earlier, the non-uniform cases could be be handled by applying the classical adaptive methods \cite{adaptiveHmatrix, Krishnan1995APA} before. Interestingly, we can still apply our quantum block-encoding procedure in some cases with .

The form of the hierarchical splitting can be slightly modified for other types of decaying kernels. For example, consider kernels of the form $k(x,x')= \frac{1}{|(x-x')+c|^p}$ on the domain $\Omega=[0,N-1)$ and $-N<c<N$ is an integer. In this case, we can use a ``shifted'' hierarchical splitting to obtain an optimal block-encoding of the kernel matrix $\mathbf{K}=(|i-j+c|^{-p})_{i,j=0}^{N-1}$, in which the hierarchy is shifted in the horizontal direction (along the rows). Similarly, for kernels of the form $k(x,x')= \frac{1}{(|x-x'|-c)^p}$ ($0\leq c<N$), we can apply a shifted hierarchical splitting along the skew-diagonal direction. We illustrate this approach in \Cref{app:variantHmatrix}.

In \Cref{app:generalized}, we also describe a more general block-encoding that exploits the structure of entry magnitudes in a given matrix even if the entries of the admissible blocks are not neighboring (in fact, \Cref{def:admissible} does not require the entries of an admissible block to be contiguous). At a high level, we use a pointer oracle which, conditioned on the level $\ell$, points to the entries whose magnitudes are in the range $[2^{-(\ell+1)}, 2^{-\ell})$. When this oracle can be efficiently implemented, this block-encoding procedure might find applications in other classes of dense matrices.

\begin{prop}[Block-encoding using generalized hierarchical splitting] Let $A \in \mathbb{R}^{N \times N}$ be a Hermitian matrix with non-negative entries which are provided via an oracle $\mathcal{O}_{A}: \ket{i}\ket{j}\ket{z} \rightarrow \ket{i}\ket{j} \ket{z\oplus \Tilde{a}_{ij}}$, where $\Tilde{a}_{ij}$ is an exact bit string description of $a_{ij}$. Suppose $N=2^L$ and $2^{-L}\leq a_{ij} \leq 1$. For any $0\leq \ell < L=\log N$ and a column $j$, define its level-$\ell$ row-index subset to be $I_{\ell}(j)=\{i|2^{-(\ell+1)} < a_{ij} \leq 2^{-\ell}\}$ and let $n_{\ell}(j)=|I_{\ell}(j)|$. Suppose there exist $n_{\ell}$ and $\gamma$ such that $\gamma  n_{\ell} \leq n_{\ell}(j) \leq n_{\ell}$ for any $j$. Furthermore, suppose there exists an \emph{invertible} function $f_j(\ell,k)$ which returns the row index of the $k$-th element (according to some fixed ordering, e.g. from top to bottom in the column) in $I_{\ell}(j)$. Then, one can construct a block-encoding of $A$ with a normalization factor of at most $\frac{2 \log N}{\gamma} \|A\|$ using two queries to $\mathcal{O}_A$, $O(\operatorname{polylog}(\frac{N}{\varepsilon}))$ elementary gates, and $O(\operatorname{polylog}(\frac{N}{\varepsilon}))$ ancilla qubits.
\label{prop:generalized}
\end{prop}

Intuitively, $\gamma$ measures how balanced the magnitudes of the matrix entries are across rows and columns. If $\frac{1}{\gamma}= \operatorname{polylog}(N)$, this yields a good block-encoding according to \Cref{def:optimal}. For instance, $\frac{1}{\gamma}$ is a constant for kernel matrices. In the proof of \Cref{prop:generalized}, we also present an efficient state preparation procedure using methods of $\mathcal{H}$-matrix splitting in \Cref{app:generalized}, which could be of independent interest.

% We leave for future work the question of whether there are other kernels such that other structured block-encoding procedures can be directly used.

% \begin{lemma}[Preparing states whose classical descriptions are stored in the data structure in \cite{kerenidis2016quantum}] \qn{This is the recommendation systems paper. Here one can prepare state as a unitary procedure, $(1,0,0)
% $-block-encoding.  But it apparently uses QRAM?}
% \end{lemma}

% \qn{More state preparation procedures to be added...}

% One can choose the appropriate state preparation procedure depending on the problem in hand, such that it is either a unitary (a.k.a $(1,0,0)$-block-encoding) or a block-encoding with a normalization factor of $O(1)$.

% \begin{lemma}[Divide-and-conquer approach] \cite{Araujo2021} \qn{Same as Kerenidis's?}
% \end{lemma}

\section{Applications}\label{sec:applications}

In this section, we discuss applications of our quantum hierarchical block-encoding method in two commonly studied settings: (i) gravitational potential calculation \cite{Barnes1986} and (ii) solving integral equations \cite{fmm_shortcourse}. The first problem requires performing matrix-vector multiplication on a quantum computer, while the second problem is equivalent to solving a quantum linear system. We also compare our procedure to the sparsification approach and address potential issues with large condition numbers.

We first state the following lemmas, which allow one to apply a block-encoded matrix or its inverse to a quantum state. Other polynomial transformations of block-encoded matrices can be performed by the quantum singular value transform (Theorem 17 of \cite{gilyen2019quantum}).

\begin{lemma}[Applying block-encoded matrices] If $U$ is an $(\alpha, a, \varepsilon)$-block-encoding of the matrix $A \in \mathbb{C}^{2^s \times 2^s}$ and $U_{\mathbf{x}}$ is a unitary that prepares the $s$-qubit quantum state $\ket{\mathbf{x}}$, then it takes $O\left(\frac{\alpha}{\|A\|} \kappa\right)$ queries to $U$ and $U_{\mathbf{x}}$ to obtain a quantum state $\ket{\mathbf{y}}$ such that $\| \ket{\mathbf{y}} - \ket{A\mathbf{x}} \| \leq \varepsilon$, where $\kappa$ is the condition number of $A$.
\label{lem:forward}
\end{lemma}
\begin{proof}
Applying the operator $U (I_a \otimes U_{\mathbf{x}})$ on the state $\ket{0^a}\otimes \ket{0^s}$ we obtain
\begin{equation}
    \ket{0^a} \frac{A}{\alpha}\ket{\mathbf{x}} + \ket{\text{garbage}},
\end{equation}
where $\ket{\text{garbage}}$ is orthogonal to the subspace $\ket{0^a}$. By post-selecting the subspace $\ket{0^a}$, we obtain the desired quantum state $\ket{A \mathbf{x}}$. This step succeeds with a probability of $\left\|\frac{A\ket{\mathbf{x}}}{\alpha} \right\|^2= \Omega((\frac{\|A\|}{\alpha \kappa} )^2)$. Using amplitude amplification \cite{amplitudeamp}, this can be improved to $\Omega(\frac{\|A\|}{\alpha \kappa} )$, giving us the specified query complexity.
\end{proof}

\begin{lemma}[Applying inverse block-encoded matrices, adapted from \cite{costa2021optimal}] If $U$ is an $(\alpha, a, \varepsilon)$-block-encoding of an invertible matrix $A \in \mathbb{C}^{2^s \times 2^s}$ and $U_{\mathbf{b}}$ is a unitary that prepares the $s$-qubit quantum state $\ket{\mathbf{b}}$, then it takes $O\left(\frac{\alpha \kappa}{\|A\|}\log (1/\varepsilon)\right)$ queries to $U$ and $U_{\mathbf{b}}$ to obtain a quantum state $\ket{\mathbf{y}}$ such that $\|\ket{\mathbf{y}} - \ket{A^{-1}\mathbf{b}}\| < \varepsilon$.
\label{lem:inv}
\end{lemma}
\begin{proof}
We apply the main theorem of \cite{costa2021optimal}, which assumed $\alpha=\|A\|=1$ to achieve a query complexity of $O(\kappa \log(1/\varepsilon))$. In general, $\kappa$ needs to be replaced by the inverse of the minimum singular value of the block-encoded part. In our case, this singular value is $\frac{\|A\|}{\alpha \kappa}$.
\end{proof}

\qn{As noted in \Cref{sec:blockenc}, the query complexities of the above procedures, and more generally quantum algorithms in the block-encoding framework \cite{gilyen2019quantum}, only depend on the ratio $\alpha/\|A\|$. Hence, a block-encoding is optimal (\Cref{def:optimal}) when $\alpha = \Theta(\|A\|)$. In the following, we apply such optimal block-encodings in \Cref{sec:main} to two specific applications of kernel matrices.}

\subsection{Gravitational potential calculation: Quantum fast multipole method}
We apply the block-encoding obtained in \Cref{sec:polykernel} to solve the gravitational potential problem \cite{Barnes1986} on a quantum computer. Since we use a quantum version of the hierarchical splitting, this is a quantum analogue to the classical fast multipole method \cite{Greengard1987}. Consider a 1D system of $N$ particles which are approximately uniformly distributed in the domain $[0,O(N)]$, where the mass and location of particle $j$ are $m_j$ and $x_j$, respectively. Furthermore, assume that $m_j$ are $\Theta(1)$ and the smallest pairwise distance between the particles is $\Omega(1)$. We would like to compute the accumulated potential at every particle $i$, stored in a vector $\Phi$ with entries $\Phi_i=\sum_{j\neq i} m_j|x_i-x_j|^{-1}$.

\begin{prop}[Quantum fast multipole method] Given an oracle $\mathcal{O}_k$ that performs $\mathcal{O}_k: \ket{i}\ket{j}\ket{z} \rightarrow \ket{i}\ket{j}\ket{z\oplus \Tilde{k}_{ij}},$
where $\Tilde{k}_{ij}$ is a $b$-qubit string description of $k(x_i,x_j)$, and a procedure $P_m$ that prepares the quantum state $\ket{\mathbf{m}}$ which stores the normalized masses of the particles, we can obtain a quantum state $\ket{\phi}$ that $\varepsilon$-approximates the normalized vector $\Phi/\|\Phi\|$, whose $j$-th entry is the accumulated potential at particle $j$, using $O(1)$ queries to $P_m$, $O\left(\operatorname{polylog}\frac{N}{\varepsilon}\right)$ one- and two-qubit gates, and $O\left(b,\operatorname{polylog}\frac{N}{\varepsilon}\right)$ additional qubits. % and $O(1)$ applications of a quantum circuit of depth $O\left(\operatorname{polylog}\frac{N}{\varepsilon}\right)$ acting on $O\left(b,\operatorname{polylog}\frac{N}{\varepsilon}\right)$ qubits.
\end{prop}
\begin{proof} We first use $P_m$ to create the state $\ket{\mathbf{m}}= \sum_{j=0}^{N-1} \frac{m_j}{\|\mathbf{m}\|}\ket{j}$ (assume $\log N$ is an integer). Let $\mathbf{K}=( k(x_i,x_j))_{i,j}$, observe that $\ket{\Phi}= \mathbf{K}\ket{\mathbf{m}} = \frac{\Phi}{\|\mathbf{m}\|}$. Using \Cref{lem:main}, we can construct a unitary $U$, which is an $(\alpha, a, \varepsilon)$-block-encoding of $\mathbf{K}$, with the specified gate complexity and wherein $\alpha = \Theta(\log N)$ and $a=O(\log N)$. 
% Applying this unitary on the state $\ket{0^a}\otimes \ket{\mathbf{1}}$ we obtain
% \begin{equation}
%     \ket{0^a} \frac{\mathbf{K}}{\alpha}\ket{\mathbf{1}} + \ket{\text{garbage}},
% \end{equation}
% where $\ket{\text{garbage}}$ is orthogonal to the subspace $\ket{0^a}$.
Then, we apply \Cref{lem:forward}, with a more careful analysis on the post-selection step. Here, the probability of succesfully post-selecting the subspace $\ket{0^a}$ is $\|\frac{\Phi}{\alpha \|\mathbf{m}\|} \|^2= \Theta(1)$ since $\|\mathbf{m}\|=\Theta(\sqrt{N})$ and $\|\Phi\|=\Omega(\sqrt{N} \log N)$ (by noting that the particle masses are $\Theta(1)$ and using similar calculations to those in \Cref{remark:opt}). This gives us the specified query complexity.
\end{proof}
We remark that the (non-unitary) state preparation procedure $P_m$ can be efficiently implemented when $m_j=\Theta(1)$ (see Theorem 1 of \cite{mitarai2019quantum}). In addition, if particles are not (nearly) uniformly distributed, we can discretize the domain to form a finer mesh as described in \Cref{remark:adaptive}. Here we have assumed these distances are $\Omega(1)$, hence the size of the mesh (and the size of the kernel matrix) is $O(N)$. This method also enables computing the potential at locations other than the sources themselves. This proposition also applies to $N$-body calculations in higher-dimensional systems as the hierarchical splitting and \Cref{lem:main} can be generalized to these cases (see \Cref{app:high-dim}).

\subsection{Integral equation: Quantum algorithm for linear systems of kernel matrix}
We provide an application of our hierarchical block-encoding procedure in solving a quantum linear systems of kernel matrices. These problems arise in many research areas, such as Gaussian processes \cite{Rasmussen2003gp} and integral equations \cite{atkinson2009numerical}. Consider the following integral equation over the circular strip of unit radius and height and width $\lambda \ll 1$ as shown in \autoref{fig:strip}.

\begin{figure*}[]
    \centering
    \includegraphics[width=0.99\textwidth]{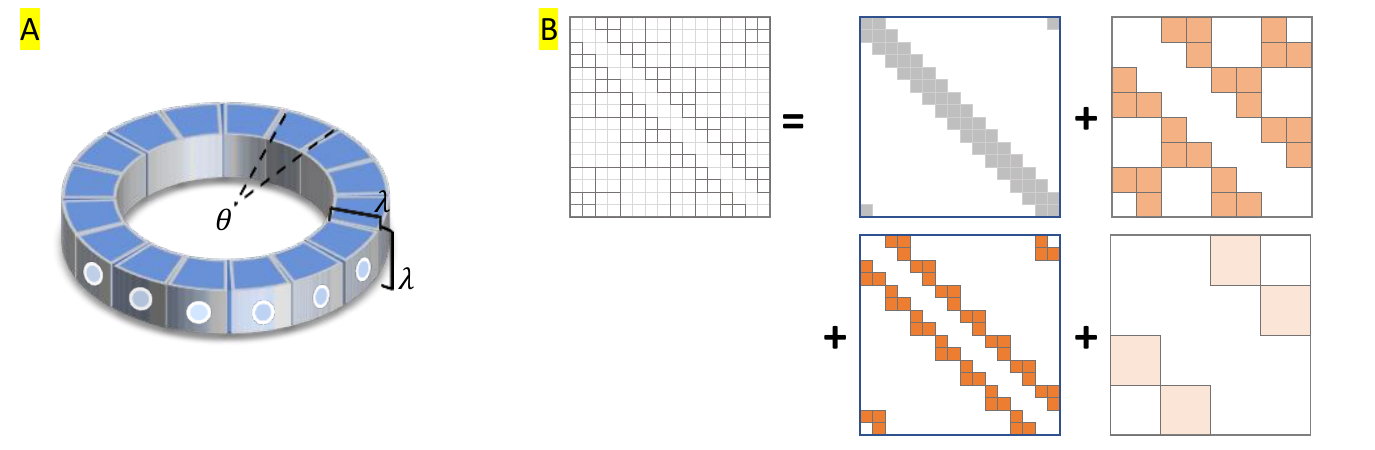}
    \caption{(A) A thin circular strip of unit radius and height $\lambda \ll 1$. When using the collocation method to solve an integral equation on its surface, one divides the strip into $N$ ``panels'' such that \qn{$\lambda=O(1/N)$}. (B) The ``cyclic'' hierarchical splitting of the kernel matrix that arises when applying the collocation method to solve for an integral equation defined over the strip (\autoref{eq:ie}).}
    \label{fig:strip}
\end{figure*}

\begin{equation}
    g(\mathbf{x}) = f(\mathbf{x}) +  \int_{\Omega} k(\mathbf{x},\mathbf{x'}) f(\mathbf{x'}) d\mathbf{x'},
    \label{eq:ie}
\end{equation}
where $k(\mathbf{x},\mathbf{x'})= \|\mathbf{x}-\mathbf{x'}\|^{-p}$ with $0< p\leq 2$, $g(\mathbf{x})$ is a known smooth function, and $f(\mathbf{x})$ is the unknown function we would like to solve. Equations of this form are commonly found in, \textit{e.g.,}, Dirichlet problems \cite{atkinson2009numerical}, where $k(\mathbf{x},\mathbf{x'})$ is the Coulomb potential, $g(\mathbf{x})$ is the given potential on the strip, and $f(\mathbf{x})$ is the unknown charge density. Here, we use the collocation method to numerically solve this integral equation which represents the unknown function in a basis $f(\mathbf{x})= \sum_{i=1}^{N} f_i \varphi_i(x)$ and solves for the coefficients $f_i$ such that the integral equation is satisfied at specified locations. Particularly, we use the piecewise-constant function as the basis $\varphi_i(\theta)=1$ if $\theta \in [2\pi i/N,2 \pi (i+1)/N)$ and 0 otherwise. We refer to the segment $[2\pi i/N,2 \pi (i+1)/N)$ as \emph{panel} $i$. Then, the integral on the RHS of \autoref{eq:ie} can be rewritten as
\begin{equation}
    \sum_{i=1}^{N} f_i \int_{\text{panel } i} \frac{1}{\|\mathbf{x}-\mathbf{x'}\|^{p}} d\mathbf{x'},
\end{equation}
In the collocation method, we solve for the integral equation at the centroids of the panels located at $c_i = 2 \pi (i+0.5)/N$. In other words, we solve the following linear system
\begin{equation}
    \mathbf{g} = (I + \mathbf{K}) \mathbf{f},
    \label{eq:discretized}
\end{equation}
where bold text indicates discretized values, \textit{e.g.,}, $\mathbf{g}_i=g(c_i)$ and $\mathbf{K}_{ij}= \int_{\text{panel } j} \|c_i-\mathbf{x'}\|^{-p} d\mathbf{x'}$.

Observe that this collocation method can handle the singularity when $i=j$ as $\mathbf{K}_{ii}=O(\lambda^{3-p})$. Furthermore, when $|i-j|$ is large, $\mathbf{K}_{ij}\approx \text{distance}^{-p} \cdot \text{panel volume} \approx \frac{\lambda^2}{N(2\sin(\pi |i-j|/N))^{p}}$. For panels $i$, $j$ close to each other (\textit{e.g.,} $|i-j|<O(1)$), simple adaptive methods can be used to approximate the integral. For instance, one can subdivide the panel $j$ into eight smaller panels and approximate the integral as the sum of these eight sub-panels; we neglect the details and refer the interested reader to \cite{bem}. In summary, the entries of $\mathbf{K}$ are simple to calculate and they approximately are
\begin{equation}
    \mathbf{K}_{ij} \approx  \begin{cases}
    \frac{\lambda^2}{N(2\sin(\pi |i-j|/N))^{p}}& \text{ for } i\neq j \\
    O(\lambda^{3-p})& \text{ for } i=j
    \end{cases}.
    \label{eq:col-kernel}
\end{equation}
The error bound and the uniqueness of the solution in the collocation method are well-studied in classical literature, we refer to \cite{atkinson2009numerical, bem, Hackbusch_BEM} for more details.

We now proceed to block-encode this kernel matrix. Observe that $\frac{1}{N (2\sin(\pi |i-j|/N))^{p}} \leq \frac{1}{(\pi |i-j|)^p}$ for any \qn{panels $i,j$ such that $|i-j|\leq N/2$}, hence we can apply a ``cyclic'' version of the hierarchical splitting in \Cref{sec:hmatrices} to \emph{optimally} block-encode $\mathbf{K}$ (see \Cref{fig:strip}B). As analyzed in \Cref{sec:polykernel}, this block-encoding only has a circuit depth of $O(\operatorname{polylog}\frac{N}{\varepsilon})$ and a normalization factor of $\alpha = \Theta(\|\mathbf{K}\|)$, as the operator norm can be lowerbounded using the same technique as in \Cref{remark:opt} by noticing that $\frac{1}{N (2\sin(\pi |i-j|/N))^p} \geq \frac{1}{(2\pi |i-j|)^p}$ for any $|i-j|\leq N/2$.

% \begin{remark}[Optimality] The above block-encoding for the kernel matrix in \autoref{eq:col-kernel} is optimal. Observe that $\frac{1}{N (\sin(\pi |i-j|/N))^p} \geq \frac{1}{(\pi |i-j|)^p}$ for any $|i-j|\leq N/2$. Thus, we have that $\|\frac{\mathbf{K}}{\lambda}\| \geq \|\frac{\mathbf{K}}{\lambda}\ket{\mathbf{1}}\| \geq \frac{1}{2\pi} \int_1^{N} x^{-1} dx = \frac{1}{2\pi} \ln N$ (where $\ket{\mathbf{1}}$ denotes the normalized all-ones vector). The existence of our block-encoding implies that $\Theta(\log N)$, up to a constant factor, is also an upper bound for $\|\frac{\mathbf{K}}{\lambda}\|$.
% \end{remark}

 Finally, it is straightforward to use \Cref{lem:linear_comb} (linear combination of block-encoded matrices) to obtain an optimal block-encoding of $I+\mathbf{K}$. Then, we directly apply \Cref{lem:inv} to solve the linear system in \autoref{eq:discretized}. Since the block-encoding is optimal, the query complexity is $O(\kappa \log(1/\varepsilon))$, for a total runtime of $O(\kappa \operatorname{polylog}(N/\varepsilon))$ when including  the circuit complexity of the block-encoding (\Cref{lem:main}).
 
 \begin{prop}[Solving discretized integral equations] Consider the integral equation on a ring as described in \autoref{eq:ie}. We first cast the equation into the $N\times N$ linear system in \autoref{eq:discretized} using the collocation method with $N$ discretized points. Then, we can obtain its numerical solution as a quantum state in time $O(\kappa \operatorname{polylog}(N/\varepsilon))$, where $\kappa$ is the condition number of the matrix $I+\mathbf{K}$ in \autoref{eq:discretized} and $\varepsilon$ is the error bound of the quantum operation.
 \label{prop:ie}
 \end{prop}
 
 We address two practical aspects when applying \Cref{prop:ie}. First, we must prepare $\mathbf{g}$ as a quantum state. In \Cref{app:stateprep}, we present an efficient procedure for preparing $\ket{\mathbf{g}}$ in the case where $\mathbf{g}$ is sampled from a smooth function. At a high level, our procedure prepares the quantum state in the Fourier regime, which is approximately sparse due to the smoothness of $g$; then we apply the quantum Fourier transform to obtain the desired real-regime state. To prove \Cref{prop:generalized}, we also present another state preparation procedure using ideas from $\mathcal{H}$-matrix splitting in \Cref{app:generalized} which could be of independent interest. Second, the query complexity of the block-encoding depends on the condition number $\kappa$ of the matrix $I+\mathbf{K}$, which can be bounded as follows
\begin{equation}
    \kappa \leq \frac{1+\|\mathbf{K}\|}{1-\|\mathbf{K}\|}.
    % = \frac{1+\lambda \Theta(\log N)}{1-\lambda \Theta(\log N)}.
    \label{eq:boundkappa}
\end{equation}
% The dependence of $\|\mathbf{K}\|$ on $N$ depends on $p$ as calculated in \Cref{remark:opt}. 
Note that here, $\|\mathbf{K}\|$ scales with $\lambda$. Since we require $\lambda \ll 1$ for the collocation method to work, the operator norm $\|\mathbf{K}\|$, and thus $\kappa$, is bounded. Therefore, the runtime of \Cref{prop:ie} is $O(\operatorname{polylog}\frac{N}{\varepsilon})$. 

% Noting that $\mathbf{K}$ is positive semidefinite, 

% \textcolor{red}{Small problem: we  require $\lambda \ll 1/N$ to use the collocation panel approximation, but this would also mean we don't need to use H-matrix block-encoding since the naive normalization factor would be $O(1)$ already.}

% Fourier transform trick, useful for preparing RHS of QLSP. See \Cref{app:stateprep}

\subsection{Notes on sparsification approach}\label{sec:sparsify}
Since the entries of kernel matrices often decay along a row or column, one may assume that they can be approximated by only considering a sparse set of matrix entries and applying the naive block-encoding method for sparse matrices (\Cref{lem:sparse}). Here, we analyze this approach showing that it is only favorable when kernel entries decay exponentially and not polynomially. In fact, the runtime of this method is exponentially worse in the  target error bound $\varepsilon$ compared to that of our approach. Consider the polynomially decaying kernel in \autoref{eq:polykernel}. One could attempt to approximate this kernel matrix by a band matrix $\mathbf{K}_b$ of sparsity $d$. Defining the error to be $\|\mathbf{K}-\mathbf{K}_b\|$, it can be seen that the error is of magnitude $\int_{d}^{N} x^{-p} dx$. Since this integral diverges when $p\leq 1$, one can only sparsify the matrix when $p>1$. In this case, to achieve an error bound of $\varepsilon_{s}$, the sparsity $d$ needs to be $\Omega(\varepsilon_s^{\frac{1}{1-p}})$. Then, one can apply \Cref{lem:sparse} (block-encoding of sparse matrices) to directly block-encode the sparsified matrix with a normalization factor of $\Theta(\varepsilon_{s}^{\frac{1}{1-p}})$ using $O(\log N+ \operatorname{polylog}(\frac{1}{\varepsilon}))$ extra gates and $O(1)$ oracle queries. 
This results in an exponentially worse runtime in terms of the error bound compared to our method, which has a normalization factor of $O(1)$. 

For the exponential kernels, however, the sparsification method requires keeping $\Theta(\log(1/\varepsilon_s))$ or more entries per row or column. With \Cref{lem:sparse}, the sparsified block-encoding has a normalization factor of $\Theta(\log(1/\varepsilon_s))$ using the same gate and query complexities as above. Thus, exponential kernels may admit a sparsification approach with an extra factor of $\Theta(\log(1/\varepsilon))$ in the runtime.

\subsection{Notes on condition number}
The condition number $\kappa$ and its relation to the dimension $N$ of the matrix play a central role in the runtime of our algorithms.  Numerical experiments shown in \Cref{app:conditionnumber} indicate that for purely polynomially decaying kernels without any regularization and ignoring self-interacting terms (diagonal entries set to zero), the condition number grows polynomially with $N$, \textit{i.e.,} the minimum singular value decays with $N$. In practice, however, the input vector often lives in the well-conditioned subspace where these smaller singular values can be ignored as observed in the gravitation example earlier. For force calculation problems, this is the case when particles have comparable masses/charges. For numerical integration or integral equations, this is equivalently the case when the input functions change slowly.

We note that \Cref{lem:inv} for matrix inversion uses the discrete adiabatic theorem, whose runtime depends directly on the condition number $\kappa$ of the matrix operation. Hence, in the case where the matrix operation $A$ is ill-conditioned but the input vector lives in the well-conditioned subspace, we can use techniques that allow for filtering out small eigenvalues before the inversion step such as the QSVT \cite{gilyen2019quantum}. The query complexity of the block-encoding in matrix inversion in this case is $O(\frac{ \log (1/\zeta, 1/\varepsilon) }{\zeta \|A^{-1} \ket{b}\| })$, where $\zeta$ is a chosen threshold bounding the singular values in the well-conditioned subspace.

Fortunately, for applications in numerical analysis or machine learning, the kernel matrix is often regularized by an added matrix, \textit{e.g.,} the addition of an identity matrix as in Fredholm integral equations of the second kind or the noise matrices $\sigma_n^{2}I$ in Gaussian processes and kernel ridge regression \cite{atkinson2009numerical,Rasmussen2003gp}. In the integral equation setting, this regularization is equivalent to scaling the magnitude of self-interacting terms (diagonal entries).
Furthermore, if the input vector is potentially in the ill-conditioned subspace, preconditioning methods in \cite{precondition} may be used to improve the condition number. 
% https://epubs.siam.org/doi/pdf/10.1137/16M1105396
% https://arxiv.org/pdf/1602.06693.pdf
% https://web.stanford.edu/~rezab/nips2013workshop/accepted/preconditioned.pdf
% \cite{chen_2005, tyrtyshnikov1992, preconditionedHHL, precondition, circprecondtion}

\section{Conclusion}
% We have presented a quantum algorithm for a class of matrices called hierarchical matrices, which are neither sparse nor low-rank. Our algorithm builds on the celebrated block-encoding framework and existing quantum linear system solvers. We particularly apply our algorithm to optimally implementing decaying kernel matrices, which have wide applications in physics, engineering, and machine learning. In addition, we discuss applications of our algorithm in an $N$-body force calculation problem and a quantum linear system arising from integral equations. We also show that our algorithm applies to a broad class of functions, including ones we called the generalized decaying kernels and even a trigonometric kernel such as the sine inverse. Our algorithm provides an exponential speedup over the best known classical and quantum algorithms and opens up a new path for block-encoding dense matrices, in which one exploits the structure of the magnitudes of the kernels.

Our results show that factoring certain kernel matrices into hierarchical or $\mathcal{H}$-matrices leads to efficient algorithms for implementing such kernel matrices on a quantum computer, thus expanding the scope of application of quantum computers to a class of matrices which are neither low-rank nor sparse. In fact, our approach to implementing such hierarchical matrices, in a sense, blends prior quantum methods for implementing unitary block encodings by hierarchically splitting a matrix into a sparse matrix component close to the diagonal and progressively larger blocks farther away from the diagonals. Looking forward, it is an interesting question whether any extensions to the $\mathcal{H}$-matrix framework can lead to further improvements in the quantum setting. Classically, the most significant extension of the $\mathcal{H}$-matrix framework is perhaps the development of $\mathcal{H}^2$-matrices, which incorporates a further level of hierarchical nested bases to provide a logarithmic speedup in the matrix dimension for performing matrix operations \cite{borm2003,lin2011fast}. Furthermore, recent classical work has parameterized $\mathcal{H}$-matrices to integrate them into neural networks and design learning algorithms with inductive biases that match the form of the hierarchical splitting \cite{fan-multiscale-net,fan2019multiscale}. Future work can study whether such parameterizations can be performed as well in the quantum setting via variational quantum algorithms or other quantum analogues to classical neural networks \cite{cerezo2021variational}.

\acknowledgements{We thank the anonymous reviewers for helpful feedback and corrections on this paper. This project was funded by AFOSR, ARO, and DARPA. BTK acknowledges funding from the MIT Energy Initiative fellowship.}

\printbibliography

\appendix

\section{Prior works on block-encoding and quantum linear systems algorithms}\label{app:compare}
In this section, we provide a brief overview and discuss the runtimes of previous works in \Cref{tab:compare}. As a reminder, the quantum linear system problem (QLSP) seeks an approximate solution $\ket{x'}$ to
\begin{equation}
    A \Vec{x} = \Vec{b}
\end{equation}
such that
\begin{equation}
    \left\|\ket{x'}-\frac{\Vec{x}}{\|\Vec{x}\|}\right\| \leq \varepsilon.
\end{equation}

Let $N$ be the dimension; $d$ be the sparsity; $\kappa$ be the condition number; and $k$ be the rank of the matrix $A$, which is assumed to be square, normalized ($\|A\|=1$), and Hermitian  \footnote{If not, one can consider the equivalent linear system $\begin{pmatrix}0&A\\A^{\dagger}&0 \end{pmatrix}\Vec{y}=\begin{pmatrix}\Vec{b}\\0 \end{pmatrix}$, whose solution is $\Vec{y}=\begin{pmatrix}0\\\Vec{x} \end{pmatrix}$.}. The original HHL algorithm \cite{hhl} runs in time $O(d^2 \kappa^2 \operatorname{polylog} (N)/\varepsilon)$ by performing a Hamiltonian simulation of $A$ (i.e. $e^{-iAt}$) and quantum phase estimation (QPE) subroutines. Here, the Hamiltonian simulation \cite{Berry2007} and QPE steps together contribute a factor of $O(d^2 \kappa\operatorname{polylog}(N)/\varepsilon)$, while the post-selection step contributes the remaining factor of $O(\kappa)$.

Later, \cite{childshhl} noticed that the QPE step can be circumvented by utilizing the fact that the inverse function $\frac{1}{x}$ can be written, to the desired precision $\varepsilon/\kappa$, as a linear combination of $O(\kappa \operatorname{polylog}(\kappa/\varepsilon))$ Chebyshev polynomials. Combining this expansion with the variable-time amplitude amplification technique in \cite{ambainis2010variable}, they achieved a runtime of $O(d \kappa \operatorname{polylog} (N \kappa /\varepsilon))$. This polynomial expansion approach is generalized in the block-encoding framework by \cite{gilyen2019quantum, power-block} to perform almost any polynomial matrix transformation with the same complexity. All these works use a naive block-encoding in the oracle-access model (Lemma 48 in \cite{gilyen2019quantum}, same as \Cref{lem:sparse} or \Cref{lem:densenaive}).

In another approach, \cite{wossnig2017} used a QRAM data structure to implement a unitary $W$ whose eigenvalues are related to the singular values of $A$ (this algorithm does not require $A$ to be square or Hermitian) via its Frobenius norm $\|A\|_F$. Thus, performing QPE on $W$ yields a quantum singular value estimation algorithm for $A$. Matrix multiplication and inversion can then be implemented immediately with a rotation operator controlled on the register storing the singular values, just like done in HHL. The runtime of both operations, including the post-selection step, is $O(\kappa^2 \|A\|_F \operatorname{polylog}(N)/\varepsilon).$ For a full-rank matrix whose singular values are $O(1)$, we have $\|A\|_F=\sqrt{N}$. This QRAM model, however, is strictly stronger than the oracle-access model as it requires the entries be loaded and pre-processed in the data structure.

Recent adiabatic algorithms for QLSP \cite{adiabaticQLSP, An_2022, Lin_2020} also enjoy success and culminate in \cite{costa2021optimal}, which achieved an optimal query complexity of $O(\kappa \log (1/\varepsilon))$ by constructing a sophisticated schedule function for the adiabatic evolution. This query complexity is better than the algorithms in \cite{childshhl, gilyen2019quantum, wossnig2017} and has been shown to be optimal \cite{costa2021optimal}.
However, all of these assume that the matrix $A$ is perfectly block-encoded inside a unitary $U = \begin{pmatrix}\frac{A}{\|A\|}&\cdot \\ \cdot&\cdot \end{pmatrix}$ which is generally hard to construct without prior knowledge of the structure of $A$. The exact query complexity and total runtime in the general case is stated in \Cref{lem:inv} and governed by the ratio $\frac{A}{\alpha}$, where $\alpha$ is the normalization factor in the block-encoding.

The runtimes shown in \Cref{tab:compare} are obtained by combining the optimal adiabatic solver in \cite{costa2021optimal} (see \Cref{lem:inv}) with the corresponding block-encoding used in \cite{childshhl, gilyen2019quantum, power-block, kerenidis2016quantum, wossnig2017}. In particular, the naive oracle-access block-encoding in \cite{childshhl, gilyen2019quantum, power-block} yields a normalization factor of $\alpha = d=N$ for dense matrices. The QRAM-based block-encoding in \cite{kerenidis2016quantum, wossnig2017} yields $\alpha=\|A\|_F=\sqrt{N}$ for full-rank matrices. The runtime for matrix multiplication is simply obtained by computing the probability of success in the post-selection step after applying the block-encoding, as described in \Cref{sec:blockenc} and \Cref{lem:forward}.

Quantum-inspired algorithms are a class of classical algorithms equipped with sample-query access, a classical analogue to QRAM. The currently fastest quantum-inspired algorithms for QLSP \cite{gilyen2022improved, shao2021faster} are based on stochastic gradient descent (specifically randomized Kaczmarz method) for a runtime of $\widetilde{O}(\|A\|_F^6 \kappa^8 /\varepsilon^2)$ (where the notation $\widetilde{O}$ suppresses the log factors). Assuming $\kappa=O(1)$ and the matrix has rank $k$, this runtime is $\widetilde{O}(k^6/\varepsilon^2)$, which is probitive when the matrix is full-rank.

We note that all the algorithms above are ``general-purpose'' since they do not assume any inherent structures in the matrix $A$. Whereas, our work builds on these past works and focuses on optimizing the performance of QLSP algorithms for dense and full-rank kernel matrices and potentially other classes of matrices that suit the $\mathcal{H}$-matrix framework.

% \qn{Note: Lemma 49 of QSVT seems to improve runtime to $\sqrt{N}$ but requires QRAM, so we go with Lemma 48 (naive block-encoding)}

\section{Details of block-encodings}\label{sec:deferred-proofs}
We list some previously established results in the block-encoding framework and prove the helper lemmas presented in the main text. Most of the following results are drawn from Section 4 of \cite{gilyen2019quantum}. As a reminder, a block-encoding is defined as follows.

% Other polynomial transformations of block-encoded matrices can be efficiently performed via the quantum singular value transform (QSVT--Lemma 17 of \cite{gilyen2019quantum}).

\begin{definition}[Block-encoding \cite{gilyen2019quantum}] Suppose that $A$ is an $s$-qubit operator, $\alpha, \varepsilon>0$, then we say that the $(s+a)$-qubit unitary $U$ is an $(\alpha, a, \varepsilon)$-block-encoding of $A$, if
\begin{equation*}
\left\|A-\alpha(\bra{0}^{\otimes a} \otimes I_s) U (\ket{0}^{\otimes a} \otimes I_s )\right\| \leq \varepsilon,
\end{equation*}
where $\|\cdot\|$ denotes the operator norm of a matrix.
\end{definition}

\subsection{Linear combinations and multiplications of block-encodings}
Combinations of block-encoded matrices can be efficiently combined linearly or multiplied with each other. First we define state-preparation-pairs used to prepare the coefficients in a linear combination.
\begin{definition}[State preparation pair \cite{gilyen2019quantum}] Let $y \in \mathbb{C}^{m}$ and $\|y\|_1\leq \beta$, the pair of unitaries $(P_L,P_R)$ is called a $(\beta,n,\varepsilon)$-state-preparation-pair of $y$ if $P_{L}\ket{0}^{\otimes n}=\sum_{j=0}^{2^{n}-1} c_{j}\ket{j}$ and $P_{R}\ket{0}^{\otimes n}=\sum_{j=1}^{2^{n}-1} d_{j}\ket{j}$ such that $\sum_{j=0}^{m-1}\left|\beta c_{j}^{*} d_{j} -y_{j}\right| \leq \varepsilon_1$ and $c_j^*d_j=0$ for any $j \in \{m,\hdots, 2^n-1\}$.
% \label{def:preppair}
\end{definition}

We can then implement a linear combination of block-encoded matrices using an above state preparation pair.
\begin{lemma}[Linear combination of block-encoded matrices, adapted from \cite{gilyen2019quantum}] Let  $A = \sum_{j=0}^{m-1} y_j A_j$ be an $s$-qubit operator where $\|y\|_1 \leq \beta$. Suppose $(P_{L}, P_{R})$ is a $(\beta, n,\varepsilon_1)$-state-preparation-pair for $y$, $W=\sum_{j=0}^{m-1}\ket{j}\bra{j} \otimes U_{j}+\left(\left(I-\sum_{j=0}^{m-1}\ket{j}\bra{j}\right) \otimes I_{a} \otimes I_{s}\right)$ is an $(n+a+s)$-qubit unitary such that for all
$j \in [m]$ we have that $U_{j}$ is an $\left(\alpha, a, \varepsilon_{2}\right)$-block-encoding of $A_{j}$. Then the unitary $\widetilde{W}=(P_{L}^{\dagger} \otimes I_{a} \otimes I_{s}) W (P_{R} \otimes I_{a} \otimes I_{s})$ is an $\left(\alpha \beta, a+n, \alpha \varepsilon_{1}+ \beta \varepsilon_{2}\right)$-block-encoding of $A$.
\label{lem:linear_comb}
\end{lemma}

\begin{proof} We have:
\begin{equation*}
\begin{aligned}
    &\left\| A-\alpha \beta \left(\bra{0}^{n} \bra{0}^{a} \otimes I_s \right) \widetilde{W} \left(\ket{0}^{n} \ket{0}^{a} \otimes I_s \right) \right\| \\
    = &\left\| A-\alpha \sum_{j=0}^{m-1} \beta\left(c_{j}^{*} d_{j}\right) \left( \bra{0}^{a} \otimes I_s \right) U_{j} \left( \ket{0}^{a} \otimes I_s \right) \right\|  \\
    \leq &\alpha \varepsilon_1 + \left\| A-\alpha \sum_{j=0}^{m-1} y_j \left( \bra{0}^{a} \otimes I_s \right) U_{j} \left( \ket{0}^{a} \otimes I_s \right) \right\| \\
    \leq &\alpha \varepsilon_1 +\sum_{j=0}^{m-1} |y_j| \left\|A_j-\alpha \left( \bra{0}^{a} \otimes I_s \right) U_{j} \left( \ket{0}^{a} \otimes I_s \right) \right\|\\
    \leq &\alpha \varepsilon_1 + \beta \varepsilon_2. \qedhere
\end{aligned} 
\end{equation*} 
\end{proof}

We note that in the original work \cite{gilyen2019quantum}, Lemma 52 states that the error of the block-encoding is $\alpha \varepsilon_1+\alpha \beta \varepsilon_2$, which is a typo. \Cref{lem:linear_comb} also applies to cases where $P_L$ and $P_R$ are non-unitary. As long as block-encodings of $P_L$, $P_R$ are provided, we can still implement the unitary $\widetilde{W}$ in \Cref{lem:linear_comb} by using the following lemma, which allows one to multiply block-encoded matrices.

\begin{lemma}[Product of block-encoded matrices \cite{gilyen2019quantum}] If $U$ is an $(\alpha, a, \delta)$-block-encoding of an $s$-qubit operator $A$, and $V$ is an $(\beta, b, \varepsilon)$-block-encoding of an $s$-qubit operator $B$ then $\left(I_{b} \otimes U\right)\left(I_{a} \otimes V\right)$ is an $(\alpha \beta, a+b, \alpha \varepsilon+\beta \delta)$-block-encoding of $A B$. Note that here $I_a$ ($I_b$) acts on the ancilla qubits of $U$ ($V$).
\label{lem:multiply-encode}
\end{lemma}
\begin{proof}
\begin{equation*}
\begin{aligned}
    &\|AB - \alpha \beta (\bra{0^{a+b}}\otimes I_s)(I_b \otimes U)(I_a \otimes V) (\ket{0^{a+b}} \otimes I_s) \| \\
    = &\|AB - \underbrace{\alpha  (\bra{0^{a}}\otimes I_{s}) U (\ket{0^{a}} \otimes I_{s})}_{\widetilde{A}} \underbrace{ \beta  (\bra{0^{b}}\otimes I_{s}) V  (\ket{0^{b}} \otimes I_{s})}_{\widetilde{B}} \| \\
    =& \|AB - \widetilde{A}B + \widetilde{A} B - \widetilde{A}\widetilde{B} \| \\
    =& \|(A-\widetilde{A})B + \widetilde{A}(B-\widetilde{B})\| \\
    \leq & \|A-\widetilde{A}\|\beta +  \alpha \|B-\widetilde{B}\| \\
    \leq & \alpha \varepsilon + \beta \delta. \qedhere
\end{aligned}
\end{equation*}
\end{proof}

\subsection{Block-encoding of oracle-access matrices} 
The following lemma can be used to block-encode any matrix whose entries are provided by an oracle.

\begin{lemma}[Block-encoding of oracle-access matrices, adapted from Lemma 48 of \cite{gilyen2019quantum}] 
Let $A \in \mathbb{C}^{2^s \times 2^s}$ be a $d_r$-row-sparse and $d_c$-column-sparse matrix. Suppose that we have the access to the following $2(s+1)$-qubit oracles
\begin{equation*}
\begin{aligned}
    &\mathcal{O}_r: \ket{i}\ket{k} \rightarrow \ket{i} \ket{r_{ik}} & 
    0\leq  i < 2^s, 0 \leq k < d_r,\\
    & \mathcal{O}_c: \ket{l}\ket{j} \rightarrow \ket{c_{lj}} \ket{j} & 0 \leq j < 2^s, 0 \leq l <d_c,
\end{aligned}
\end{equation*}
where $r_{ik}$ is the index for the $k$-th non-zero entry of the $i$-th row of A, or if there are less than $k$ non-zero
entries, then $r_{ik}=k+2^s$, and similarly $c_{lj}$ is the index for the $l$-th non-zero entry of the $j$-th column of $A$, or if there are less than $l$ non-zero entries, then $c_{lj}=l+2^s$. Additionally, let $\hat{a} \geq \max_{i,j} |a_{ij}|$ and suppose the following oracle is provided
\begin{equation*}
    \mathcal{O}_{A}: \ket{i}\ket{j}\ket{0}^{\otimes b} \rightarrow \ket{i}\ket{j}\ket{\Tilde{a}_{ij}}, \hspace{0.5cm}  0 \leq i,j < 2^s,
\end{equation*}
where $\Tilde{a}_{ij}$ is the (exact) $b$-qubit description of $a_{ij}/\hat{a}$, and if $i$ or $j$ is out of range then $\Tilde{a}_{ij}=0$. Then one can implement an $(\hat{a}\sqrt{d_r d_c},s+3, \varepsilon)$-block-encoding of $A$ with a single use of $\mathcal{O}_r, \mathcal{O}_c$, two uses of $\mathcal{O}_A$ and additionally $O(s+\operatorname{polylog}(\frac{\hat{a} \sqrt{d_r d_c}}{\varepsilon}))$ one- and two-qubit gates while using $O(b,\operatorname{polylog}(\frac{\hat{a} \sqrt{d_r d_c}}{\varepsilon}))$ ancilla qubits (which are discarded before the post-selection step).
\label{lem:sparse}
\end{lemma}

\begin{proof}
Let $G_d$ be the $(s+1)$-qubit unitary that performs: $\ket{0} \rightarrow \frac{1}{\sqrt{d}} \sum_{i=0}^{d-1} \ket{i}$, which can be implemented with $O(s)$ Hadamard gates (assume $d$ is a power of 2 and note that $d\leq 2^s$). Additionally, define the following $2(s+1)$-qubit operators: $V_R = \mathcal{O}_c (G_{d_c}\otimes I_{s+1})$ and $V_L = \mathcal{O}_r (I_{s+1} \otimes G_{d_r})\swap_{s+1}$, where $\swap_{s+1}$ swaps two $(s+1)$-qubit registers (which uses $s$ two-qubit swap gates). For any $0\leq i,j <2^s$, we have that
\begin{equation*}
\begin{aligned}
    &\bra{0^{s+1}} \bra{i} V_L^{\dagger} V_R \ket{0^{s+1}} \ket{j} \\ 
    =&\left(\frac{1}{\sqrt{d_r}}\sum_{k=0}^{d_r-1} \bra{i} \bra{r_{ik}} \right) \left(\frac{1}{\sqrt{d_c}} \sum_{l=0}^{d_c-1} \ket{c_{lj}} \ket{j}  \right) \\
    =& \frac{1}{\sqrt{d_r d_c}} \text{ if } a_{ij}\neq 0 \text{ and } 0 \text{ otherwise.} 
\end{aligned}
\end{equation*}
% Note that the registers $\ket{i}, \ket{j}$ are $(s+1)$-qubit to handle

Observe that after applying $V_R$, instead of applying $V_L^{\dagger}$ right away, we can first append a $b$-qubit register and call $\mathcal{O}_A$ to query the matrix entry, then append the register $\ket{0}$ and perform the controlled rotation: $\ket{0}\ket{\Tilde{a}_{ij}} \rightarrow (a_{ij}/\hat{a} \ket{0} +\sqrt{1-|a_{ij}/\hat{a}|^2}\ket{1}) \ket{\Tilde{a}_{ij}} $, and finally call $\mathcal{O}_A$ again to uncompute the $b$-qubit ancilla register. Conditioning on the ancilla qubit being $\ket{0}$, we obtain the desired block-encoding. We only count $s+3$ qubits as ancilla qubits in the block-encoding notation because the others can be discarded immediately after the uncomputation step.

We analyze the circuit complexity of the controlled rotation step above. Observe that, assuming the $b$-qubit description $\ket{\Tilde{a}_{ij}}$ is exact, if the controlled rotation can be done with an accuracy of $O(\frac{\varepsilon}{\hat{a} \sqrt{d_r d_c}})$ then we achieve the overall accuracy of the block-encoding. This step implements circuits to compute elementary functions such as the square root and trigonometric functions with $O(\log (\frac{\hat{a} \sqrt{d_r d_c}}{\varepsilon}))$ additional qubits. Implementing these elementary functions on these additional qubits requires $O(b,\operatorname{polylog}(\frac{\hat{a} \sqrt{d_r d_c}}{\varepsilon}))$ elementary logic gates \cite{computational, barak}.
% This is the operation complexity on a multi-tape Turing machine <-> convertable to Boolean circuit (Computational Complexity: A Modern Approach, Sanjeev Arora and Boaz Barak). BCK15 page 12 uses simple techniques based on Taylor series and long multiplication gives log^(2.5), which I found agrees with https://en.wikipedia.org/wiki/Computational_complexity_of_mathematical_operations#cite_note-8
After computing the rotation angles, it takes $O(\log (\frac{\hat{a} \sqrt{d_r d_c}}{\varepsilon}))$ controlled quantum gates to perform the controlled rotation.
\end{proof}

We can see that the normalization factor of the above block-encoding depends on the sparsity of the matrix. Without loss of generality, assume $\hat{a}=1$. Then, the above lemma when applied to dense matrices produces a block-encoding with a normalization factor of $O(N)$. This is not optimal (see definition in \Cref{sec:blockenc}) if the operator norm of the matrix is significantly smaller than $O(N)$. Hence, we refer to it as the ``naive'' block-encoding. We provide some examples in which this lemma works optimally in \cref{app:naiveoptimal}.

\begin{remark}[``Naive'' block-encoding of dense oracle-access matrices (same as \Cref{lem:densenaive})]
In the dense case, the oracles $\mathcal{O}_r, \mathcal{O}_c$ are not needed; we can simply use $s$ Hadamard gates to query all matrix entries. It can be easily verified that in this case we can obtain an $(\hat{a}2^s,s+1, \varepsilon)$-block-encoding  with two uses of $\mathcal{O}_A$ and additionally $O(s+\operatorname{polylog}(\frac{\hat{a}2^s}{\varepsilon}))$ one- and two-qubit gates while using $O(b,\operatorname{polylog}(\frac{\hat{a}2^s}{\varepsilon}))$ ancilla qubits.
\label{remark:densenaive}
\end{remark}

If the matrix $A$ is rectangular, that is $A \in \mathbb{C}^{2^s \times 2^t}$, where we assume $s>t$ without loss of generality, we can block-encode $A$ by applying the above lemmas on the $2^{s}$ by $2^s$ matrix $\widetilde{A}$ defined as
\begin{equation}
    \widetilde{A}_{ij} = \begin{cases} A_{ij}& \text{if } i < 2^s \text{ and } j < 2^t  \\
    0 & \text{otherwise}
    \end{cases}.
\end{equation}
Then for an input state $\ket{\psi}$, we can obtain $A \ket{\psi}$ by applying $\widetilde{A}$ on the state $\ket{0^{s-t}}\otimes \ket{\psi}$.

\subsection{Block-encoding of low-rank matrices} \label{app:proofa3}
Recall \Cref{lem:lowrank-encode}, copied below.
\begin{lemma}[Block-encoding of low-rank operators with state preparation unitaries, inspired by Lemma 1 of \cite{quek2021fast}]
 Let $A = \sum_{i=0}^{p-1} \sigma_i \mathbf{u}_i \mathbf{v}_i^{\dagger} \in \mathbb{C}^{2^s \times 2^s}$. 
Let $r=\lceil \log p \rceil$ and  $\sum_{i=0}^{p-1} |\sigma_i| \leq \beta$. Suppose the following $(r+s)$-qubit unitaries are provided:
\begin{equation*}
    G_L: \ket{i}\ket{0^s} \rightarrow \ket{i}\ket{\mathbf{u}_i},
    G_R: \ket{i}\ket{0^s} \rightarrow \ket{i}\ket{\mathbf{v}_i},
\end{equation*}
where $0 \leq i <p$ and $\ket{\mathbf{u}_i}$ $(\ket{\mathbf{v}_i})$ is the quantum state whose amplitudes are the entries of $\mathbf{u}_i$ $(\mathbf{v}_i)$. Let $(P_L,P_R)$ be a $(\beta,r, \varepsilon)$-state-preparation-pair for the vector $[\sigma_0, \hdots, \sigma_{p-1}]$.
Then, one can construct a $(\beta, r+s, \varepsilon)$-block-encoding of $A$ with one use of each of $G_L, G_R^{\dagger}, P_L^{\dagger}, P_R$ and a $\swap{}$ gate on two $s$-qubit registers.
\end{lemma}
\begin{proof}
    Observe that the $(2s+r)$-qubit unitary $U = (P_L^{\dagger}\otimes I_{2s})(G_R^{\dagger}\otimes I_{s})(I_r \otimes \swap_s)(G_L \otimes I_{s})(P_R \otimes I_{2s})$ is a $(\beta, r+s, \varepsilon)$-block-encoding of $A$. Indeed, for any $0\leq m, n < 2^s$ we have that
\begin{equation*}
\begin{aligned}
    &(\bra{0^{r+s}}\otimes \bra{m}) U (\ket{0^{r+s}}\otimes \ket{n}) \\
    = &\left(\sum_{j=0}^{p-1} c_j^*\bra{j} \bra{\mathbf{v}_j} \bra{m} \right)(I_r \swap_s) \left(\sum_{i=0}^{p-1} d_i \ket{i} \ket{\mathbf{u}_i} \ket{n}\right) \\
    = & \sum_{i=0}^{p-1} c_i^* d_i \bra{\mathbf{v_i}} n \rangle \langle m \ket{\mathbf{u}_i} = \langle m |  \left( \sum_{i=0}^{p-1} c_i^* d_i  \ket{\mathbf{u}_i} \bra{\mathbf{v_i}} \right)| n \rangle,
\end{aligned}
\end{equation*}
and 
\begin{equation*}
\begin{aligned}
    & \left\|\beta \left(\sum_{i=0}^{p-1} c_i^* d_i  \ket{\mathbf{u}_i} \bra{\mathbf{v_i}}\right) - A \right\| \\
    = &\left\| \left(\sum_{i=0}^{p-1} (\beta c_i^* d_i - \sigma_i)  \ket{\mathbf{u}_i} \bra{\mathbf{v_i}}\right) \right\|\\
    \leq & \sum_{i=0}^{p-1} |(\beta c_i^* d_i - \sigma_i)| \leq \varepsilon. \qedhere
\end{aligned}
\end{equation*}
\end{proof}
The unitaries $G_R, G_L$ can be constructed using existing state preparation methods based on the structure of the entries in the singular vectors (\textit{e.g.,} \cite{grover2002creating} utilized efficiently integrable distributions, \cite{kerenidis2016quantum} used a QRAM data structure). As an example,
when given oracle-access to the entries of the vectors $\mathbf{u}, \mathbf{v}$, one can use the following lemma to block-encode rank-1 operators.
\begin{lemma}[Block-encoding of rank-1 operators with oracle access entries] Let $A=\mathbf{u} \mathbf{v}^{\dagger} \in \mathbb{C}^{2^s \times 2^s}$ be a rank-1 operator. Suppose the following oracles are provided:
\begin{equation*}
    \mathcal{O}_{\mathbf{u}}: \ket{i}\ket{z}^{b} \rightarrow \ket{i}  \ket{z \oplus \Tilde{u}_i}, \mathcal{O}_{\mathbf{v}}: \ket{j}\ket{z}^{b} \rightarrow \ket{j}  \ket{z \oplus \Tilde{v}_j},
\end{equation*}
where $\Tilde{u}_i (\Tilde{v}_j)$ is the $b$-qubit exact description of $u_i (v_j)$. Let $\max_i |u_i| \leq \hat{u}$ and $\max_j |v_j| \leq \hat{v}$. Then one can obtain a $(2^{s}\hat{u}\hat{v},s+2b+2,0)$-block-encoding of $A$ with two uses of $\mathcal{O}_{\mathbf{u}}$, $\mathcal{O}_{\mathbf{v}}$ each.
\label{lem:rank1-oracle}
\end{lemma}

\begin{proof}
    Let $\operatorname{CR}_{\hat{u}}$ denote the conditioned rotation operator: $\operatorname{CR}_{\hat{u}} \ket{\Tilde{u}_i}\ket{0} = \ket{\Tilde{u}_i} \left(\frac{u_i}{\hat{u}} \ket{0} + \sqrt{1-\left|\frac{u_i}{\hat{u}}\right|^2 }\ket{1}\right)$. \\
Let $G_{\mathbf{u}}=  (\mathcal{O}_{\mathbf{u}} \otimes I) (I_{s} \otimes CR_{\hat{u}}) ( \mathcal{O}_{\mathbf{u}} \otimes I) ( H^{\otimes s} \otimes I_{b+1})$. Observe that
\begin{equation*}
    (I_{s} \otimes \bra{0}^{b+1}) G_{\mathbf{u}} \ket{0}^{s+b+1}  = 2^{-s/2} \sum_{i=0}^{2^s-1}\frac{u_i}{\hat{u}} \ket{i}.
\end{equation*}
Define $G_{\mathbf{v}}$ similarly as above. Then we have\footnote{Here,  $G_{\mathbf{u}}$ and $G_{\mathbf{v}}$ act on each other register $\ket{0}^{b+1}$.}
\begin{equation*}
\begin{aligned}
    &\bra{0}^{s} \bra{m} \bra{0}^{b+1} \bra{0}^{b+1}(G_{\mathbf{v}}^{\dagger} \otimes I)(I \otimes  \swap_{s}) (G_{\mathbf{u}}\otimes I) \ket{0}^{s} \ket{n} \ket{0}^{b+1} \ket{0}^{b+1}\\
    & = \frac{1}{2^{s}} \sum_{i,j=0}^{2^{s}-1} \braket{j|n} \braket{m|i} \frac{u_i}{\hat{u}}\frac{v_j}{\hat{v}} = \frac{1}{2^{s}\hat{u}\hat{v}} u_m v_n. \qedhere
\end{aligned}
\end{equation*}
\end{proof}
This low-rank block-encoding can also be applied for rectangular matrices. Suppose $\mathbf{u} \in \mathbb{C}^{2^m}$ and $\mathbf{v} \in \mathbb{C}^{2^n}$, where we assume without loss of generality $m>n$, then we can apply the lemmas on the matrix $\widetilde{A} =\mathbf{u} (\bra{0}^{m-n} \otimes \mathbf{v}^{\dagger})$.

\begin{remark} The normalization factor $2^s\hat{u}\hat{v}$ is to ensure that the block-encoding is unitary. Thus, \Cref{lem:rank1-oracle} works best when the entries $u_i$ and $v_j$ change slowly. If $\mathbf{u}$ or $\mathbf{v}$ is sparse, we can apply the techniques in \Cref{lem:sparse} to improve the normalization factor of this low-rank block-encoding. In general, however, the above lemma is rather redundant since we can directly apply \Cref{lem:sparse} naively by constructing the matrix entry-access oracle from $\mathcal{O}_{\mathbf{u}}$ and $\mathcal{O}_{\mathbf{v}}$. We simply present this lemma here for completeness.

% This is the case when, \textit{e.g.,}, $u_i$ and $v_j$ are sampled from a smooth function.
\end{remark}

\subsection{Block-encoding of block-sparse matrices (\Cref{lem:block-sparse})} \label{app:proofa6}
In this section, we provide a block-encoding procedure for block-sparse matrices.
First we consider the block-diagonal case and later use these techniques for constructing general block-sparse matrices.

\begin{lemma}[Block-encoding of block-diagonal matrices] Let $\{A^j: 0\leq j \leq 2^t-1\}$ be a set of $s$-qubit operators and each $U^j$ be an $(\alpha_j,a,\varepsilon)$-block-encoding of $A^j$. Let $W=\sum_{j=0}^{2^t-1} \ket{j}\bra{j} \otimes U^j$ and $\hat{\alpha} = \max_j \alpha_j$. Furthermore, suppose that the following oracle is provided: $\mathcal{O}_{\alpha}: \ket{j}\ket{z} \rightarrow \ket{j}\ket{z\oplus \Tilde{\alpha}_j}$, where $\Tilde{\alpha}_j$ is the (exact) $b$-qubit description of $\alpha_j$ and $z$ is a $b$-qubit string. Then one can obtain an $(\hat{\alpha}, a+1, 2 \varepsilon)$-block-encoding of $A=\sum_{j=0}^{2^t-1} \ket{j}\bra{j} \otimes A^j$ with two uses of $\mathcal{O}_{\alpha}$, one use of $W$, and additionally $O(\operatorname{polylog}(\frac{\hat{\alpha}}{\varepsilon}))$ one- and two-qubit gates while using $O(b,\operatorname{polylog}(\frac{\hat{\alpha}}{\varepsilon}))$ ancilla qubits (which are discarded before the post-selection step).
\label{lem:block-diagonal}
\end{lemma}
\begin{proof}
Let $\operatorname{CR}$ denote the controlled rotation operator: $\operatorname{CR} \ket{\Tilde{\alpha}_j}\ket{0} = \ket{\Tilde{\alpha}_j} \left(\Bar{\alpha}_j \ket{0} + \sqrt{1-|\Bar{\alpha}_j|^2 }\ket{1}\right)$, with $|\Bar{\alpha}_j-\frac{\alpha_j}{\hat{\alpha}}|\leq \frac{\varepsilon}{\hat{\alpha}}$. Implementing this operator requires $O(\operatorname{polylog}(\frac{\hat{\alpha}}{\varepsilon}))$ one- and two-qubit gates while using $O(b,\operatorname{polylog}(\frac{\hat{\alpha}}{\varepsilon}))$ ancilla qubits (see explanation in the proof of \Cref{lem:sparse}).
Observe that $\widetilde{W}= (\mathcal{O}_{\alpha} \otimes I_{s+a+1}) (\operatorname{CR} \otimes I_{s+t+a}) (\mathcal{O}_{\alpha} \otimes I_{s+a+1}) (I_{b+1}\otimes W)$ is the desired block-encoding\footnote{\label{firstnote}The subscripts indicate the registers the identity act on.}. Indeed, we have that
\begin{equation*}
    \widetilde{W} (\ket{0}\otimes \ket{0}^b \otimes I_{s+t+a}) = \sum_{j=0}^{2^t-1} \left(\Bar{\alpha}_j\ket{0}+\sqrt{1-|\Bar{\alpha}_j|^2 } \ket{1}\right)\otimes \ket{0}^b\otimes \ket{j} \bra{j} \otimes U^j,
\end{equation*}

Therefore,
\begin{equation*}
\begin{aligned}
    &\left\| \hat{\alpha} \left(\bra{0}^{1+b+a} \otimes I_{s+t}\right) \widetilde{W} \left(\ket{0}^{1+b+a} \otimes I_{s+t}\right)- A \right\| \\
    =&   \left\| \bigoplus_{0\leq j<2^t} \left( \hat{\alpha} \Bar{\alpha}_j \bra{0}^a  U^j\ket{0}^a  - A^j\right)  \right\| \\ 
    \leq &   \left\| \bigoplus_{0\leq j<2^t} \left( \alpha_j \bra{0}^a  U^j\ket{0}^a  - A^j\right)  \right\|  + \left\| \bigoplus_{0\leq j<2^t} (\hat{\alpha}\Bar{\alpha}_j - \alpha_j) \bra{0}^a  U^j\ket{0}^a  \right\| \\ 
    \leq & \max_{0\leq j<2^t} \left\| \alpha_j \bra{0}^a  U^j\ket{0}^a  - A^j\right\| + \max_{0\leq j<2^t} \left\| (\hat{\alpha} \Bar{\alpha}_j - \alpha_j) \bra{0}^a  U^j\ket{0}^a \right\| \\ \leq & 2 \varepsilon.
\end{aligned}
\end{equation*}
We do not include $b$ in the notation of the block-encoding because this register can actually be taken out of the system right after the uncomputation step.
\end{proof}

In the above lemma, if the normalization factors $\alpha_j$ are the same for all blocks, then the operator $W$ is already an $(\hat{\alpha}, a,\varepsilon)$-block-encoding of the desired block-diagonal matrix.

Next, we apply the proof techniques above to generalize to the block-sparse case. The following lemma is the same as \Cref{lem:block-sparse}.

\begin{lemma}[Block-encoding of block-sparse matrices] Let $A = \sum_{i,j=0}^{2^t-1} \ket{i}\bra{j} \otimes A^{ij}$ be a $d_r$-row-block-sparse and $d_c$-column-block-sparse matrix, where each $A^{ij}$ is an $s$-qubit operator. Let $U^{ij}$ be an $(\alpha_{ij},a,\varepsilon)$-block-encoding of $A^{ij}$. Suppose that we have the access to the following $2(t+1)$-qubit oracles
\begin{equation*}
\begin{aligned}
    &\mathcal{O}_r: \ket{i}\ket{k} \rightarrow \ket{i} \ket{r_{ik}} \qquad \forall i \in \{0,\hdots, 2^t-1\}, k \in \{0,\hdots, d_r-1\},\\
    & \mathcal{O}_c: \ket{l}\ket{j} \rightarrow \ket{c_{lj}} \ket{j} \qquad \forall j \in \{0,\hdots, 2^t-1\}, l \in \{0,\hdots, d_c-1\},
\end{aligned}
\end{equation*}
where $r_{ik}$ is the index for the $k$-th non-zero block of the $i$-th block-row of A, or if there are less than $k$ non-zero
blocks, then $r_{ik}=k+2^t$, and similarly $c_{lj}$ is the index for the $l$-th non-zero block of the $j$-th block-column of $A$, or if there are less than $l$ non-zero blocks, then $c_{lj}=l+2^t$. Additionally, suppose the following oracle is provided: 
\begin{equation*}
\begin{aligned}
    &\mathcal{O}_{\alpha}: \ket{i}\ket{j}\ket{z} \rightarrow \ket{i}\ket{j}\ket{z\oplus \Tilde{\alpha}_{ij}}, \forall i,j \in \{0,\hdots,2^t-1\}\\
\end{aligned}
\end{equation*}
where $\Tilde{\alpha}_{ij}$ is the $b$-qubit description of $\alpha_{ij}$ (if $i$ or $j$ is out of range then $\Tilde{\alpha}_{ij}=0$), and let the $(2t+2+s+a)$-qubit unitary $W = \sum_{i,j: A^{ij}\neq 0} (\ket{i}\ket{j} \bra{i}\bra{j}) \otimes U^{ij} + \left(I-\sum_{i,j: A^{ij}\neq 0} (\ket{i}\ket{j} \bra{i}\bra{j}) \right) \otimes I_{s+a}$. Let $\hat{\alpha} = \max_{i,j} \alpha_{ij}$. Then one can implement an $(\hat{\alpha}\sqrt{d_r d_c},t+a+3, 2\sqrt{d_r d_c} \varepsilon)$-block-encoding of $A$, with one use of each of $\mathcal{O}_r, \mathcal{O}_c,$ and $W$, two uses of $\mathcal{O}_{\alpha}$, $O(t+\operatorname{polylog}(\frac{\hat{\alpha} }{\varepsilon}))$ additional one- and two-qubit gates, and $O(b,\operatorname{polylog}(\frac{\hat{\alpha}}{\varepsilon}))$ extra ancilla qubits (which are discarded before the post-selection step).
\end{lemma}

\begin{proof} Let $\operatorname{CR}$ be the controlled rotation operator: $\operatorname{CR} \ket{\Tilde{\alpha}_{ij}}\ket{0} = \ket{\Tilde{\alpha}_{ij}} \left(\Bar{\alpha}_{ij} \ket{0} + \sqrt{1-|\Bar{\alpha}_{ij}|^2 }\ket{1}\right)$, where $|\Bar{\alpha}_{ij}-\frac{\alpha_{ij}}{\hat{\alpha}}|\leq \frac{\varepsilon}{\hat{\alpha} }$. Let $G_d$ be a $(t+1)$-qubit unitary such that $G_d\ket{0^{t+1}} = \sum_{k=0}^{d-1} \frac{\ket{k}}{\sqrt{d}}$, where $d \leq 2^t$. Additionally, let $V_R = \mathcal{O}_c (G_{d_c}\otimes I_{t+1})$ and $V_L = \mathcal{O}_r (I_{t+1} \otimes G_{d_r})\swap_{t+1}$. For any $0\leq i, j<2^t$, observe that
\begin{equation}
\begin{aligned}
    \bra{0^{t+1}}\bra{i} V_L^{\dagger} V_R \ket{0^{t+1}} \ket{j} &= \left(\sum_{k=0}^{d_r-1} \bra{i} \frac{\bra{r_{ik}}}{\sqrt{d_r}} \right) \left(\sum_{k=0}^{d_c-1} \ket{c_{lj}} \frac{\ket{j}}{\sqrt{d_c}}  \right) \\
    & =\frac{1}{\sqrt{d_r d_c}} \text{ if } A^{ij}\neq 0 \text{ and } 0 \text{ otherwise.} 
\end{aligned}
\label{eq:VLVR}
\end{equation}

Let the $(2t+2+s+a)$-qubit unitary $W = \sum_{i,j: A^{ij}\neq 0} (\ket{i}\ket{j} \bra{i}\bra{j}) \otimes U^{ij} + (I-\sum_{i,j: A^{ij}\neq 0} (\ket{i}\ket{j} \bra{i}\bra{j})) \otimes I_{s+a}$ and, similarly to \Cref{lem:block-diagonal}, $\widetilde{W}= (\mathcal{O}_{\alpha} \otimes I_{s+a+1}) (\operatorname{CR} \otimes I_{s+2t+2+a}) (\mathcal{O}_{\alpha} \otimes I_{s+a+1}) (I_{b+1}\otimes W)$ be a $(2t+s+a+b+3)$-qubit unitary
\footnote{The subscripts indicate the registers the identity act on.}. Then $(V_L^{\dagger}\otimes I_{s+a+b+1} ) 
\widetilde{W} (V_R\otimes I_{s+a+b+1})$ is the desired block-encoding. Indeed, this follows directly from \autoref{eq:VLVR} and a similar proof to that of \Cref{lem:block-diagonal}, noting that
$\widetilde{W}$ leaves the $\ket{c_{lj}}$ and $\ket{j}$ registers unchanged, and
\begin{equation*}
\begin{aligned}
      & \left\| \sum_{i,j: A^{ij}\neq 0 } \ket{i} \bra{j} \otimes (\hat{\alpha} \Bar{\alpha}_{ij} \bra{0}^a U^{ij}\ket{0}^a - A^{ij}) \right\| \\
      \leq & \left\| \sum_{i,j: A^{ij}\neq 0 } \ket{i} \bra{j} \otimes (\alpha_{ij} \bra{0}^a U^{ij}\ket{0}^a - A^{ij}) \right\| + \left\| \sum_{i,j: A^{ij}\neq 0 } \ket{i} \bra{j} \otimes (\hat{\alpha} \Bar{\alpha}_{ij} -\alpha_{ij}) \bra{0}^a U^{ij}\ket{0}^a) \right\| \\
      \leq & \sqrt{d_r d_c} \cdot \left( \max_{i,j: A^{ij}\neq 0} \left\|\alpha_{ij} \bra{0}^a U^{ij}\ket{0}^a - A^{ij} \right\| + \max_{i,j: A^{ij}\neq 0} \left\|(\hat{\alpha} \Bar{\alpha}_{ij} - \alpha_{ij} ) \bra{0}^a U^{ij}\ket{0}^a \right\| \right) \\
      \leq & 2 \sqrt{d_r d_c} \varepsilon.  
\end{aligned}
\end{equation*}
Implementing $G_d$ and the controlled rotation operator requires the indicated additional gate and ancilla qubit complexities (see explanation in the proof of \Cref{lem:sparse}).
\end{proof}

This lemma does not require sparsity to work. It generalizes \Cref{lem:sparse} (also Lemma 48 of \cite{gilyen2019quantum}) for block-encoding matrices with oracle access to entries. If $\alpha_{ij}$ are the same for all blocks, we can omit the oracle $\mathcal{O}_{\alpha}$ and the controlled rotation step, hence only $t+a+2$ ancilla qubits are needed.

\section{Detailed analysis of block-encodings of kernel matrices}

\subsection{Circuit complexity analysis of \Cref{sec:polykernel}}\label{app:detail-encode}
In this section, we provide detailed analysis on circuit and query complexities of the block-encoding procedure presented in \Cref{sec:polykernel}. Our problem is to optimally block-encode the kernel matrix $\mathbf{K}=(k(x_i,x_j))_{i,j=0}^{N-1}$, where $x_i=i$, $N=2^L$, and $k(x,x')=|x-x'|^{-p}$ is a polynomially decaying kernel, given the oracle
\begin{equation}
    \mathcal{O}_{k}: \ket{i}\ket{j}\ket{0}^{\otimes b} \rightarrow \ket{i}\ket{j}\ket{\Tilde{k}(x_i,x_j)},
\end{equation}
where $\Tilde{k}(x_i,x_j)$ is the (exact) $b$-qubit description of $k(x_i,x_j)$.

As a reminder, our procedure uses the hierarchical splitting in \Cref{fig:splitting}.
% which is also copied below. 
% \begin{center}
%     \includegraphics[width=.8\textwidth]{splitting.pdf}
% \end{center}
Specifically, at level $\ell$ of the hierarchy, each admissible block has size $2^{L-\ell}$. In addition, the maximum entry of an admissible block in $\mathbf{K}^{(\ell)}$ is bounded by $d_{\min}^{-p}= 2^{-(L-\ell)p}$. Thus, an admissible block at level $\ell$ can be $(\alpha_{\ell},L-\ell+1,\frac{\varepsilon}{3})$-block-encoded via the naive procedure for dense matrices (\Cref{lem:densenaive} or \Cref{remark:densenaive}), where $\alpha_{\ell} = 2^{(L-\ell)(1-p)}$. From here, each $\mathbf{K}^{(\ell)}$, which is a block-sparse matrix of sparsity 3, can be $(3\alpha_{\ell},L+3,\varepsilon)$-block-encoded via the procedure for block-sparse matrices (\Cref{lem:block-sparse}). The adjacent interaction part $\mathbf{K}_{ad}$, which is a tridiagonal matrix, can be $(3,L+3,\varepsilon)$-block-encoded with \Cref{lem:sparse}. Finally, we use the procedure for linear combination of block-encoded matrices (\Cref{lem:linear_comb}) to obtain a block-encoding of $\mathbf{K} = \sum_{\ell=2}^{L} \mathbf{K}^{(\ell)} + \mathbf{K}_{ad}$ as $U = \sum_{\ell=2}^{L} 3 \alpha_{\ell} U^{(\ell)} +  3 U_{ad}$, where the $U$'s denote the corresponding block-encodings. This step requires a $(\log L)$-qubit state preparation pair (\Cref{def:preppair}) which prepares the coefficients $[3,3\alpha_2,\hdots, 3\alpha_L]$ in the linear combination. This entire procedure results in a $(\alpha, L +\log L + 3,\alpha \varepsilon)$-block-encoding of the \emph{exact} kernel matrix $\mathbf{K}$, where,

\begin{equation}
    \alpha = \begin{cases} 3(1+ \frac{2^{1-p} -N^{1-p}}{2^{1-p}- 4^{1-p}}) \leq O(1)  & \text{ if } p >1 \\
     3\log N  & \text{ if } p = 1\\
     3(1+ \frac{N^{1-p}-2^{1-p}}{4^{1-p}-2^{1-p}}) & \text{ if } p < 1\\
    \end{cases}.
    \label{eq:alpha}
\end{equation}

In the main text, we have shown that this procedure results in an optimal block-encoding. That is, in the obtained $(\alpha, \log N + \log \log N + 3, \varepsilon)$-block-encoding, the normalization factor $\alpha$ is exactly $\Theta(\|\mathbf{K}\|)$. Here, we show why exactly two queries to $\mathcal{O}_k$ and $O(\operatorname{polylog}(\frac{N}{\varepsilon}))$ extra one- and two-qubit gates are needed by carefully combining \Cref{lem:block-sparse}, \Cref{lem:linear_comb}, and \Cref{lem:densenaive}.

\paragraph{Resources} We use 4 registers $A$ ($\log L$ qubits), $B$ ($L+1$ qubits), $C$ ($L+1$ qubits), $D$ (single qubit). For simplicity, assume $\log L$ is an integer. We want to construct a $(2L+\log L +3)$-qubit unitary $U$ such that, for any $0\leq m,n \leq 2^L-1$, we have
\begin{equation}
    \bra{0}_A \bra{0}_B \bra{m}_C \bra{0}_D U \ket{0}_A \ket{0}_B \ket{n}_C \ket{0}_D = \frac{k_{mn}}{\alpha}.
    \label{eq:goal}
\end{equation}
Here, although $0 \leq  m,n \leq 2^L-1$, we allocate $(L+1)$ qubits to the registers $B,C$ for later use.

\paragraph{State preparation pair} 
Observe that \Cref{lem:linear_comb} requires a state preparation pair $(P_L,P_R)$ (\Cref{def:preppair}) that prepares the coefficient vector $\boldsymbol{\beta} = [3,3\alpha_2,\hdots, 3\alpha_L]$ (the first entry being $\alpha_{ad}$ of $\mathbf{K}_{ad}$). Note that $\|\boldsymbol{\beta}\|_1=\alpha$, the overall normalization factor. These operators are $\log L=\log \log N $ qubits, hence we assume they are provided for now (later we will show how to efficiently construct them). In particular, we choose $P_L=P_R$ and
\begin{equation}
    P_R \ket{0}_A = \frac{1}{\sqrt{\alpha}}\sum_{\ell=0}^{L-1} \sqrt{\beta_{\ell}} \ket{\ell}_A.
    \label{eq:prepare}
\end{equation}

\paragraph{Implementing the ``right'' part}
Suppose the registers are in the state
\begin{equation}
    \ket{0}_A \ket{0}_B \ket{n}_C \ket{0}_D.
\end{equation}
We apply $P_R$ to get (omitting the register names for ease of notation)
\begin{equation}
    \frac{1}{\sqrt{\alpha}}\sum_{\ell=1}^{L} \sqrt{\beta_{\ell-1}} \ket{\ell-1} \ket{0} \ket{n} \ket{0}.
    \label{eq:PR}
\end{equation}
Next, we take care of \Cref{lem:block-sparse} for the block-sparse levels $\mathbf{K}^{(\ell)}$. Conditioned on register A being $\ell-1$, where $2\leq \ell \leq L$, we focus on a particular $\ell$. Define the operator $D_3: \ket{0^{L+1} }\rightarrow \frac{1}{\sqrt{3}}\sum_{c=1}^{3} \ket{0}\ket{c}\ket{0^{L-2}}$. Apply $D_3$ on the register $B$ we get
\begin{equation}
    \frac{1}{\sqrt{3}} \sum_{c=1}^{3}\ket{\ell-1} \underbrace{\ket{0}\ket{c}\ket{0^{L-2}}}_{\text{register } B} \ket{n} \ket{0}.
\end{equation}

Next, we apply the controlled block-row index oracle $\mathcal{O}_r$, which, based on $\ell, c, n$, computes the block-row index $m_b$ of the $c$-th non-zero admissible block in $\mathbf{K}^{(\ell)}$ that intersects column $n$ (there are at most $3$ such blocks since $\mathbf{K}^{(\ell)}$ is 3-block-sparse) and writes the result $m_b$ to the first $\ell+1$ qubits of register $B$. Note that this computation is easy and reversible. We denote this set of non-zero block-row index $m_b$ as $I(\ell,n)$. If there are less than $c$ such blocks, $\mathcal{O}_r$ writes $c+2^{\ell}$ instead (this choice ensures the oracle is reversible and also the first qubit of $\ket{m_b}$ acts as a flag that will be useful later).
\begin{equation}
    \frac{1}{\sqrt{3}} \sum_{m_b \in I(\ell,n)} \ket{\ell-1} \underbrace{ \ket{m_b}\ket{0^{L-\ell}}}_{\text{register } B} \ket{n} \ket{0}.
\end{equation}

Next, we take care of the steps in \Cref{lem:densenaive}. Also conditioned on register $A$ being $\ell-1$, where $2\leq \ell \leq L$, we apply Hadamard gates on the last $L-\ell$ qubits of register $B$
\begin{equation}
    \frac{1}{\sqrt{3}} \frac{1}{\sqrt{2^{L-\ell}}} \sum_{m_b \in I(\ell,n)} \sum_{m_i=0}^{2^{L-\ell}-1} \ket{\ell-1} \underbrace{ \ket{m_b}\ket{m_i}}_{\text{register } B} \ket{n} \ket{0}.
\end{equation}

Observe that when the first qubit of $\ket{m_b}$ is zero, $\ket{m_b}\ket{m_i}$ together form a valid row index $m<2^L$. If this qubit is one, then $\ket{m_b}\ket{m_i}$ is not a valid row index and we can simply take $k_{mn}=0$ when querying the matrix entry oracle $\mathcal{O}_k$. In particular, we query the oracle $\mathcal{O}_k$ on registers $B,C$ to obtain $\ket{\Tilde{k}_{mn}}$ in another appended $b$-qubit register. Then, we perform the controlled rotation to rotate register $D$ as: $\ket{0} \rightarrow \left( \frac{k_{mn}}{k^{(\ell)}_{\max}}\ket{0} +\sqrt{1-\left|\frac{k_{mn}}{k^{(\ell)}_{\max}}\right|^2} \ket{1} \right)$, where $k^{(\ell)}_{\max}$ is the maximum entry in the admissible blocks at level $\ell$, which we have shown to be $2^{-(L-\ell)p}$. 
Finally, we call $\mathcal{O}_k$ one more time to uncompute and discard the appended register. This gives
\begin{equation}
    \frac{1}{\sqrt{3}} \frac{1}{\sqrt{2^{L-\ell}}} \sum_{m_b \in I(\ell,n)} \sum_{m_i=0}^{2^{L-\ell}-1} \ket{\ell-1} \underbrace{ \ket{m_b}\ket{m_i}}_{\text{register } B} \ket{n} \left( \frac{k_{mn}}{k^{(\ell)}_{\max}}\ket{0} +\sqrt{1-\left|\frac{k_{mn}}{k^{(\ell)}_{\max}}\right|^2} \ket{1} \right)
    \label{eq:querylevels}
\end{equation}
for $2 \leq \ell \leq L$.

On the other hand, conditioned on $\ell=1$ (register $A$ being $\ket{0}$), we take care of the adjacent part $\mathbf{K}_{ad}$ rather simply. Apply the $(2L+2)$-qubit operator $\mathcal{O}_{ad}$ which performs the map $\ket{0}\ket{n} \rightarrow \frac{1}{\sqrt{3}} \sum_{i=-1}^{1} \ket{n+i}\ket{n}$ for $0 \leq n <  2^L$ (if $n+i=-1$ then it writes to the first register any value that is not a valid row index, e.g., $2^{L+1}-1$) on registers $B, C$. Then, we perform the oracle query, controlled rotation, and uncomputation similarly as done above to obtain
\begin{equation}
    \frac{1}{\sqrt{3}} \sum_{i=-1}^{1} \ket{0}_A \ket{n+i}_B \ket{n}_C \left( k_{n+i,n}\ket{0} + \sqrt{1-|k_{n+i,n}|^2}\ket{1}\right),
    \label{eq:queryAdjacent}
\end{equation}
where $k_{n+i,n}=0$ if $n+i$ is an invalid row index.

To summarize, the entire procedure from \autoref{eq:PR} to \autoref{eq:querylevels} (and \autoref{eq:queryAdjacent}), denoted by $U_R$, performs the following
\begin{equation}
\small
\begin{aligned}
    U_R P_R \ket{0} \ket{0} \ket{n} \ket{0}=  &\frac{1}{\sqrt{3\alpha}}  \left[ \sqrt{\beta_0} \sum_{i=-1}^{1} \ket{0} \ket{n+i} \ket{n} \left( k_{n+i,n}\ket{0} + \sqrt{1-|k_{n+i,n}|^2}\ket{1}\right) \right. \\
     & \left. + \sum_{\ell=2}^{L} \frac{\sqrt{\beta_{\ell-1}}}{\sqrt{2^{L-\ell}}} \sum_{m_b \in I(\ell,n)} \sum_{m_i=0}^{2^{L-\ell}-1} \ket{\ell-1} \underbrace{ \ket{m_b}\ket{m_i}}_{\text{equiv. to some } m} \ket{n} \left( \frac{k_{mn}}{k^{(\ell)}_{\max}}\ket{0} +\sqrt{1-\left|\frac{k_{mn}}{k^{(\ell)}_{\max}}\right|^2} \ket{1} \right) \right],
\end{aligned}
\label{eq:rightpart}
\end{equation}
where $\beta_0=3$ is the normalization factor of the adjacent level $\mathbf{K}_{ad}$ and $\beta_{\ell-1}=3\alpha_{\ell}$ is the normalization factor of the hierarchical level $\mathbf{K}^{(\ell)}$. Importantly, notice that register $B$ has full support on all $2^L$ possible values of row index $m$ in this superposition.

\paragraph{Implementing the ``left'' part} We would like to transform the state
\begin{equation}
    \ket{0}\ket{0}\ket{m}\ket{0}
\end{equation}
to obtain similar result to that of the ``right'' procedure above. The procedure follows the same steps, except
\begin{itemize}
    \item $P_R$ needs to be replaced by $P_L$.
    \item $\mathcal{O}_r$ needs to be replaced by $\mathcal{O}_c$, the controlled block-column index oracle which, based on $\ell, m$, computes the block-column indices $n_b$ of the non-zero admissible blocks in $\mathbf{K}^{(\ell)}$ that intersect row $m$ (there are at most $3$ such blocks since $\mathbf{K}^{(\ell)}$ is 3-block-sparse) and writes the result $n_b$ to the first $\ell+1$ qubits of register $B$. We denote this set of non-zero block-column indices $n_b$ as $J(\ell,m)$.
    \item There is no need to query the entry oracle $\mathcal{O}_k$ (nor the accompanied controlled rotation).
    \item Registers $B$ and $C$ are swapped at the end.
\end{itemize}
Using similar notation to that of \autoref{eq:rightpart}, the above modifications give
\begin{equation}
\begin{aligned}
    \swap_{B,C} U_L P_L \ket{0} \ket{0} \ket{m} \ket{0} =  \frac{1}{\sqrt{3 \alpha}} & \left[ \sqrt{\beta_0} \sum_{i=-1}^{1} \ket{0} \ket{m} \ket{m+i}  \ket{0} \right. \\
     & \left. + \sum_{\ell=2}^{L} \frac{\sqrt{\beta_{\ell-1}}}{\sqrt{2^{L-\ell}}} \sum_{n_b \in J(\ell,m)} \sum_{n_j=0}^{2^{L-\ell}-1} \ket{\ell-1} \ket{m} \underbrace{ \ket{n_b} \ket{n_j}}_{\text{equiv. to some } n}  \ket{0} \right].
\end{aligned}
\label{eq:leftpart}
\end{equation}

\paragraph{Combining both parts} Combining \autoref{eq:rightpart} with \autoref{eq:leftpart} and noting that $\beta_0=3, \beta_{\ell-1}=3\cdot 2^{(L-\ell)(1-p)}$ and $k^{(\ell)}_{\max} = 2^{-(L-\ell)p}$, we get
\begin{equation}
    \bra{0} \bra{0} \bra{m} \bra{0}  P_L^{\dagger} U_L^{\dagger} \swap_{B,C} U_R P_R \ket{0} \ket{0} \ket{n} \ket{0} = \frac{k_{mn}}{\alpha},
    \label{eq:bothpart}
\end{equation}
which is exactly what we set out to prove (\autoref{eq:goal}).

\paragraph{Circuit complexity} As seen above, this block-encoding only uses two calls to $\mathcal{O}_k$, and one call to each of $\mathcal{O}_r$, $\mathcal{O}_c$ and $P_R$, $P_L$. To achieve an overall error bound of $\varepsilon$, it suffices for the controlled rotation on $\ket{\Tilde{k}_{mn}}\ket{0}$ to have error $O(\frac{\varepsilon}{N})$. This can be done with $O(b, \operatorname{polylog}(\frac{N}{\varepsilon}))$ extra qubits and one- and two-qubit gates (see explanation in proof of \Cref{lem:sparse}).

\paragraph{Constructing state preparation pair} $P_L$ and $P_R$ act on only $\log \log N$ qubits, hence were assumed given as unitaries. If this is not the case, we can still construct them (\autoref{eq:prepare}) as block-encoded operators with only an extra normalization factor of $\log N$ as follows. Let
\begin{equation}\begin{aligned}
    \mathcal{P}: \ket{0^{\log L}}\ket{0} &\rightarrow \frac{1}{\sqrt {L}} \sum_{\ell=0}^{L-1} \ket{\ell}\ket{0} \\
    & \rightarrow \frac{1}{\sqrt {L}} \sum_{\ell=0}^{L-1} \ket{\ell} \ket{\Tilde{\beta}_{\ell}}\ket{0} \qquad \text{(query $\beta_{\ell}$ into an appended register)}\\
    & \rightarrow \frac{1}{\sqrt {L}} \sum_{\ell=0}^{L-1} \ket{\ell} \left(\sqrt{\frac{\beta_{\ell}}{\hat{\beta}}} \ket{0} + \sqrt{1-\left|\frac{\beta_{\ell}}{\hat{\beta}} \right|}\ket{1} \right),\\
\end{aligned}
\end{equation}
where $\beta_0 = 3$ and $\beta_{\ell-1}=3 \alpha_{\ell}=3\cdot 2^{(L-\ell)(1-p)}$ for $2 \leq \ell \leq L$ and $\hat{\beta}=\max_{\ell} |\beta_{\ell}|$. As explained in \Cref{lem:sparse}, the controlled rotation can be implemented with error bound $\varepsilon$ using $O(\operatorname{polylog}(\frac{L}{\varepsilon}))$ extra one- and two-qubit gates. Comparing with \autoref{eq:prepare} we see that $\mathcal{P}$ is a $(\sqrt{\frac{\hat{\beta} L}{\alpha}},1,\varepsilon)$-block-encoding of $P_R$. Thus, using \Cref{lem:multiply-encode} (multiplying block-encoded matrices) we can achieve \autoref{eq:bothpart} with just two more ancilla qubits. The RHS of \autoref{eq:bothpart} then becomes $\frac{k_{mn}}{\alpha} \left(\frac{\hat{\beta} L}{\alpha}\right)^{-1}$. In other words, the normalization factor will become a factor of $\frac{\hat{\beta} L}{\alpha}$ larger than the optimal block-encoding. We can bound this extra factor rather easily. As a reminder, $\alpha$ is given in \autoref{eq:alpha}.
\begin{itemize}
    \item For $p>1$: $\hat{\beta}=3$ and $\alpha \geq 3$, thus $\frac{\hat{\beta} L}{\alpha} \leq L = \log N$.
    \item For $p=1$: $\hat{\beta}=3$ and $\alpha = 3 \log N$, thus $\frac{\hat{\beta} L}{\alpha}=1$.
    \item For $p<1$: $\hat{\beta}=3\cdot 2^{(L-2)(1-p)}=O(N^{1-p})$ and $\alpha = \Omega(N^{1-p})$, thus $\frac{\hat{\beta} L}{\alpha} \leq O(L) = O(\log N)$.
\end{itemize}
Therefore, this extra factor only increases the runtimes of \Cref{lem:forward} (applying block-encoded matrices) and \Cref{lem:inv} (applying inverse of block-encoded matrices) by a factor of $O(\log N)$.

\subsection{Variants of hierarchical splitting}\label{app:variantHmatrix}
% inverse multiquadric
% \qn{copied from main text}
The form of the hierarchical splitting can be slightly modified for other types of decaying kernels. For example, consider kernels of the form $k(x,x')= \frac{1}{|(x-x')+c|^p}$ on the domain $\Omega=[0,N-1)$ and $-N<c<N$ is an integer. In this case, we can use a ``shifted'' hierarchical splitting to obtain an optimal block-encoding of the kernel matrix $\mathbf{K}=(|i-j+c|^{-p})_{i,j=0}^{N-1}$, in which the hierarchy is shifted in the horizontal direction (along the rows, see \Cref{fig:skew}.) Similarly, for kernels of the form $k(x,x')= \frac{1}{(|x-x'|-c)^p}$ ($0\leq c<N$), we can apply a hierarchical splitting shifted along the skew-diagonal direction, such that the upper and lower triangular halves of the matrix are shifted in opposite directions.
\begin{figure}[H]
    \centering
    \includegraphics[width=0.99\textwidth]{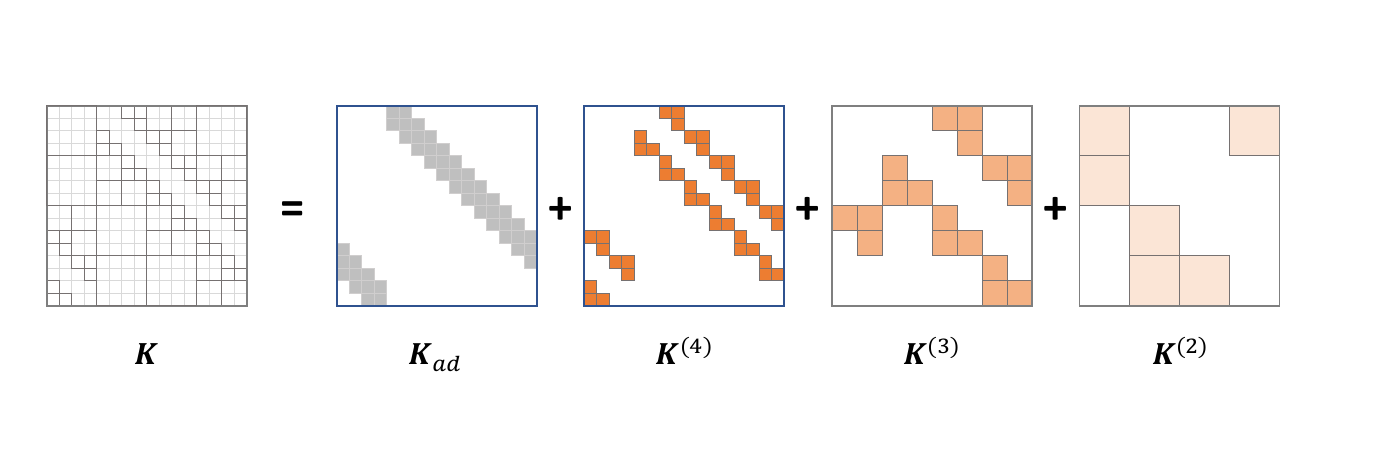}
    \caption{Hierarchical splitting corresponding to the kernel $k(x,x')=\frac{1}{|x-x'+4|^p}$ with uniform distribution of particles, which is obtained by shifting the splitting in \Cref{fig:splitting}C 4 units to the right. Note that the fine splittings in the first 4 columns are actually only needed when the sum $x-x'+4$ is modulo $N$.}
    \label{fig:skew}
\end{figure}

\subsection{Adaptive algorithm for non-uniform particle distribution}\label{app:adaptive}
In particle simulation tasks, if the distribution of particles is not approximately uniform, an adaptive hierarchical splitting can be employed. Let the smallest interparticle spacing be $N^{-c}$, where $N$ is the number of particles. Then, we can scale the mesh such that the interparticle distance is at least $1$ and the mesh size is $N^c$. The added mesh sites are treated as zero-mass particles. The kernel matrix has size $N^c \times N^c$ and the hierarchical splitting constructed on this mesh has $c \log N$ levels. The entire block-encoding procedure presented in \Cref{sec:main} applies, with only an extra factor of $\operatorname{poly}(c)$ in the normalization factor and the required resources of the block-encoding. The optimality of the normalization factor of this adaptive block-encoding is harder to state precisely, since it depends on the distribution of particles which is not determined a priori.

\subsection{Generalization of quantum hierarchical block-encoding}\label{app:generalized}
Here, we provide a block-encoding procedures in which one has direct access to the structure of the entry magnitudes, thus generalizing hierarchical block-encoding procedure presented in \Cref{sec:polykernel} and \Cref{app:detail-encode}.

First, we present a procedure for preparing a quantum state $\ket{x} \in \mathbb{C}^{N}$, where $N=2^L$, with high probability of success. For any $0\leq \ell < L-1$, let $I_{\ell}=\{i:2^{-(\ell+1)}<|x_i|\leq 2^{-\ell}\}$ and $I_{L-1}=\{i:|x_i|\leq 2^{-(L-1)}\}$. Also, let $n_{\ell}=|I_{\ell}|$. Apart from the oracle $\mathcal{O}_x: \ket{i}\ket{0} \rightarrow \ket{i} \ket{\Tilde{x}_i}$, our procedure assumes the following oracle is provided:
\begin{equation}
    \mathcal{O}_{index}: \ket{\ell}\ket{j} \rightarrow \ket{0} \ket{f(\ell,j)}, \qquad  0\leq j < n_{\ell},
    \label{eq:O_index}
\end{equation}
where the first register is $\log L$ qubits and the second register is $L= \log N$ qubits, and $f(\ell,j)$ is an efficiently computable and \emph{reversible} function that computes the index of the $j$-th element in the set $I_{\ell}$ (according to some fixed ordering, e.g. top to bottom if $\ket{x}$ is treated as a column vector) if $0\leq j < n_{\ell}$, or else the oracle writes $N +(j-n_{\ell}+1) + \sum_{l=0}^{\ell-1}(N-n_{l})$ as a bit string on the concatenation of the two registers (so that the first register will be non-zero). This ensures the oracle is reversible. An example of such efficient and reversible $f(\ell,j)$ is in regular $\mathcal{H}$-matrix splitting. Our procedure is as follows:
\begingroup
\allowdisplaybreaks
\begin{align*}
\small
    \ket{0^{\log L}}\ket{0} &\rightarrow  \sum_{\ell=0}^{L-1} \frac{\sqrt{n_{\ell}} 2^{-\ell}}{\beta}
    \ket{\ell}\ket{0}, \qquad \text{where } \beta = \left(\sum_{\ell} n_{\ell} 2^{-2\ell} \right)^{1/2} \\
    &\rightarrow  \sum_{\ell=0}^{L-1} \frac{\sqrt{n_{\ell}} 2^{-\ell}}{\beta}
    \ket{\ell} \sum_{j=1}^{n_{\ell}} \frac{\ket{j}}{\sqrt{n_{\ell}}} \qquad \text{(Hadamards controlled on $\ell$)}\\
    &\rightarrow  \sum_{\ell=0}^{L-1}  \sum_{j=1}^{n_{\ell}} \frac{ 2^{-\ell}}{\beta}
    \ket{0} \ket{f(\ell,j)}\qquad \text{(call oracle $\mathcal{O}_{index}$)}\\
    & \rightarrow  \sum_{\ell, j} \frac{ 2^{-\ell}}{\beta}
    \ket{0} \ket{f(\ell,j)} \left(\frac{x_{f(\ell,j)}}{2^{-\ell}} \ket{0} + \sqrt{1-\left|\frac{x_{f(\ell,j)}}{2^{-\ell}} \right|^2}\ket{1} \right) \\
    & \qquad \qquad \qquad \qquad \text{(query $x_{f(\ell,j)}$ and rotate ancilla qubit)}\\
\end{align*}
\endgroup
Above, we have ignored the terms when $j$ is invalid ($j\geq n_{\ell}$) in the sum because we can simply set $x_{f(\ell,j)}=0$ for those cases. Furthermore, observe that the superposition has support on all indices. Thus, conditioning on the ancilla being zero, we obtain the desire state $\ket{x}$. The probability of successfully post-selecting  is $\frac{1}{\beta^2}$, where
\begin{equation}
    \beta^2 = \sum_{\ell=0}^{L-1} n_{\ell} 2^{-2\ell} = \sum_{\ell=0}^{L-2} n_{\ell} 2^{-2\ell} + n_{L-1}2^{-2(L-1)} \leq  4 \sum_{i}|x_i|^2 + \frac{4}{N} \leq 8.
\end{equation}
We also assumed that the first operation in the procedure, which maps $\ket{0^{\log L}}$ to a weighted superposition over $\ket{\ell}$, can be done perfectly. In fact, it can only be block-encoded with an extra normalization factor of $\sqrt{L}$: for a vector $\ket{y}\in \mathbb{C}^{L}$ one can simply do
\begin{equation*}
\begin{aligned}
    \ket{0^{\log L}}\ket{0} &\rightarrow  \sum_{\ell=0}^{L-1} \frac{1}{\sqrt{L}}
    \ket{\ell}\ket{0}, \\
    &\rightarrow \sum_{\ell=0}^{L-1} \frac{1}{\sqrt{L}}
    \ket{\ell} \left(y_{\ell} \ket{0} + \sqrt{1- |y_{\ell}|^2}  \ket{1} \right)
    \qquad \text{(query $y_{\ell}$ and rotate ancilla qubit)}
\end{aligned}
\end{equation*}
When including this step into the procedure above, the probability of successfully post-selecting $\ket{00}$ and obtaining $\ket{x}$ is $\frac{1}{L \beta^2}\geq \frac{1}{8 \log N}$.

Intuitively, our procedure exploits the structure of the amplitudes such that in the controlled rotation step, the ancilla has a large component in the good subspace $\ket{0}$.

The procedure to block-encode a matrix $A$ follows similarly. For simplicity, assume $A$ is Hermitian and $N^{-1} \leq |A_{ij}|\leq  1$. Let $I_{\ell}(j)$ be the set of row indices of the entries in column $j$ whose values are in the range $[2^{-(\ell+1)},2^{-\ell})$. Let $n_{\ell}(j)=|I_{\ell}(j)|$, let $n_{\ell}= \max_{j} n_{\ell}(j)$. Let $\gamma$ be such that $\min_{j} n_{\ell}(j) \geq \gamma  n_{\ell}$. Our procedure, generalizing \Cref{app:detail-encode}, proceeds as follows:
\paragraph{The right part} For any column index $j$:
\begin{equation}
\begin{aligned}
    \ket{0^{\log L}}\ket{0} \ket{j}\ket{0} &\rightarrow  \sum_{\ell=0}^{L-1} \frac{\sqrt{n_{\ell}2^{-\ell}} }{\beta}
    \ket{\ell}\ket{0}\ket{j}\ket{0}, \qquad \text{where } \beta = \left(\sum_{\ell} n_{\ell} 2^{-\ell} \right)^{1/2} \\
    &\rightarrow  \sum_{\ell=0}^{L-1} \frac{\sqrt{n_{\ell}2^{-\ell}}}{\beta}
    \ket{\ell} \sum_{i=1}^{n_{\ell}} \frac{\ket{i}}{\sqrt{n_{\ell}}} \ket{j}\ket{0} \qquad \text{(Hadamards controlled on $\ell$)}\\
    &\rightarrow  \sum_{\ell=0}^{L-1}  \sum_{i=1}^{n_{\ell}} \frac{ \sqrt{2^{-\ell}}}{\beta}
    \ket{0} \ket{f_j(\ell,i)} \ket{j}\ket{0} \qquad \text{(call oracle $\mathcal{O}_{index}$)}\\
    & \rightarrow  \sum_{\ell; i\in [n_{\ell}(j)]} \frac{ \sqrt{2^{-\ell}}}{\beta}
    \ket{0} \ket{f_j(\ell,i)} \ket{j} \left(\frac{A_{f_j(\ell,i),j}}{2^{-\ell}} \ket{0} + \sqrt{1-\left|\frac{A_{f_j(\ell,i),j}}{2^{-\ell}} \right|^2}\ket{1} \right) \\
    & \qquad \qquad \qquad \qquad \text{(query $A_{f_j(\ell, i),j}$ and rotate ancilla qubit)}\\
\end{aligned}
\label{eq:aaa}
\end{equation}
Here, $\mathcal{O}_{index}$ is similar to that in \autoref{eq:O_index}; it returns the row index, $f_j(\ell, i)$, of the $i$-th element in $I_{\ell}(j)$. We ignore the terms when $i$ is invalid ($i \geq n_{\ell}(j)$) as we can set $A_{f_j(\ell,i),j}=0$ for those terms. Importantly, the above superposition has support over all row indices.
\paragraph{The left part} Similarly, we can construct a unitary that queries  all the column indices
\begin{equation}
\begin{aligned}
    \ket{0^{\log L}}\ket{0} \ket{i} \ket{0} \rightarrow \sum_{\ell; j \in [n_{\ell}(i)]} \frac{ \sqrt{2^{-\ell}}}{\beta}
    \ket{0} \ket{i} \ket{g_i(\ell,j)} \ket{0},
\end{aligned}
\label{eq:bbb}
\end{equation}
where $g_i(\ell, j)$ is the column index of the $j$-th element in $I_{\ell}(i)$.

Combining \autoref{eq:aaa} and \autoref{eq:bbb} we obtain a unitary $U$ that is block-encodes $A$ with a normalization factor of $\beta^2$. If we take into consideration of first step \autoref{eq:aaa} (i.e. preparing $\sum_{\ell=0}^{L-1} \frac{\sqrt{n_{\ell}2^{-\ell}} }{\beta}\ket{\ell}$) as did in the quantum state preparation procedure above, then the normalization factor is $\beta^2 \log N$. We now compare $\beta^2$ with the operator norm $\|A\|$. If the entries of $A$ are real and non-negative, we have that
\begin{equation}
    \|A\| \geq \|A \ket{\mathbf{1}}\| \geq \sum_{\ell} \min_j n_{\ell}(j) 2^{-(\ell+1)} \geq \frac{\gamma}{2} \beta^2.
\end{equation}
Thus, compared to the optimal block-encoding (one that has normalization factor equal to $\|A\|$), the presented block-encoding for $A$ is worse by at most a factor of $2 (\log N )/ \gamma$. In other words, if $1/\gamma = O(\operatorname{polylog}(N))$, then the normalization factor is only a factor of $\operatorname{polylog}(N)$ worse than that of the optimal block-encoding. As before (\cref{app:detail-encode}), the numbers of elementary gates for the controlled rotation step and ancilla qubits are $O(\operatorname{polylog}(\frac{N}{\varepsilon}))$.

% \subsection{Rectangular hieararchical matrices}
% Only need to take care of the admissible blocks being rectangular, which is simple b/c Lemma A.5 and the low-rank lemma apply to rectangular matrices just fine.
% \qn{ref section 2 to here.}

\subsection{Condition number of polynomially and exponentially decaying kernels}\label{app:conditionnumber}
\begin{figure}[H]
    \centering
    \includegraphics[width=0.7\textwidth]{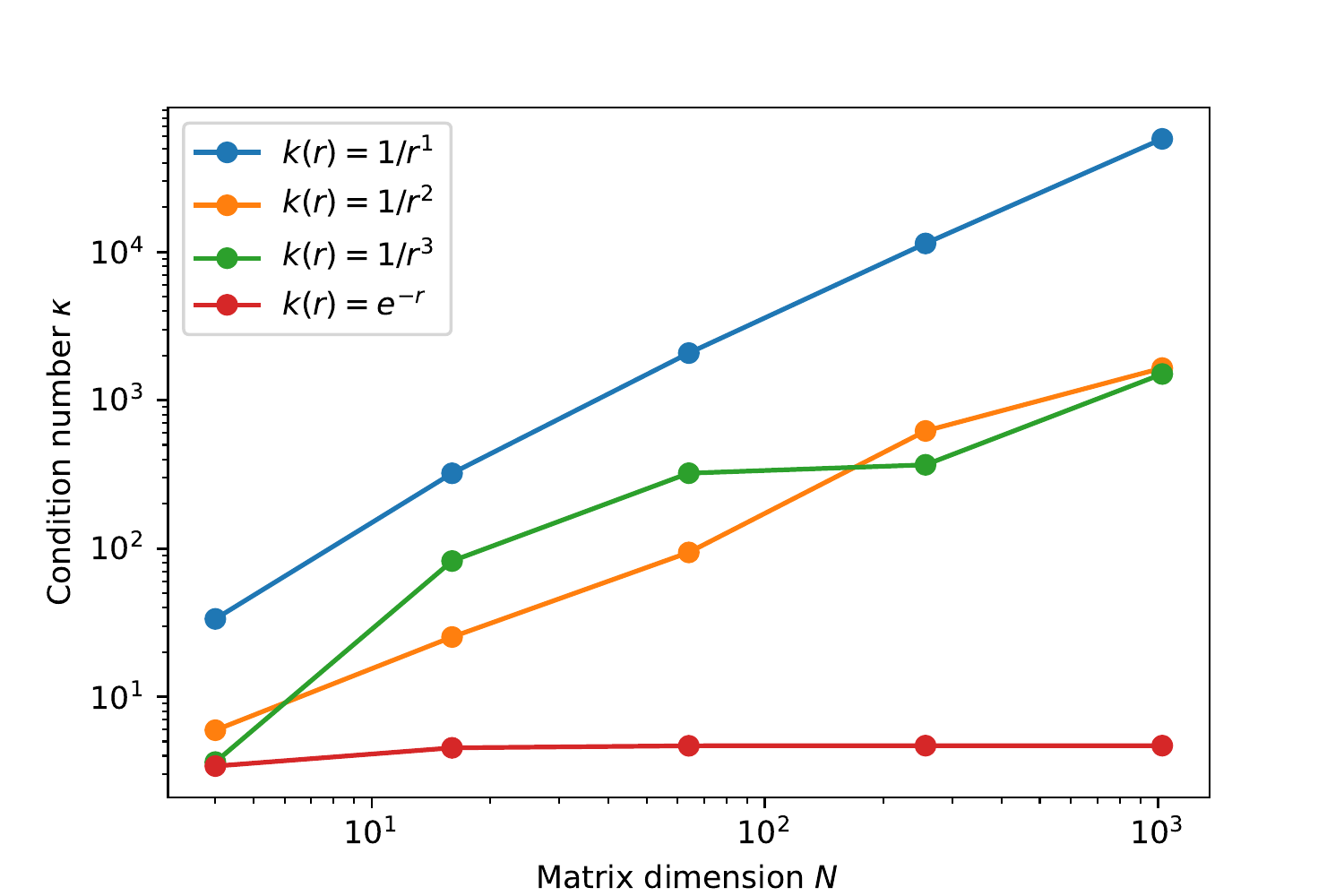}
    \caption{The condition number of polynomially decaying kernels grow polynomially with the matrix dimension. Whereas, the condition number of the exponentially decaying kernel converges. Here, the kernel matrices take the form $\mathbf{K}=(k(|x_i-x_j|))_{i,j=0}^{N-1}$, where $x_i=i$, $k(|x-x'|)=k(r)$, and $k(0)=0$ (no self-interactions). We remark that, the condition number can be dramatically improved if the self-interacting terms are sufficiently large. For example, for $k(0)=2$ we found that the condition number for $k(r)=1/r$ grows only logarithmically in $N$, while those of the other kernels scale as $O(1)$. This can be explained using similar calculations to \autoref{eq:boundkappa}.}
    \label{fig:conditionumber}
\end{figure}

% Thus, appending two ancilla qubits into \autoref{eq:bothpart} and replacing $P_R, P_L$ by $\mathcal{P}$ (which uses the first of the added qubits) and $\mathcal{P}'$ (which uses the other added qubit) we obtain with just two more ancilla qubits the the RHS as
% \begin{equation}
%   \bra{0}\bra{0} \bra{0} \mathcal{P}^{\dagger} \mathcal{P} \ket{0}\ket{0}\ket{0}
% \end{equation}

\subsection{Kernels for which the naive block-encoding approach (\Cref{lem:densenaive}) works optimally} \label{app:naiveoptimal}
We provide examples of kernels  that are studied in classical $\mathcal{H}$-matrix or fast multipole method literature where the naive block-encoding approach of \Cref{lem:densenaive} (or \Cref{remark:densenaive}) can be directly used to produce an optimal block-encoding of their kernel matrices.

\paragraph{The log kernel} $k(x,x')=\log |x-x'|$ (chapter 1, \cite{Hmatrixnotes}). We use the same setup as in the main text. Let the kernel matrix be $\mathbf{K}= (k(x_i,x_j))_{i,j=0}^{N-1}$, where $x_i=i$, $N=2^L$, and $k(x,x)=0$. Directly applying \Cref{lem:densenaive} yields a $(N\log N,\log N + 1, \varepsilon)$-block-encoding using only $O(\operatorname{polylog}\frac{N}{\varepsilon})$ additional resources. To show that this block-encoding  is optimal, we bound the operator norm from below: $\|\mathbf{K}\|\geq \|\mathbf{K}\ket{\mathbf{1}}\| \geq  \int_{1}^{N/2} \log x dx = \Omega(N \log N)$.

\paragraph{The multiquadric kernel} $k(x,x')=\sqrt{c^2+ |x-x'|^2}$ (chapter 2.1, \cite{fmm_shortcourse}). For this, let the kernel matrix be $\mathbf{K}= (k(x_i,x_j))_{i,j=0}^{N-1}$, where $x_i=i/N$, $N=2^L$. Directly applying \Cref{lem:densenaive} yields a $(2N,\log N + 1, \varepsilon)$ using only $O(\operatorname{polylog}\frac{N}{\varepsilon})$ additional resources. The operator norm of $\mathbf{K}$ is bounded as $\|\mathbf{K}\|\geq \|\mathbf{K}\ket{\mathbf{1}}\| \geq  N \int_{0}^{1/2} \sqrt{c^2+x^2} dx = \Omega(N)$, implying that the block-encoding is optimal.
\vspace{0.5cm}

We also list here, omitting the proofs, some typical kernels (which are not considered in classical $\mathcal{H}$-matrix) that can be optimally block-encoded by the naive approach of \Cref{lem:densenaive}. For example, these include polynomially growing kernels $k(x,x')=|x-x'|^{p}$ and polyharmonic radial basis functions $k(x,x')=|x-x'|^{p}\log |x-x'|$. Our numerical experiments indicate that these kernels are numerically low-rank, as shown in \Cref{fig:lowrank-ker}. Furthermore, the first principal component (singular vector corresponding to the largest singular value) is close to uniform (the normalized all-ones vector). This explains why these kernels can be block-encoded by the naive block-encoding.

\begin{figure}
    \centering
    \includegraphics[width=0.75\textwidth]{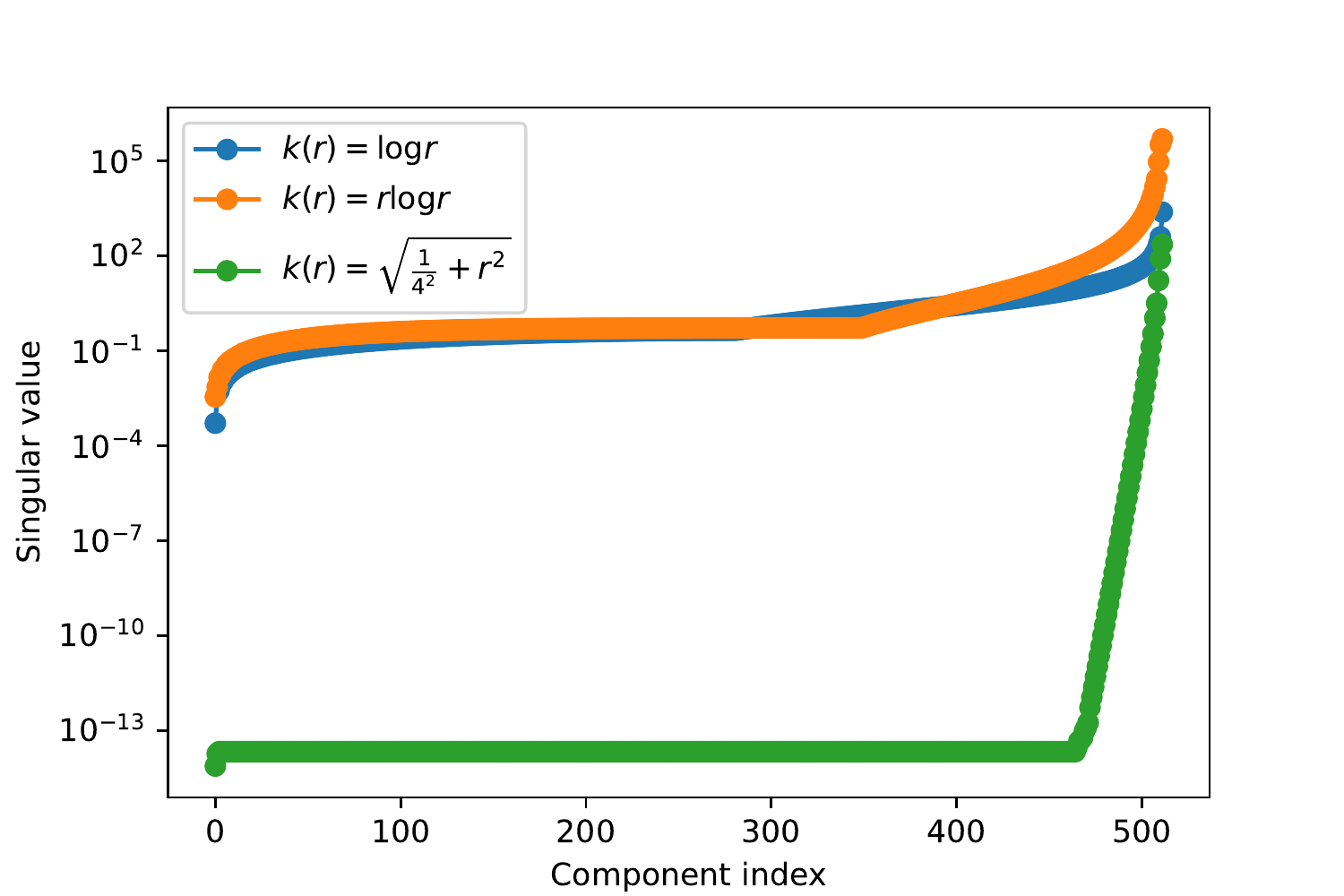}
    \caption{Singular values of the 512 by 512 kernel matrices $\mathbf{K}= (k(x_i,x_j))_{i,j=0}^{N-1}$ of the log, polyharmonic, and multiquadric kernels. Here $r=|x-x'|$, $x_j=j/N$ for the multiquadric kernel ($k(r)= \sqrt{1/4^2+r^2}$), and $x_j=j$ for the other kernels. The plots indicate that all kernels are numerically low-rank.}
    \label{fig:lowrank-ker}
\end{figure}

\section{Classical $\mathcal{H}$-matrices}
\label{app:hmatrix}

\subsection{Low-rank approximation using asymptotic smoothness}
In this section, we show how the asymptotic smoothness condition (\autoref{eq:condition1}) can be used to approximate an admissible block as a low-rank matrix. As a reminder, the asymptotic smoothness condition is
\begin{align*}
    | \partial_{x'}^p k(x,x') | & \leq C p! |x-x'|^{-p} \qquad (\forall p \in \mathbb{N}).
\end{align*}
For instance, the log kernel $k(x,x')=\log (|x-x'|)$ directly holds this property.

Consider the admissible block $\mathbf{K}^{\sigma, \rho}=(k(x_i,x_j))_{i\in \sigma, j\in \rho}$, where $\sigma, \rho \in \mathcal{P}^{(\ell)}$. Let $c_{\rho}$ be the center of $\rho$, we first obtain the $p$-th order Taylor expansion of $k(x,x')$ around $x'=c_{\rho}$
\begin{equation}
    k(x,x') = \sum_{q=0}^{p-1} \frac{(x'-c_{\rho})^q}{q!} \partial_{x'}^q k(x,c_{\rho}) + R(x,x'),
\end{equation}
where the remainder $R(x,x')$ can be bounded using the asymptotic smoothness property as follows
\begin{equation}
    \begin{aligned}
        k(x,x') &= \sum_{q=p}^{\infty} \frac{(x'-c_{\rho})^q}{q!} \partial_{x'}^q k(x,c_{\rho}) \leq C \sum_{q=p}^{\infty} \left(\frac{x'-c_{\rho}}{|x-c_{\rho}|}\right)^q \\
        & \leq C \sum_{q=p}^{\infty} \left(\frac{r_{\rho}}{\operatorname{dist}(\sigma,\rho)}\right)^q  \leq C \sum_{q=p}^{\infty} \left(\frac{r_{\rho}}{2 r_{\rho}}\right)^q = O(2^{-p}).
    \end{aligned}
\label{eq:taylor}
\end{equation}
Here, we have used the admissibility criteria (\Cref{def:admissible}) in the third inequality and that $r_{\sigma} = r_{\rho}$. Thus, $\mathbf{K}^{\sigma, \rho}$ can be approximated with exponentially small error by a rank-$p$ matrix
\begin{equation}
    \mathbf{K}^{\sigma, \rho} \approx \mathbf{\widetilde{K}}^{\sigma, \rho} = \boldsymbol{\Psi}^{\sigma, \rho} \mathbf{D} \left(\boldsymbol{\Phi}^{ \rho}\right)^{\dagger},
\end{equation}
where $\mathbf{D} =\operatorname{Diag}(1/q!)_{q\in P} \in \mathbb{R}^{p\times p}$, $P=\{0,1,\hdots p-1\}$ and
\begin{align}
    & \boldsymbol{\Psi}^{\sigma, \rho} = (\partial_{x'}^q k(x_i,c_{\rho}))_{i\in\sigma, q \in P} \in \mathbb{R}^{|\sigma| \times p}, \\
    & \boldsymbol{\Phi}^{\rho} = ((x_{j}'-c_{\rho})^q)_{j\in\rho, q \in P} \in \mathbb{R}^{|\rho| \times p}.
\end{align}

Next, we give an example of a kernel that does not directly satisfy the asymptotic smoothness condition (\autoref{eq:condition1}), but can still be manipulated to apply the low-rank approximation. Consider the kernel $k(x,x')=|x-x'|^{-2}$. This  kernel does not directly satisfy \autoref{eq:condition1}. Nevertheless, for an admissible block $ \mathbf{K}^{\sigma, \rho}= (k(x_i,x_j))_{i\in \sigma, j\in \rho}$ at level $\ell$, the kernel can be rewritten as 
\begin{equation}
\begin{aligned}
    k \left(x-x^{\prime}\right) &=k\left(x-c_{\rho}+c_{\rho}-x^{\prime}\right) =k\left(\left(x-c_{\rho} \right)\left(1+\frac{c_{\rho} -x'}{x-c_{\rho}}\right)\right) \\
    &=k\left(x-c_{\rho}\right) k\left(1+\frac{c_{\rho}-x^{\prime}}{x-c_{\rho}}\right) 
    =k\left(x-c_{\rho}\right) \sum_{q=0}^{\infty} \frac{k^{(q)}(1)}{q !}\left(\frac{c_{\rho}-x^{\prime}}{x-c_{\rho}}\right)^{q},
\end{aligned}
\end{equation}
where $c_{\rho}$ is the center of $\rho$. In our setting (see \Cref{sec:hmatrices}), $c_{\rho}$ is simply the middle point of the corresponding segment on the domain $\Omega$. Furthermore, noting that $k^{(q)}(1)= (-1)^q (q+1)!$, let
\begin{equation}
\begin{aligned}
    &A_{q}\left(x-c_{\rho}\right)=\frac{k\left(x-c_{\rho}\right)}{\left(x-c_{\rho}\right)^{q}},\\
    &S_{q}\left(c_{\rho}-x^{\prime}\right)= (-1)^q (q+1) \left(c_{\rho}-x^{\prime}\right)^{q}, \\
\end{aligned}
\end{equation}
and
\begin{equation}
    k_p(x-x')= \sum_{q=0}^{p}
    A_{q}\left(x-c_{\rho}\right)
    S_{q}\left(c_{\rho}-x^{\prime}\right).
    \label{eq:grav-series}
\end{equation}
Since $\mathbf{K}^{\sigma,\rho}$ is an admissible block (\Cref{def:admissible}), we have that $\left|\frac{c_{\rho}-x_j'}{x_i-c_{\rho}}  \right| \leq \frac{1}{2}, \forall i \in \sigma, j \in  \rho$. Therefore, 
\begin{equation}
\begin{aligned}
    |k(x_i,x_j) - k_p(x_i,x_j)| &\leq  (\sup_{x,x' \in \Omega} |k(x,x')|) \left( \sum_{q=p+1}^{\infty} (q+1) \left| \frac{c_{\rho}-x'}{x-c_{\rho}}\right|^q \right)\\
    &\leq  \sum_{q=p+1}^{\infty} (q+1)  2^{-q}  = 2^{-p} (2p+4).
\end{aligned}
\end{equation}
Thus, $k(x_i,x_j)$ can be approximated with error $\varepsilon/N$ by the truncated sum in \autoref{eq:grav-series} of order $p =  O(\operatorname{polylog} N/\varepsilon)$. Accordingly, the admissible block $\mathbf{K}^{\sigma, \rho}$ can be approximated to error $\varepsilon$ by a rank-$p$ matrix.

% \begin{remark}
% In the above Taylor expansion, we assume $K(s)$ is smooth in the neighborhood of s = 1 and use the fact $K(st) = K(s)K(t)$. For other kernels, \textit{e.g.,} $K(s) = \log s$,
% we can replace by $K(st) = K(s) + K(t)$ and obtain a similar expansion.
% \end{remark}

\subsection{Classical fast $\mathcal{H}$-matrix-vector multiplication}
We count the number of operations needed to compute the matrix-vector multiplication, where $\mathbf{K} \in \mathbb{R}^{N \times N}$ and $\mathbf{v} \in \mathbb{R}^{N}$,
\begin{equation}
    \mathbf{K}\mathbf{v} = \left(\sum_{\ell=2}^{L} \mathbf{K}^{(\ell)} + \mathbf{K}_{ad} \right) \mathbf{v} = \sum_{\ell=2}^{L} \mathbf{u}_{\ell} + \mathbf{u}_{ad}.
\end{equation}
The adjacent term $\mathbf{K}_{ad}$ is $3$-sparse, hence requires only $O(N)$ operations. In level $\ell$, each block-row (row of blocks) requires computing 3 matrix-vector multiplications on dense matrices of size $N\cdot 2^{-\ell}$, which takes $O(pN2^{-\ell})$ operations (where $p$ is the rank of the admissible blocks). There are $2^{\ell}$ block-rows, thus computing each $\mathbf{K}^{(\ell)}$ requires $O(pN)$ operations. Since there are $L-1=O(\log N)$ levels, the total operation complexity to compute all $\mathbf{u}_{\ell}$ is $O(pN\log N)$. Adding up $\mathbf{u}_{\ell}$ and $\mathbf{u}_{ad}$ costs another $O(N \log N)$. As we have shown above, it is required that $p=O(\operatorname{polylog}\frac{N}{\varepsilon})$ to achieve an error bound of $\varepsilon$ on $\mathbf{K}$. Thus, the overall runtime of matrix-vector multiplication on $\mathcal{H}$-matrices is $O(N \operatorname{polylog}\frac{N}{\varepsilon})$. We refer to \cite{part1hmatrices} for operation complexities of other tasks on $\mathcal{H}$-matrices, such as matrix-matrix multiplication and matrix inversion, which are also $O(N \operatorname{polylog}\frac{N}{\varepsilon})$ when $p=O(\operatorname{polylog}\frac{N}{\varepsilon})$.

\subsection{High dimensional hierarchical splittings}\label{app:high-dim}
In the main text, we have seen that two key properties of the 1D hierarchical splitting that enables optimal quantum block-encodings were the logarithmic number of hierarchical levels and the block sparseness of each level (the third key property was that the off-diagonal entries were decaying or sufficiently smooth). Here, we show that these properties still hold for the 2D case and make a generalization argument for higher dimensional cases.

The 2D and 3D hierarchical splittings are well-studied in classical literature \cite{part2hmatrices}. We will follow the indexing rules and conventions in the original work \cite{part2hmatrices}. For an overview of hierarchical matrices and splittings, we suggest interested readers to read \cite{fmm_shortcourse}--which covers the fast multipole algorithm (a special case of hierarchical splittings), \cite{part1hmatrices}--which is the original paper that introduced the 1D hieararchical splitting, and \cite{part2hmatrices}--the follow-up paper that generalized the hierarchical splittings to 2D and 3D problems.

Consider the uniform 2D grid of size $[1,N]\times [1,N]$ and $N=2^L$, where unit-mass particles are placed at sites $(i,j)$ (hence there are $N^2$ particles). We define the hierarchical partitions for this grid in the same spirit as \autoref{eq:splittings}. For level $\ell<L$, define the partitioning to be
\begin{equation}
    \mathcal{P}^{(\ell)}=\{\sigma^{(\ell)}_{I,J}|I,J \in \{0,\hdots,2^{\ell}-1\}]\},
\end{equation}
where the clusters $\sigma^{(\ell)}_{I,J}$ are
\begin{equation}
    \sigma^{(\ell)}_{I,J}=\{(i,j) |i \in [2^{L-\ell}I, 2^{L-\ell}(I+1) ), j \in [2^{L-\ell} J, 2^{L-\ell} (J+1) ) \}.
\end{equation}

Before we proceed to construct the hierarchical splitting of the kernel matrix, we first have to decide how to number particles in a 2D grid. We follow the ``canonical'' choice in Section 4 of \cite{part2hmatrices} (see \Cref{fig:2Dnumber}). In this numbering rule, the conversion from a grid site $(i,j)$ to the corresponding vectorized index $m$ is given by
\begin{equation}
    (i_1\hdots i_L, j_1\hdots j_L) \rightarrow m(i,j)=\sum_{\ell=1}^{L} f(i_{\ell},j_{\ell}) 4^{L-\ell},
\end{equation}
where $i_1\hdots i_L$ ($j_1 \hdots j_L$) is the binary representation of the row (column) index which ranges from $0$ to $2^L-1$ and $f$ is the CNOT function whose target bit is the second: $f(0,0)=0$, $f(0,1)=1$, $f(1,0)=3$, $f(1,1)=2$. The conversion from the vectorized index $K$ to grid site $(i,j)$ can be done by first writing $m$ in base-$4$ representation, then inverting the function $f$. In summary, the entries of the kernel matrix are of the form
\begin{equation}
    \mathbf{K}_{m(i_1,j_1),m(i_2,j_2)}=k((i_1,j_1),(i_2,j_2)).
\end{equation}

\begin{figure}
    \centering
    \includegraphics[width=0.75\textwidth]{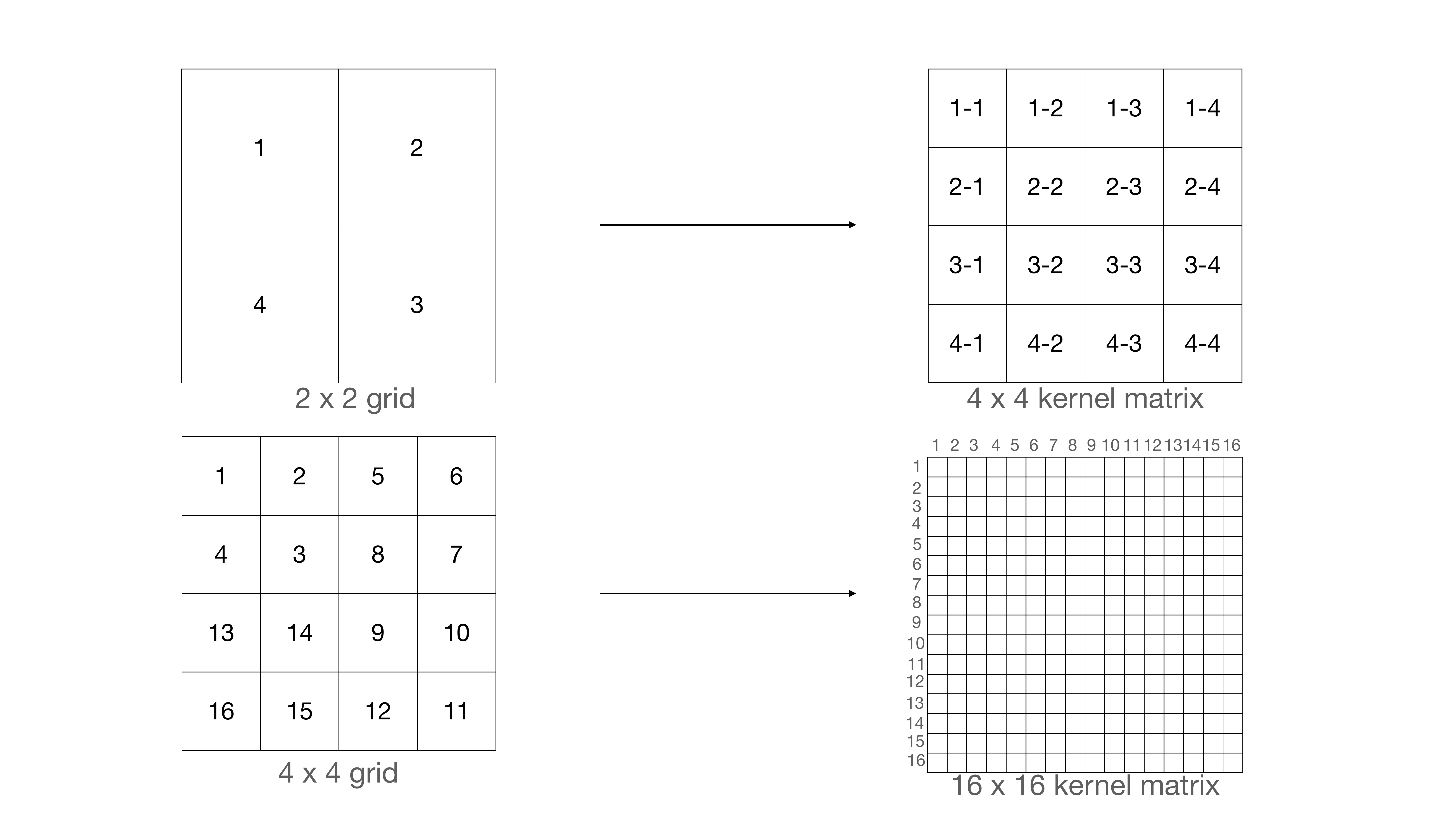}
    \caption{Left: 2D grid site numbering rule for vectorizing input. Right: the associated kernel matrix representing pairwise interactions based on this numbering system.}
    \label{fig:2Dnumber}
\end{figure}

Next, we also need to define the admissibility criteria (\Cref{def:admissible}) for now 2D clusters. For two clusters in the same level, \textit{e.g.}, $\sigma^{(\ell)}_{I,J}$ and $\sigma^{(\ell)}_{I',J'}$, wee define their diameters to be the length of their diagonal, which are easily seen to be $\operatorname{diam}(\sigma^{(\ell)}_{I,J})= \sqrt{2}\cdot 2^{L-\ell}$. Define the distance between these two clusters to be $\operatorname{dist}(\sigma^{(\ell)}_{I,J}, \sigma^{(\ell)}_{I',J'})=2^{L-\ell}\cdot \sqrt{\max\{0,|I-I'|-1\}^2+ \max \{0,|J-J'|-1 \}^2}$. We say the interactions between these two clusters are admissible if and only if
\begin{equation}
    \operatorname{dist}(\sigma^{(\ell)}_{I,J}, \sigma^{(\ell)}_{I',J'})\geq \eta \max \{ \operatorname{diam}(\sigma^{(\ell)}_{I,J}) \operatorname{diam}(\sigma^{(\ell)}_{I',J'}) \},
    \label{eq:2Dadmiss}
\end{equation}
where we choose $\eta=\frac{1}{\sqrt{2}}$.

In other words, two clusters (hence the block in the kernel matrix $\mathbf{K}$ that contains the interactions between them) are admissible when they are not neighboring along a column, row, or diagonal.

Having established the admissibility criteria, we proceed similarly to \Cref{sec:hmatrices} to obtain the following decomposition of the kernel matrix $\mathbf{K}$
\begin{equation}
    \mathbf{K} = \sum_{\ell=2}^{L} \mathbf{K}^{(\ell)}+ \mathbf{K}_{ad},
\end{equation}
where $\mathbf{K}^{(\ell)}$ contains the admissible interactions in level $\ell$ that have \emph{not} been included in $\mathbf{K}^{(\ell-1)}$.

Next, we calculate the block-sparsity of $\mathbf{K}^{(\ell)}$ by counting, among the $4^{\ell}$ clusters in level $\ell$, how many clusters contribute to the potential at the particles inside a particular cluster $\sigma^{(\ell)}_{I,J}$. In fact, at most $6^2-3^2=27$ level-$\ell$ clusters need to be included. This can be most easily observed in \Cref{fig:interaction}. Thus, the row-block-sparsity of $\mathbf{K}^{(\ell)}$ is 27. And since $\mathbf{K}^{(\ell)}$ is symmetric, the column-block-sparsity is also 27.

% https://arxiv.org/abs/1902.01829
% https://www.google.com/search?q=2+dimensional+hierarchical+matrix&tbm=isch&ved=2ahUKEwiP25f7y631AhXPr3IEHeAEDf0Q2-cCegQIABAA&oq=2+dimensional+hierarchical+matrix&gs_lcp=CgNpbWcQAzoHCCMQ7wMQJ1DsBljLDmDNEGgAcAB4AIABdogB5QOSAQM2LjGYAQCgAQGqAQtnd3Mtd2l6LWltZ8ABAQ&sclient=img&ei=p4LfYc_rNM_fytMP4Im06A8&bih=916&biw=1552#imgrc=qRWhJmSOHLtxBM&imgdii=6G3rETmN5sndOM

% https://www.sciencedirect.com/science/article/pii/S0167947319300374

% https://hal.archives-ouvertes.fr/hal-02864782/document

\begin{figure}
    \centering
    \includegraphics[width=.75\textwidth]{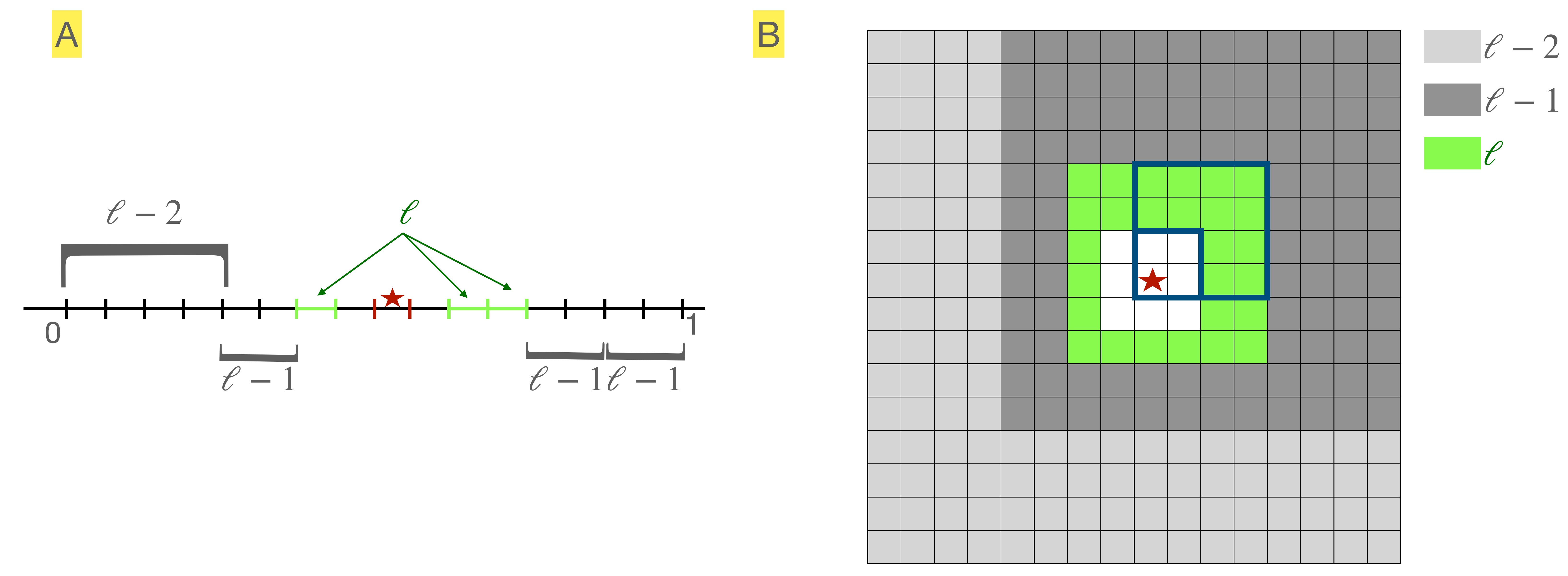}
    \caption{(A) shows the structure of the clusters in level $\ell$ (here $\ell=4$) of a 1D hierarchical splitting. Since most clusters have been taken into account in previous levels of the hierarchy, at most $6-3=3$ of them (green) will contribute to the potential at the target (red star). Similarly, (B) shows the structure of the clusters in level $\ell$ ($\ell=4$) of a 2D hierarchical splitting. Since most clusters have been included in previous levels of the hierarchy, at most $6^2-3^2=27$ of them (green) will contribute to the potential at the target.}
    \label{fig:interaction}
\end{figure}

In summary, we still have $\log N$ hierarchical levels, each of which is block-sparse in the 2D case. Therefore, we can apply the entire block-encoding procedure in \Cref{sec:main} to block-encode kernel matrices of a 2D grid. Below, we provide a similar analysis to that of \Cref{sec:polykernel} for polynomially decaying kernels ($k(\mathbf{x},\mathbf{x}')=\|
\mathbf{x}-\mathbf{x}'\|^{-p}$), showing that the hierarchical block-encoding is optimal for the 2D case.

At level $\ell$ ($2\leq \ell \leq L$), an admissible matrix block $\mathbf{K}^{\sigma,\rho}$, which contains the interactions between two clusters $\sigma, \rho$ at level $\ell$, has size $4^{L-\ell} \times 4^{L-\ell}$; and the minimum pairwise distance between particles  within this block is $2^{L-\ell}$ by the definition of admissibility (\autoref{eq:2Dadmiss}). It follows that the maximum entry in $\mathbf{K}^{\sigma,\rho}$ is $2^{-(L-\ell)p}$. Thus, $\mathbf{K}^{\sigma,\rho}$ can be block-encoded by \Cref{lem:densenaive} with a normalization factor of $\alpha_{\ell}=2^{(L-\ell)(2-p)}$. Then, $\mathbf{K}^{(\ell)}$ can be block-encoded by \Cref{lem:block-sparse} with a normalization factor of $27\alpha_{\ell}$. The adjacent interaction part, $\mathbf{K}_{ad}$, is a sparse matrix with sparsity $9$ and the largest entry equal to $1$. Thus, $\mathbf{K}_{ad}$ can be block-encoded with a normalization factor of $9$ using \Cref{lem:sparse}. Finally, the entire kernel matrix $\mathbf{K}$ is obtained by summing over all levels using \Cref{lem:linear_comb}. The overall normalization factor is 
\begin{equation}
    \alpha = 9 + \sum_{\ell=2}^{L} 27\alpha_{\ell}= \begin{cases} 9 +27\frac{2^{2-p}-N^{2-p}}{2^{2-p}-4^{2-p}} \leq O(1) & \text{if } p>2\\
    9+27(\log N-1) & \text{if } p=2 \\
    9 +27\frac{N^{2-p}-2^{2-p}}{4^{2-p}-2^{2-p}} & \text{if } p<2
    \end{cases}.
\end{equation}

This block-encoding can be shown to be optimal by a similar method to \Cref{remark:opt}. We have that $\|\mathbf{K}\| \geq \|\mathbf{K}\ket{\mathbf{1}}\|$, where $\ket{\mathbf{1}}$ is the normalized all-ones vector. Furthermore, observe that
\begin{equation}
\begin{aligned}
 \|\mathbf{K}\ket{\mathbf{1}}\| &\geq \int_{1}^{N}\int_{1}^{N} (x^2+y^2)^{-p/2}dxdy\\
 &\geq \int_{0}^{\pi/2} \int_{\sqrt{2}}^{\sqrt{N^2+1}} r^{-p} r dr d\theta - 2 \int_{1}^{\sqrt{N^2+1}} (x^2+1)^{-p/2} dx.\\
\end{aligned}
\end{equation}
% as illustrated in the figure below.
% \begin{center}
% \includegraphics[width=0.25\textwidth]{2D_integral.pdf}
% \end{center}
In the limit $N \rightarrow \infty$, it can be easily verified that the above integral is $O(N^{2-p})$ when $p\neq 2$ and $O(\log N)$ when $p=2$. Therefore, $\alpha=\Theta(\|\mathbf{K}\|)$ and the block-encoding procedure using hierarchical splitting is optimal.

Finally, the analysis in this section can be generalized to any $d$-dimensional setting, where the block-sparsity of each hierarchical level $\mathbf{K}^{(\ell)}$ is $6^d-3^d$ and the sparsity of the adjacent part $\mathbf{K}_{ad}$ is $3^d$. For a visualization of 2D and 3D hierarchical splittings, see Figure 6 of \cite{boukaram2019hierarchical}, which has a smaller block-sparsity than the aforementioned value since diagonally neighboring clusters are considered admissible by the authors of \cite{boukaram2019hierarchical}, as opposed to our treatment in this section.

\section{Quantum state preparation for smooth functions}\label{app:stateprep}
We provide an efficient procedure to prepare a quantum state whose entries are sampled from a smooth function. Consider a smooth function $g$ bounded by $|g(x)|\leq 1$ in the domain $\Omega= [0,2\pi]$ which can be approximated by a low-order Fourier series
\begin{equation}
\begin{aligned}
    g(x_j) &= \sum_{n=-p}^{p} c_n e^{2\pi inj/N},
    %  = \sum_{n=0}^{N-1} c_n e^{2\pi inj/N},
\end{aligned}
\end{equation}
where $p=O(1)$ and the coefficients $c_n$ are $\Theta(1)$. 

We want to prepare the quantum state $\ket{\mathbf{g}}$ where the vector $\mathbf{g}$ is
\begin{equation}
    \mathbf{g}_j = g(2\pi j/N)
\end{equation}
Let the QFT matrix be
\begin{equation}
    F_N = \sum_{k,j=0}^{N-1} \frac{1}{\sqrt{N}} e^{-2\pi ijk/N} \ket{k}\bra{j}.
\end{equation}
Let $\widetilde{\mathbf{g}} = \left[\begin{array}{c}
     c_0 \\
     \vdots \\
     c_{N-1}
\end{array}\right]$, where at most $d=2p+1$ terms $c_n$ are non-zero, and $c_{N-k}=c_{-k}$ for $1\leq k \leq p$.

Observe that 
\begin{equation}
    \mathbf{g} = \frac{F_N}{\sqrt{N}} \widetilde{\mathbf{g}}.
\end{equation}

Since $\widetilde{\mathbf{g}}$ is sparse and the values and locations of the non-zero entries are known, one can easily prepare its quantum version $\ket{\widetilde{\mathbf{g}}}$ with high probability using the following procedure.
\begin{enumerate}
    \item Prepare $\ket{0} \rightarrow \frac{1}{\sqrt{d}} \sum_{k=0}^{d-1} \ket{k} \rightarrow \frac{1}{\sqrt{d}} \sum_{j:c_j\neq 0} \ket{j}$ 
    %(where $l_j$ are the locations of the non-zero Fourier coefficients)
    \item Call the oracle $\mathcal{O}_{c}: \ket{j}\ket{0}  \rightarrow  \ket{j}\ket{\Tilde{c}_j}$ (where $\Tilde{c}_j$ is a qubit string description of $c_j$)
    \item Apply the conditioned rotation $\ket{j}\ket{\Tilde{c}_j}\ket{0} \rightarrow \frac{c_j}{\max_j c_j}\ket{j}\ket{\Tilde{c}_j}\ket{0} + \sqrt{1-\left|\frac{c_j}{\max_j c_j}\right|^2} \ket{j}\ket{\Tilde{c}_j}\ket{1}$
    \item Uncompute and post-select on the subspace $\ket{0}$ of the third register to obtain $\ket{\widetilde{\mathbf{g}}}$
\end{enumerate}
Upon successfully obtaining $\ket{\widetilde{\mathbf{g}}}$ (which happens with probability $\Omega(1)$), one applies the QFT to obtain $\ket{\mathbf{g}}= F_N \ket{\widetilde{\mathbf{g}}}$. This procedure clearly applies for higher dimensional smooth functions. In these cases, one uses a tensor product of the conventional QFTs to transform the state back and forth between the real regime and the Fourier regime.

\end{document}